\documentclass{amsart}[11pt]
\usepackage{amsmath, amsfonts, amsthm, amssymb,mathtools}
\usepackage{subfigure}
\usepackage{stmaryrd}
\usepackage{verbatim}
\usepackage{hyperref}
\usepackage{color}
\usepackage{ulem}
\usepackage[dvips]{graphicx}
\usepackage{enumerate}
\numberwithin{equation}{section}
\newtheorem{theorem}{Theorem}[section]
\newtheorem{corollary}[theorem]{Corollary}
\newtheorem{lemma}[theorem]{Lemma}
\newtheorem{proposition}[theorem]{Proposition}

\theoremstyle{definition}

\newtheorem{remark}[theorem]{Remark}
\newtheorem{assumption}[theorem]{Assumption}

\newcommand{\ind}{1\hspace{-2.1mm}{1}} 
\newcommand{\I}{\mathtt{i}}
\newcommand{\RR}{\mathbb{R}}

\newcommand{\Ddh}{\mathcal{K}_{\mathrm{H}}}
\newcommand{\PP}{\mathbb{P}}
\newcommand{\D}{\mathrm{d}}

\newcommand{\E}{\mathrm{e}}
\newcommand{\LO}{\Lambda_{0}}

\newcommand{\sgn}{\mathrm{sgn}}
\newcommand{\Nn}{\mathcal{N}}

\begin{document}

\title{Asymptotics of forward implied volatility}
\author{Antoine Jacquier and Patrick Roome}
\address{Department of Mathematics, Imperial College London}
\email{a.jacquier@imperial.ac.uk, p.roome11@imperial.ac.uk}
\thanks{The authors would like to thank all the participants of the 2013 AMaMeF Conference
and Peter Tankov for useful discussions.
They are also indebted to the anonymous referees for helpful comments.}
\date{\today}
\maketitle
\begin{abstract}
We prove here a general closed-form expansion formula for forward-start options and the forward implied volatility smile in a large class of models, including the Heston stochastic volatility and time-changed exponential L\'evy models.
This expansion applies to both small and large maturities and is based solely on the knowledge
of the forward characteristic function of the underlying process.	
The method is based on sharp large deviations techniques, and allows us to recover (in particular) many results 
for the spot implied volatility smile.
In passing we show 
(i) that the small-maturity exploding behaviour of forward smiles depends on whether the quadratic variation of the underlying is bounded or not, 
and (ii) that the forward-start date also has to be rescaled in order to obtain non-trivial small-maturity asymptotics.

\end{abstract}

\section{Introduction}

Consider an asset price process $\left(\E^{X_t}\right)_{t\geq0}$ with $X_0=0$, paying no dividend, 
defined on a complete filtered probability space $(\Omega,\mathcal{F},(\mathcal{F}_t)_{t\geq0},\mathbb{P})$ 
with a given risk-neutral measure $\mathbb{P}$, and assume that interest rates are zero. 
In the Black-Scholes-Merton (BSM) model, the dynamics of the logarithm of the asset price are given by 
\begin{align}\label{eq:BSDynamics}
\D X_t=-\frac{1}{2}\Sigma^2  \D t +\Sigma \D W_t,
\end{align}
where $\Sigma>0$ is the instantaneous volatility and $W$ a standard Brownian motion.
The no-arbitrage price of the call option at time zero is then given by the famous BSM formula~\cite{BS73,M73}:
$C_{\textrm{BS}}(\tau,k,\Sigma) := \mathbb{E}\left(\E^{X_\tau}-\E^k\right)_+
=\Nn\left(d_+\right)-\E^k\Nn\left(d_-\right)$,
with $d_\pm:=-\frac{k}{\Sigma\sqrt{\tau}}\pm\frac{1}{2}\Sigma\sqrt{\tau}$,
where $\Nn$ is the Gaussian distribution function. 
For a given market price $C^{\textrm{obs}}(\tau,k)$ of the option at strike~$\E^k$ and maturity~$\tau$ we define the spot implied volatility~$\sigma_{\tau}(k)$ as the unique solution to 
the equation $C^{\textrm{obs}}(\tau,k)=C_{\textrm{BS}}(\tau,k,\sigma_{\tau}(k))$.

For any~$t,\tau>0$ and $k\in\mathbb{R}$, we define~\cite{B04,VL}
a Type-I forward-start option with forward-start date~$t$, maturity~$\tau$ and strike~$\E^k$ as a European option with payoff 
$\left(\E^{X_{\tau}^{(t)}}-\E^k\right)^+$ where $X_{\tau}^{(t)}:=X_{t+\tau}-X_t$ pathwise. 
In the BSM model~\eqref{eq:BSDynamics} its value is simply worth $C_{\textrm{BS}}(\tau,k,\Sigma)$. 
For a given market price $C^{\textrm{obs}}(t,\tau,k)$ of the option at strike~$\E^k$, forward-start date~$t$ and maturity~$\tau$ 
we can define the forward implied volatility smile $\sigma_{t,\tau}(k)$ as the unique solution to 
$C^{\textrm{obs}}(t,\tau,k)=C_{\textrm{BS}}(\tau,k,\sigma_{t,\tau}(k))$. 
A second type of forward-start option exists~\cite{VL} and corresponds to a European option with payoff 
$\left(\E^{X_{t+\tau}}-\E^{k+X_t}\right)^+$. 
In the BSM model~\eqref{eq:BSDynamics} the value of the Type-II forward-start option is worth $C_{\textrm{BS}}(\tau,k,\Sigma)$ 
~\cite{R91}. 
Again, for a given market price $C^{\textrm{obs,II}}(\tau,t,k)$ of such an option, we define the Type-II forward implied volatility smile 
$\widetilde{\sigma}_{t,\tau}(k)$ as the unique solution to $C^{\textrm{obs,II}}(\tau,t,k)=C_{\textrm{BS}}(\tau,k,\widetilde{\sigma}_{t,\tau}(k))$. 
Both definitions of the forward smile are generalisations of the spot implied volatility smile 
since they reduce to the spot smile when $t=0$. 

The literature on implied volatility asymptotics is extensive and has been using a diverse range of mathematical techniques.
In particular, small-maturity asymptotics have historically received wide attention due to earlier results from the 1980s on expansions of the heat kernel~\cite{Benarous}.
PDE methods for continuous-time diffusions~\cite{BBF,Hagan,Pascucci}, large deviations~\cite{DFJV,FJSmall}, saddlepoint methods~\cite{FJL11}, Malliavin calculus~\cite{Gobet,KT04} and differential geometry~\cite{GHLOW,PHL} are among the main methods used to tackle the small-maturity case.
Extreme strike asymptotics arose with the seminal paper by Roger Lee~\cite{Lee} and have been further extended by Benaim and Friz~\cite{BF1, BF2} and in~\cite{Guli, GS, FGGS, DFJV}.
Comparatively, large-maturity asymptotics have only been studied in~\cite{Tehr, FJ09, JMKR, JM12, FJM10} using large deviations and saddlepoint methods.
Fouque et al.~\cite{Fouque} have also successfully introduced perturbation techniques in order to study slow and fast mean-reverting stochastic volatility models.
Models with jumps (including L\'evy processes), studied in the above references for large maturities and extreme strikes, `explode' in small time, in a precise sense investigated in~\cite{Alos, AndLipton, Tank, Nutz, MijTankov, FL}.

A collection of implied volatility smiles over a time horizon $(0,T]$ is also known to be equivalent to the marginal distributions of the asset price process over $(0,T]$. 
Implied volatility asymptotics have therefore provided a set of tools to analytically understand 
the marginal distributions of a model and their relationships to market observable quantities such as volatility smiles. 
However many models can calibrate to implied volatility smiles (static information) with the same degree of precision and produce radically different prices and risk sensitivities for exotic securities. 
This can usually be traced back to a complex and often non-transparent dependence on transitional probabilities 
or equivalently on model-generated dynamics of the smile. 
The dynamics of the smile is therefore a key model risk associated with these products and any model used for pricing and risk management should produce realistic dynamics that are in line with trader expectations and historical dynamics. One metric that can be used to understand the dynamics of implied volatility smiles (\cite{B04} calls it a 'global measure' of the dynamics of implied volatilities) is to use the forward smile defined above. The forward smile is also a market-defined quantity and naturally extends the notion of the spot implied volatility smile.
Forward-start options also serve as natural hedging instruments for several exotic securities 
(such as Cliquets, Ratchets and Napoleons; see~\cite[Chapter 10]{G06}) and are therefore worth investigating.

Despite significant research on implied volatility asymptotics, 
there are virtually no results on the asymptotics of the forward smile:
Glasserman and Wu~\cite{GW10} introduced different notions of forward volatilities to assess their predictive values in determining future option prices and future implied volatility, 
Keller-Ressel~\cite{KR} studies a very specific type of asymptotic (when the forward-start date becomes large), and empirical results have been carried out by practitioners in \cite{B04,B02,G06}. 
Recently, in~\cite{JR2013} the authors proved that for fixed $t>0$ the Heston forward smile 
(corresponding to~$X_{\tau}^{(t)}$) explodes (except at-the-money) as $\tau$ tends to zero.

We consider here a continuous-time stochastic process $(Y_\varepsilon)_{\varepsilon>0}$ and 
prove---under some assumptions on its characteristic function---an expansion 
for European option prices on $Y_\varepsilon$ of the form 
$$
\mathbb{E}\left(\E^{Y_{\varepsilon}f(\varepsilon)}-\E^{kf(\varepsilon)}\right)^+
 = 
\mathcal{I}(k,c,\varepsilon)+
\alpha(k,c) \E^{-\Lambda^*(k)/\varepsilon+k f(\varepsilon)}
\left(c \sqrt{\varepsilon }\ind_{\{c>0\}}+\varepsilon^{3/2} f(\varepsilon )\ind_{\{c=0\}}\right)
\left[1+\alpha_1(k,c)\varepsilon + \mathcal{O}\left(\varepsilon^2\right)\right],
$$
as $\varepsilon\downarrow 0$, 
for some (explicit) functions $\alpha, \alpha_1$ and a residue term $\mathcal{I}$
(Theorem~\ref{theorem:GeneralOptionAsymp} and Corollary~\ref{cor:CallPrice}).
Here $f$ is a continuous function satisfying $\varepsilon f(\varepsilon)=c+\mathcal{O}(\varepsilon)$ as $\varepsilon\downarrow 0$, 
and $\Lambda^*$ can be interpreted as a large deviations rate function.
Setting $Y_{\varepsilon}\equiv X_{\varepsilon\tau}^{(\varepsilon t)}$ and $f(\varepsilon)\equiv1$ 
or $Y_\varepsilon\equiv \varepsilon X_{\tau/\varepsilon}^{(t)}$ and $f(\varepsilon)\equiv\varepsilon^{-1}$ yields `diagonal' small-maturity (Corollary~\ref{cor:ShortTimeAsymp}) and large-maturity (Corollary~\ref{Cor:BSOptionLargeTime}) expansions of forward-start option prices.
The diagonal small-maturity re-scaling results in non-degenerate small-maturity asymptotics that are far more accurate than the small-maturity asymptotic in~\cite{JR2013}. 
This result also applies when $t=0$, and generalises the results in~\cite{FJM10},~\cite{FJ09},~\cite{JMKR}.
We also translate these results into closed-form asymptotic expansions for the forward implied volatility smile (Type I and Type II).
In Section~\ref{sec:Examples}, we provide explicit examples for the Heston and time-changed exponential L\'evy processes. 
Section~\ref{sec:Numerics} provides numerical evidence supporting the practical relevance of these results
and we leave the proofs of the main results to Section~\ref{sec:Proofs}.

\textbf{Notations: } $\Nn(\mu,\sigma^2)$ shall represent the Gaussian distribution with mean $\mu$ and variance $\sigma^2$.
Furthermore $\mathbb{E}$ and $\mathbb{V}$ shall always denote expectation and variance under a risk-neutral measure~$\PP$ given a priori.
We shall refer to the standard (as opposed to the forward) implied volatility as the spot smile and denote it $\sigma_\tau$.
The (Type-I) forward implied volatility will be denoted $\sigma_{t,\tau}$ as above.
In the remaining of this paper $\varepsilon$ will always denote a strictly positive (small) quantity,
and we let $\mathbb{R}^*:=\mathbb{R}\setminus\{0\}$ and $\mathbb{R}^*_+:=(0,\infty)$.
For two functions $g,h:\mathbb{R}_+\to\mathbb{R}_+$ we use the notation $g\sim h$ to mean $\lim_{\varepsilon \downarrow 0}g(\varepsilon)/h(\varepsilon)=1$
and we let $\sgn(p)=1$ if $p\geq 0$ and $-1$ otherwise.
For a sequence of sets $(\mathcal{D}_{ \varepsilon})_{\varepsilon>0}$ in $\mathbb{R}$, 
we may, for convenience, use the notation $\lim_{\varepsilon \downarrow 0}\mathcal{D}_{ \varepsilon}$, 
by which we mean the following (whenever both sides are equal):
$
\liminf_{\varepsilon \downarrow 0}\mathcal{D}_{ \varepsilon}
:= \bigcup_{\varepsilon>0}\bigcap_{s\leq\varepsilon}\mathcal{D}_{s}
 = \bigcap_{\varepsilon>0}\bigcup_{s\leq\varepsilon}\mathcal{D}_{s}
=:\limsup_{\varepsilon\downarrow 0}\mathcal{D}_{\varepsilon}.
$
Finally, for a given set $A\subset\mathbb{R}$, we let $A^o$ denote its interior (in $\mathbb{R}$)
and $\Re(z)$ and $\Im(z)$ denote the real and imaginary parts of a complex number~$z$.

\section{General Results} \label{sec:GeneralResults}
This section gathers the main notations of the paper as well as the general results.
The main result is Theorem~\ref{theorem:GeneralOptionAsymp}, 
which provides an asymptotic expansion---up to virtually any arbitrary order---of option prices 
on a given process $(Y_{\varepsilon})$, as $\varepsilon$ tends to zero.
This general formulation allows us, by suitable scaling, 
to obtain both small-time (Section~\ref{sec:SmallDiag}) and large-time (Section~\ref{sec:LargeMatGeneral}) expansions.

\subsection{Notations and main theorem}
\subsubsection{Notations and preliminary results}
Let $(Y_{\varepsilon})$ be a stochastic process with re-normalised logarithmic moment generating function (lmgf)
\begin{align}\label{eq:Renorm-mgf}
\Lambda_{\varepsilon}(u):=
\varepsilon\log\mathbb{E}\left[\exp\left(\frac{uY_{\varepsilon}}{\varepsilon}\right)\right],
\qquad\text{for all }
u\in\mathcal{D}_{\varepsilon}:=\{u\in\mathbb{R}:|\Lambda_{\varepsilon}(u)|<\infty\}.
\end{align}
We further define $\mathcal{D}_{0}:=\lim_{\varepsilon\downarrow 0}\mathcal{D}_{\varepsilon}$
and now introduce the main assumptions of the paper.

\begin{assumption}  \label{assump:Differentiability}\ 
\begin{enumerate}[(i)]
\item \textbf{Expansion property:} 
For each $u\in\mathcal{D}_{0}^{o}$ the following Taylor expansion holds as $\varepsilon$ tends to zero:
\begin{equation}\label{eq:LambdaExpansion}
\Lambda_{\varepsilon}(u)=\sum_{i=0}^{2}\Lambda_{i}(u)\varepsilon^{i}+\mathcal{O}(\varepsilon^3);
\end{equation}
\item \textbf{Differentiability:} 
There exists $\varepsilon_0>0$ such that the map $(\varepsilon,u)\mapsto \Lambda_{\varepsilon}(u)$
is of class $\mathcal{C}^{\infty}$ on $(0,\varepsilon_0)\times\mathcal{D}_{0}^{o}$;
\item \textbf{Non-degenerate interior:} $0\in\mathcal{D}_{0}^o$;
\item \textbf{Essential smoothness: }$\LO$ is strictly convex and essentially smooth\footnote{\cite[Definition 2.3.5]{DZ93}. 
A convex function $h:\RR\supset\mathcal{D}_{h}\to(-\infty,\infty]$ is essentially smooth if
$\mathcal{D}_{h}^o$ is non-empty,
if~$h$ is differentiable in~$\mathcal{D}_{h}^o$,
and if~$h$ is steep, e.g. $\lim_{n\uparrow\infty}|h'(u_n)|=\infty$ for every sequence $(u_n)_{n\in\mathbb{N}}$ in $\mathcal{D}_{h}^o$ that converges to a boundary point of~$\mathcal{D}_{h}^o$.
}
on $\mathcal{D}_{0}^o$;
\item \textbf{Tail error control:}
For any fixed $p_{r}\in\mathcal{D}_{0}^o\backslash\{0\}$,
\begin{enumerate}[(a)]
\item
$\Re\left(\Lambda_{\varepsilon}\left(\I p_{i} +p_{r}\right)\right) = \Re\left(\LO\left(\I p_{i} +p_{r}\right)\right)+\mathcal{O}(\varepsilon)$, for any $p_i\in\RR$;
\item
the function $L:\RR\ni p_{i}\mapsto \Re\left(\LO\left(\I p_{i} +p_{r}\right)\right)$ 
has a unique maximum at zero and is bounded away from $L(0)$ as $|p_{i}|$ tends to infinity;
\item
there exist $\varepsilon_1,p^*_{i}>0$ such that for all $|p_{i}|\geq p^*_{i}$ and $\varepsilon \leq \varepsilon_1$ there exists
$M$ (independent of $p_i$ and~$\varepsilon$) such that
$\Re\left[\Lambda_{\varepsilon}(\I p_{i} +p_{r}) - \LO(\I p_{i} +p_{r})\right] \leq M \varepsilon$.
\end{enumerate}
\end{enumerate}
\end{assumption}
Assumption~\ref{assump:Differentiability}(i) implies that the functions 
$ \lim_{\varepsilon\downarrow 0}\partial_{\varepsilon}^i\Lambda_\varepsilon(u)$ exist on
$\mathcal{D}_{0}^{o}$ for $i=0,1,2$.
Assumption~\ref{assump:Differentiability}(ii) could be relaxed to
$\mathcal{C}^4((0,\varepsilon_0)\times\mathcal{D}_{0}^{o})$, 
but this hardly makes any difference in practice and does,
 however, render some formulations awkward.
If the expansion~\eqref{eq:LambdaExpansion} holds up to some higher order $n\geq 3$, one can in principle
show that both forward-start option prices and the forward implied volatility expansions below hold to order $n$ as well.
However expressions for the coefficients of higher order are extremely cumbersome and scarcely useful in practice.
Assumption~\ref{assump:Differentiability}(v) is a technical condition (readily satisfied by practical models) required to show that the dependence of option prices on the tails of the characteristic function of the asset price is exponentially small (see Appendix~\ref{sec: TailEstimates} and~\ref{append:tailverif} for further details). 
We do not require this condition to be satisfied at $p_r=0$ since this corresponds to an option strike at which our main result does not hold anyway ($k=\Lambda_{0,1}(0)$ in Theorem~\ref{theorem:GeneralOptionAsymp} below).
We note that this assumption is not required if one is only interested in the leading-order behaviour of option prices and forward implied volatility. 
Assumption~\ref{assump:Differentiability}(iv) is the key property that needs to be checked in practical computations and can be violated by well-known models under certain parameter configurations (see Section~\ref{sec:LargeMatHeston} for an example).

Define now the function $\Lambda^*:\mathbb{R}\to\mathbb{R}_+$ as the Fenchel-Legendre transform of $\LO$:
\begin{align}\label{eq:LambdaStarDefinition}
\Lambda^*(k):=\sup_{u\in\mathcal{D}_{0}}\{uk-\LO(u)\},
\qquad\text{for all }k\in\mathbb{R}.
\end{align}
For ease of exposition in the paper we will use the notation
\begin{align}\label{eq:LambdaDerivatives}
 \Lambda_{i,l}(u):=\partial^{l}_{u}\Lambda_{i}(u)\quad\text{for }l\geq1, i=0,1,2.
\end{align}
The following lemma gathers some immediate properties of the functions~$\Lambda^*$ and~$\Lambda_{il}$ 
which will be needed later.
\begin{lemma}\label{lem:Properties}
Under Assumption~\ref{assump:Differentiability}, the following properties hold:
\begin{enumerate}[(i)]
\item For any $k\in\mathbb{R}$, there exists a unique~$u^*(k)\in\mathcal{D}_{0}^o$ such that
\begin{align}
& \Lambda_{0,1}(u^*(k))=k,\label{eq:u*definition}\\
& \Lambda^*(k)=u^*(k)k-\LO\left(u^*(k)\right);\label{eq:LambdaStarRep2}
\end{align}
\item $\Lambda^*$ is strictly convex and differentiable on $\mathbb{R}$;
\item if $a\in\mathcal{D}_{0}^o$ such that $\LO(a)=0$, 
then $\Lambda^*(k)> ak$ for all  $k\in\mathbb{R}\setminus\{\Lambda_{0,1}(a)\}$ 
and $\Lambda^*(\Lambda_{0,1}(a))=a\Lambda_{0,1}(a)$.
\end{enumerate}
\end{lemma}

\begin{proof}\
\begin{enumerate}[(i)]
\item By Assumption~\ref{assump:Differentiability} $\Lambda_{0,1}$ is a strictly increasing differentiable function from $-\infty$ to $\infty$ on $\mathcal{D}_{0}$.
\item By (i), $\partial_k\Lambda^{*}(k)=\Lambda_{0,1}^{\leftarrow}(k)$ for all $k\in\mathbb{R}$. 
In particular $\partial_k\Lambda^{*}$ is strictly increasing on $\mathbb{R}$.
\item  Since $\Lambda_{0,1}$ is strictly increasing, $\Lambda_{0,1}(a)=k$ if and only if $u^*(k)=a$ 
and then $\Lambda^*(\Lambda_{0,1}(a))=a\Lambda_{0,1}(a)$ using~\eqref{eq:LambdaStarRep2}. Using the definition~\eqref{eq:LambdaStarDefinition} with $a\in\mathcal{D}_{0}^o$ and $\LO(a)=0$ gives  $\Lambda^*(k)\geq a k$. Since $\Lambda^*$ is strictly convex from (ii) it follows that  $\Lambda^*(k)> ak$ for all  $k\in\mathbb{R}\setminus\{\Lambda_{0,1}(a)\}$.
\end{enumerate}
\end{proof}
\begin{remark}
The saddlepoint~$u^*$ is not always available in closed-form, but can be computed via a simple root-finding algorithm.
Furthermore, a Taylor expansion around any point can be computed iteratively in terms of the derivatives of~$\LO$.
For instance, around $k=0$, we can write
$u^{*}(k)=u^*(0) + \frac{k}{\Lambda_{0,2}(u^*(0))} - \frac{1}{2}\frac{\Lambda_{0,3}(u^*(0))}{\Lambda_{0,2}(u^*(0))^3}k^2
+ \mathcal{O}(k^3)$.
A precise example can be found in the proof of Corollary~\ref{cor:DiagTaylorExpansion}.
\end{remark}

The last tool we need is a (continuous) function~$f:\mathbb{R}_+\to\mathbb{R}_+$ such that there exists $c\geq 0$ for which 
\begin{align}\label{eq:LDP Rescaling}
f(\varepsilon)\varepsilon=c+\mathcal{O}(\varepsilon),\quad
\text{as }\varepsilon\text{ tends to zero}. 
\end{align}
This function will play the role of rescaling the strike of European options and will give us the flexibility to deal with both small-and large-time behaviours.
Finally, for any $b\geq0$ we now define the functions
$A_b,\bar{A}_b:\mathbb{R}\setminus\{\Lambda_{0,1}(0), \Lambda_{0,1}(b)\}\times\mathbb{R}_{+}^*\to\mathbb{R}$ by
$$
\bar{A}_b(k,\varepsilon) 
 :=\displaystyle \frac{b \sqrt{\varepsilon }\ind_{\{b>0\}}+\varepsilon^{3/2} f(\varepsilon )\ind_{\{b=0\}} }
{u^*(k) \left(u^*(k)-b\right)\sqrt{2\pi\Lambda _{0,2}(u^*(k))}}
\quad\text{and}\quad
A_b(k,\varepsilon) := \displaystyle 1+\Upsilon(b,k)\varepsilon+\frac{u^*(k)(\varepsilon  f(\varepsilon )-b)}
{\left(u^*(k)-b\right)b}\ind_{\{b>0\}}+\frac{\varepsilon  f(\varepsilon )}{u^*(k)}\ind_{\{b=0\}},
$$
where
$\Upsilon:\mathbb{R}_{+}\times\mathbb{R}\backslash\{\Lambda_{0,1}(0), \Lambda_{0,1}(b)\}\to\mathbb{R}$ is given by
\begin{align}\label{eq:Upsilon}
\Upsilon(b,k)
 &:= \Lambda_{2}-\frac{5 \Lambda _{0,3}^2}{24 \Lambda _{0,2}^3}+
\frac{4\Lambda _{1,1}\Lambda _{0,3}+\Lambda_{0,4}}{8\Lambda _{0,2}^2}
-\frac{\Lambda _{1,1}^2+\Lambda _{1,2}}{2 \Lambda _{0,2}}
-\frac{\Lambda _{0,3}}{2u^*(k) \Lambda _{0,2}^2}-\frac{\Lambda _{0,3}}{2\left(u^*(k)-b\right)\Lambda _{0,2}^2} \\ \nonumber
& -\frac{\Lambda _{1,1} \left(b-2u^*(k)\right)+3}{u^*(k) \left(u^*(k)-b\right)\Lambda _{0,2}}
-\frac{b^2}{u^*(k)^2\left(u^*(k)-b\right)^2 \Lambda _{0,2}}.
\end{align}
For ease of notation we write $\Lambda_{i}$ and $\Lambda_{i,l}$ in place of $\Lambda_{i}\left(u^*(k)\right)$ and $\Lambda_{i,l}\left(u^*(k)\right)$. The domains of definition of $A_b$ and $\bar{A}_b$ excludes the set $\{\Lambda_{0,1}(0), \Lambda_{0,1}(b)\}=\{k\in\mathbb{R}:u^*(k)\in\{0,b\}\}$. 
For all $k$ in this domain, $\Lambda_{0,2}(u^*(k))>0$ by Assumption~\ref{assump:Differentiability}(iv),
so that $A_b$  and $\bar{A}_b$ are both well-defined real-valued functions. 

\subsubsection{Main theorem and corollaries}
The following  theorem on asymptotics of option prices is the main result of the paper. 
A quick glimpse at the proof of Theorem~\ref{theorem:GeneralOptionAsymp} in Section~\ref{sec:proofMainProp}
shows that this result can be extended to any arbitrary order. \\

\begin{theorem}~\label{theorem:GeneralOptionAsymp}
Let~$(Y_\varepsilon)_{\varepsilon>0}$ satisfy Assumptions~\ref{assump:Differentiability},
and~$f:\mathbb{R}_+\to\mathbb{R}_+$ be a function satisfying~\eqref{eq:LDP Rescaling} with constant $c\in\mathcal{D}_{0}^o\cap\mathbb{R}_{+}$. 
Then the following expansion holds for all $k\in\mathbb{R}\backslash\{\Lambda_{0,1}(0), \Lambda_{0,1}(c)\}$ as $\varepsilon\downarrow 0$:
\begin{equation*}
\E^{-\Lambda^*(k)/\varepsilon+k f(\varepsilon)+\Lambda_{1}}
\bar{A}_c(k,\varepsilon)\left[A_c(k,\varepsilon) + \mathcal{O}\left(\varepsilon^2\right)\right]
 = \left\{
\begin{array}{ll}
\mathbb{E}\left(\E^{Y_{\varepsilon}f(\varepsilon)}-\E^{kf(\varepsilon)}\right)^+,
\quad & \text{if }k>\Lambda_{0,1}(c),\\
\mathbb{E}\left(\E^{kf(\varepsilon)}-\E^{Y_{\varepsilon}f(\varepsilon)}\right)^+,
\quad & \text{if }k<\Lambda_{0,1}(0),\\
 - \mathbb{E}\left(\E^{Y_{\varepsilon}f(\varepsilon)}\wedge \E^{kf(\varepsilon)}\right),
\quad & \text{if }\Lambda_{0,1}(0)<k<\Lambda_{0,1}(c).
\end{array}
\right.
\end{equation*}
\end{theorem}
Using Put-Call parity, the theorem can also be read as an expansion for European Call options (or for Put options) for all strikes, except at the two points $\Lambda_{0,1}(0)$ and $\Lambda_{0,1}(c)$:
\begin{corollary}\label{cor:CallPrice}
Under the assumptions of Theorem~\ref{theorem:GeneralOptionAsymp}, we have, 
for $k\in\mathbb{R}\backslash\{\Lambda_{0,1}(0), \Lambda_{0,1}(c)\}$, as $\varepsilon\downarrow 0$:
$$
\mathbb{E}\left(\E^{Y_{\varepsilon}f(\varepsilon)}-\E^{kf(\varepsilon)}\right)^+
=\E^{\Lambda_{\varepsilon}(f(\varepsilon)\varepsilon)/\varepsilon}\ind_{\{k<\Lambda_{0,1}(c)\}}
-\E^{k f(\varepsilon)}\ind_{\{k<\Lambda_{0,1}(0)\}}
+\E^{-\Lambda^*(k)/\varepsilon+k f(\varepsilon)+\Lambda_{1}}
\bar{A}_c(k,\varepsilon)\left[A_c(k,\varepsilon) + \mathcal{O}\left(\varepsilon^2\right)\right].
$$
\end{corollary}

\subsection{Forward-start option asymptotics}
We now specialise Theorem~\ref{theorem:GeneralOptionAsymp} to forward-start option asymptotics.
For a stochastic process~$(X_t)_{t\geq 0}$, we define (pathwise), for any~$t\geq 0$, the process $(X_{\tau}^{(t)})_{\tau\geq 0}$ by
\begin{equation}\label{eq:XtTauDef}
X_{\tau}^{(t)}:=X_{t+\tau}-X_{t}.
\end{equation}
\subsubsection{Diagonal small-maturity asymptotics}\label{sec:SmallDiag}
We first consider asymptotics when both $t$ and $\tau$ are small, 
which we term \textit{diagonal small-maturity asymptotics}.
Set $(Y_{\varepsilon}):=(X^{(\varepsilon t)}_{\varepsilon\tau})$ and $f\equiv 1$. 
Then~$c=0$ and the following corollary follows from Theorem~\ref{theorem:GeneralOptionAsymp}:
\begin{corollary}\label{cor:ShortTimeAsymp}
If $(X^{(\varepsilon t)}_{\varepsilon\tau})_{\varepsilon>0}$ satisfies Assumption~\ref{assump:Differentiability}, then the following holds for $k\ne\Lambda_{0,1}(0)$,
as $\varepsilon\downarrow 0$:
\begin{equation*}
\frac{\E^{{-\Lambda^{*}(k)/\varepsilon+k+\Lambda_{1}}}{\varepsilon}^{3/2}}{u^*(k)^2\sqrt{2\pi \Lambda_{0,2}}}
\left(1+\left(\Upsilon(0,k)+\frac{1}{u^*(k)}\right)\varepsilon +\mathcal{O}\left(\varepsilon^2\right)\right) 
=\left\{
\begin{array}{ll}
\mathbb{E}\left(\E^{X^{\left(\varepsilon t\right)}_{\varepsilon\tau}}-\E^k\right)^+,
\quad & \text{if }k>\Lambda_{0,1}(0), \\
\mathbb{E}\left(\E^k-\E^{X^{\left(\varepsilon t\right)}_{\varepsilon\tau}}\right)^+,
\quad & \text{if }k<\Lambda_{0,1}(0).
\end{array}
\right.
\end{equation*}
\end{corollary}

In the Black-Scholes model, all the quantities above can be computed explicitly and we obtain:
\begin{corollary}\label{Cor:BSOptionSmallTime}
In the BSM model~\eqref{eq:BSDynamics} the following expansion holds for all $k\ne 0$, as $\varepsilon\downarrow 0$:
\begin{equation*}
\frac{\E^{k/2-k^2/(2\Sigma^2 \tau\varepsilon)}\left(\Sigma^2\tau\varepsilon\right)^{3/2}}{k^2\sqrt{2\pi}}\left[1-\left(\frac{3}{k^2}+\frac{1}{8}\right)\Sigma^2\tau\varepsilon+\mathcal{O}(\varepsilon^2)\right]=\left\{
\begin{array}{ll}
\mathbb{E}\left(\E^{X^{\left(\varepsilon t\right)}_{\varepsilon\tau}}-\E^k\right)^+,
\quad & \text{if }k>0, \\
\mathbb{E}\left(\E^k-\E^{X^{\left(\varepsilon t\right)}_{\varepsilon\tau}}\right)^+,
\quad & \text{if }k<0.
\end{array}
\right.
\end{equation*}
\end{corollary}
\begin{proof}
For the rescaled (forward) process $(X^{(\varepsilon t)}_{\varepsilon\tau})_{\varepsilon>0}$ in the BSM model~\eqref{eq:BSDynamics} we have
$\Lambda_{\varepsilon}(u)=\LO(u)+\varepsilon\Lambda_{1}(u)$
for $u\in\mathbb{R}$, where
$\LO(u)=u^2\sigma^2\tau/2$
and 
$\Lambda_{1}(u)=-u\sigma^2\tau/2$.
It follows that
$\Lambda_{0,1}(u)=u\sigma^2\tau$, 
$\Lambda_{0,2}(u)=\sigma^2\tau$ and 
$\Lambda_{1,1}(u)=-\sigma^2\tau/2$.
For any $k\in\mathbb{R}$,
$u^*(k):=k/(\sigma^2\tau)$ is the unique solution to the equation $\Lambda_{0,1}(u^*(k))=k$ and $\Lambda^*(k)=k^2/(2\sigma^2\tau)$.
$\LO$ is essentially smooth and strictly convex on $\mathbb{R}$ 
and the BSM model satisfies the other conditions in Assumption~\ref{assump:Differentiability}. 
Since $\Lambda_{0,1}(0)=0$, the result follows from Corollary~\ref{cor:ShortTimeAsymp}.
\end{proof}

It is natural to wonder why we considered diagonal small-maturity asymptotics 
and not the small-maturity asymptotic of~$\sigma_{t,\tau}$ for fixed $t>0$. 
In this case it turns out that in many cases of interest (stochastic volatility models, time-changed exponential L\'evy models), 
the forward smile blows up to infinity (except at-the-money) as $\tau$ tends to zero. 
However under the assumptions given above, this degenerate behaviour does not occur in the diagonal small-maturity regime (Corollary~\ref{cor:ShortTimeAsymp}).
In the Heston case, this explosive behaviour has been studied in~\cite{JR2013}.
More generally, we can provide a preliminary conjecture explaining the origin of this behaviour.
Consider a two-state Markov-chain
$\D X_t = -\frac{1}{2}V\D t+ \sqrt{V}\D W_t$, starting at $X_0=0$,
where $W$ is a standard Brownian motion and where $V$ is independent of $W$ and takes value $V_1$ 
with probability $p\in (0,1)$ and value $V_2 \in (0,V_1)$ with probability $1-p$. 
Conditioning on $V$ and by the independence assumption, we have
$$
\mathbb{E}\left(\E^{u(X_{t+\tau}-X_t)}\right)=p\E^{V_1 u \tau (u-1) /2}+(1-p)\E^{V_2 u \tau (u-1) /2},
\qquad\text{for all }u\in\mathbb{R}.
$$
Consider now the small-maturity regime where $\varepsilon=\tau$, $f(\varepsilon)\equiv1$ and $Y_{\varepsilon}:=X_{\varepsilon}^{(t)}$ for a fixed $t>0$. 
In this case an expansion for the re-scaled lmgf in~\eqref{eq:LambdaExpansion} as $\tau$ tends to zero is given by
$$
\Lambda_{\varepsilon}(u)
=\tau\log\mathbb{E}\left(\E^{u(X_{t+\tau}-X_t)/\tau}\right)
=\frac{V_1}{2}u^2 + \tau \log\left(p\E^{-V_1 u/2}\right)+\tau\mathcal{O}\left(\E^{-u^2(V_1-V_2)/(2\tau)}\right),
\qquad\text{for all }u\in\mathbb{R}.
$$
Since $V_1>V_2$ the remainder tends to zero exponentially fast as $\tau\downarrow 0$. 
The assumptions of Theorem~\ref{theorem:GeneralOptionAsymp} are clearly satisfied and a simple calculation shows that 
$\lim_{\tau \downarrow 0}\sigma_{t,\tau}(k)=\sqrt{V_1}$.
This example naturally extends to $n$-state Markov chains, and a natural conjecture is 
that the small-maturity forward smile does not blow up if and only if the quadratic variation of the process is bounded.
In practice, most models have unbounded quadratic variation 
(see examples in Section~\ref{sec:Examples}), 
and hence the diagonal small-maturity asymptotic is a natural scaling.

\subsubsection{Large-maturity asymptotics}\label{sec:LargeMatGeneral}
We now consider large-maturity asymptotics, when $\tau$ is large and $t$ is fixed.
Consider $(Y_\varepsilon):=(\varepsilon X^{(t)}_{1/\varepsilon})$, $\varepsilon:=1/\tau$ and
$f(\varepsilon)\equiv 1/\varepsilon$ (so that $c=1$).
Proposition~\ref{theorem:GeneralOptionAsymp} then applies and we obtain the following expansion
for forward-start options:
\begin{corollary}\label{cor:LargeTimeAsymp}
If~$(\tau^{-1}X^{(t)}_{\tau})_{\tau>0}$ satisfies Assumption~\ref{assump:Differentiability}
with $\varepsilon=\tau^{-1}$ and $1\in\mathcal{D}_{0}^o$, then the following expansion holds for all 
$k\neq \{\Lambda_{0,1}(0), \Lambda_{0,1}(1)\}$ as $\tau\uparrow \infty$:
\begin{equation*}
\frac{\E^{-\tau\left(\Lambda ^*(k)-k \right)+\Lambda_{1}}\tau^{-1/2}}
{u^*(k) \left(u^*(k)-1\right)\sqrt{2\pi\Lambda _{0,2}}}
\left(1+\frac{\Upsilon(1,k)}{\tau}+\mathcal{O}\left(\frac{1}{\tau ^2}\right)\right)
 = \left\{
\begin{array}{ll}
\mathbb{E}\left(\E^{X^{(t)}_{\tau}}-\E^{k\tau}\right)^+,
\quad & \text{if }k>\Lambda_{0,1}(1),\\
\mathbb{E}\left(\E^{k\tau}-\E^{X^{(t)}_{\tau}}\right)^+,
\quad & \text{if }k<\Lambda_{0,1}(0),\\
- \mathbb{E}\left(\E^{X^{(t)}_{\tau}}\wedge \E^{k\tau}\right),
\quad & \text{if }\Lambda_{0,1}(0)<k<\Lambda_{0,1}(1).
\end{array}
\right.
\end{equation*}
\end{corollary}
In the Black-Scholes model, all the quantities above can be computed in closed form, and we obtain:
\begin{corollary}\label{Cor:BSOptionLargeTime}
In the BSM model~\eqref{eq:BSDynamics} the following expansion holds for all $k\notin\left\{-\Sigma^2/2,\Sigma^2/2\right\}$ as $\tau\uparrow\infty$:
\begin{equation*}
\frac{ \E^{-\tau\left(\left(k+\Sigma^2/2\right)^2/(2\Sigma^2)-k\right)} 4\Sigma^3}{\left(4 k^2-\Sigma^4\right)\sqrt{2\pi\tau}}\left(1-\frac{4\Sigma^2\left(\Sigma^4+12k^2\right)}{\left(4k^2-\Sigma^4\right)^2\tau}+\mathcal{O}\left(\frac{1}{\tau^2}\right)\right)
 = \left\{
\begin{array}{ll}
\mathbb{E}\left(\E^{X^{(t)}_{\tau}}-\E^{k\tau}\right)^+,
\quad & \text{if }k>\frac{1}{2}\Sigma^2,\\
\mathbb{E}\left(\E^{k\tau}-\E^{X^{(t)}_{\tau}}\right)^+,
\quad & \text{if }k<-\frac{1}{2}\Sigma^2,\\
- \mathbb{E}\left(\E^{X^{(t)}_{\tau}}\wedge \E^{k\tau}\right),
\quad & \text{if }-\frac{1}{2}\Sigma^2<k<\frac{1}{2}\Sigma^2.
\end{array}
\right.
\end{equation*}
\end{corollary}
\begin{proof}
Consider the  process $(X^{(t)}_{\tau}/\tau)_{\tau>0}$ and set $\varepsilon=\tau^{-1}$. 
In the BSM model~\eqref{eq:BSDynamics}, 
$\Lambda_{\varepsilon} (u) 
:= \tau^{-1}\log\mathbb{E}(\exp(uX_{\tau}^{(t)}))
 = \LO(u)
 = \frac{1}{2}\Sigma^2u(u-1)$.
Thus
$\Lambda_{0,1}(u)=\Sigma^2\left(u-1/2\right)$ and 
$\Lambda_{0,2}(u)=\Sigma^2$.
For any $k\in\mathbb{R}$, $\Lambda_{0,1}(u^*(k))=k$ 
has a unique solution 
$u^*(k)=1/2+k/\Sigma^2$ and hence
$\Lambda^*(k)=\left(k+\Sigma^2/2\right)^2/(2\Sigma^2)$.
$\LO$ is essentially smooth and strictly convex on $\mathbb{R}$ 
and Assumption~\ref{assump:Differentiability} is satisfied. 
Since $\{0,1\}\subset\mathcal{D}_{0}^o$ the result follows from Corollary~\ref{cor:LargeTimeAsymp}.
\end{proof}

\subsection{Forward smile asymptotics}\label{sec:FwdSmile}
We now translate the forward-start option expansions above into asymptotics of the forward implied volatility smile 
$k\mapsto\sigma_{t,\tau}(k)$, which was defined in the introduction.

\subsubsection{Diagonal small-maturity forward smile}
We first focus on the diagonal small-maturity case. For $i=0,1,2$ we define the functions
$v_i:\mathbb{R}^*\times\mathbb{R}_{+}\times\mathbb{R}_{+}^*\to\mathbb{R}$ by
\begin{equation}\label{eq:v01SmallTime}
\left.
\begin{array}{rll}
v_0(k,t,\tau) & := \displaystyle \frac{k^2}{2 \tau  \Lambda^*(k)},\\
v_1(k,t,\tau) & := \displaystyle \frac{v_0(k,t,\tau)^2 \tau}{k}
\left[
1+ \frac{2}{k}
\log \left(\frac{k^2 \E^{\Lambda_{1}(u^*(k))}}{u^*(k)^2 \sqrt{\Lambda _{0,2}(u^*(k))}\left({\tau  v_0(k,t,\tau)}\right)^{3/2}}\right)
\right],\\
\\
v_2(k,t,\tau) & := \displaystyle \frac{2\tau ^2 v_0^3(k,t,\tau)}{k^2}\left(\frac{3}{k^2}+\frac{1}{8}\right)
 +\frac{2 \tau  v_0^2(k,t,\tau)}{k^2}\left(\Upsilon(0,k)+\frac{1}{u^*(k)}\right)\\
 & \displaystyle +\frac{v_1^2(k,t,\tau)}{v_0(k,t,\tau)}-\frac{3 \tau}{k^2}v_0(k,t,\tau) v_1(k,t,\tau),
\end{array}
\right.
\end{equation}
where $\Lambda^*$, $u^*$, $\Lambda_{i,l}$, $\Upsilon$ are defined in~\eqref{eq:LambdaStarDefinition}~\eqref{eq:u*definition},~\eqref{eq:LambdaDerivatives},~\eqref{eq:Upsilon}.
The diagonal small-maturity forward smile asymptotic is now given in the following proposition, proved in Section~\ref{sec:proofMainProp}.
\begin{proposition}\label{Prop:GeneralBSFwdVolShortTime}
Suppose that $(X^{(\varepsilon t)}_{\varepsilon\tau})_{\varepsilon>0}$ satisfies Assumption~\ref{assump:Differentiability} and that $\Lambda_{0,1}(0)=0$ 
(defined in~\eqref{eq:LambdaDerivatives}).
The following expansion then holds for the corresponding forward smile for all $k\in\mathbb{R}^*$ 
as $\varepsilon$ tends to zero:
\begin{equation}\label{eq:DiagSmallMaturityForwardSmile}
\sigma_{\varepsilon t,\varepsilon \tau}^2(k)
 = v_0(k,t,\tau)+v_1(k,t,\tau) \varepsilon +v_2(k,t,\tau) \varepsilon ^2+\mathcal{O}\left(\varepsilon ^3\right).
\end{equation}
\end{proposition}

\begin{remark}\label{rem:contdiagfuncs}\
\begin{enumerate}[(i)] 
\item 
When $\Lambda_{0,1}(0)=0$ then $\Lambda^*(k)>0$ for $k\in\mathbb{R}^*$ and $\Lambda^*(0)=0$ from Assumption~\ref{assump:Differentiability} and Lemma~\ref{lem:Properties}(iii) so that $v_0$ is always strictly positive, and 
all the $v_i$ ($i=0,1,2$) are well-defined on $\mathbb{R}^*$. 
\item
The condition $\Lambda_{0,1}(0)=0$ is equivalent to 
$\lim_{\varepsilon\downarrow0}\mathbb{E}(X^{(\varepsilon t)}_{\varepsilon \tau})=0$, 
which imposes some regularity on the paths of the process at $\varepsilon=0$.
Diffusion processes seem more readily able to satisfy this condition as opposed to jump processes where
it is well-known that implied volatility asymptotics explode in small-time (see~\cite{Tank}).
Under this condition the zeroth-order term $v_0(\cdot,t,\tau)$ in~\eqref{eq:v01SmallTime} has a well-defined limit at the origin.
\item
Using Taylor expansions in a neighbourhood of $k=0$ it can be shown that
$v_1(\cdot,t,\tau)$ has a well-defined limit at $0$
if and only if $\Lambda_{0,1}(0)=2\Lambda_{1,1}(0)+\Lambda_{0,2}(0)=0$
and $v_2(\cdot,t,\tau)$ has a well-defined limit at $0$ if and only if $\lim_{k\to0}v_0(k,t,\tau)$ and $\lim_{k\to0}v_1(k,t,\tau)$ are well-defined and 
$6\Lambda_{2,1}(0)+3\Lambda_{1,2}(0)+\Lambda_{0,3}(0)=0$.
Interestingly, these conditions can be written in similar ways to (ii).
For example, the condition $2\Lambda_{1,1}(0)+\Lambda_{0,2}(0)=0$ is equivalent to 
$
\lim_{\varepsilon\downarrow0}\mathbb{E}(X^{(\varepsilon t)}_{\varepsilon \tau})/\varepsilon=-\lim_{\varepsilon\downarrow0}\mathbb{V}(X^{(\varepsilon t)}_{\varepsilon \tau})/(2\varepsilon),
$
imposing a constraint on the mean and variance of~$X^{(\varepsilon t)}_{\varepsilon \tau}$ at $\varepsilon=0$.
Most models used in practice (and in particular those in Section~\ref{sec:Examples}) satisfy these properties and we leave the precise study of this phenomenon for future work.
\end{enumerate}
\end{remark}

\subsubsection{Large-maturity forward smile}
In the large-maturity case, define for $i=0,1,2$, the functions 
$v_i^{\infty}:\mathbb{R}\backslash\{\Lambda_{0,1}(0), \Lambda_{0,1}(1)\}\times\mathbb{R}_{+}\to\mathbb{R}$ by
\begin{equation}\label{eq:vLargeTime}
\left.
\begin{array}{rl}
v_0^{\infty}(k,t) & :=\left\{ 
  \begin{array}{l l}
   2 \left(2 \Lambda^*(k) -k-2 \sqrt{\Lambda^*(k)(\Lambda^*(k)-k)}\right), 
 & \quad \text{if }k\in\RR\backslash \left[\Lambda_{0,1}(0),\Lambda_{0,1}(1)\right],\\
 2 \left(2 \Lambda^*(k) -k+2 \sqrt{\Lambda^*(k)(\Lambda^*(k)-k)}\right),
 & \quad \text{if }k\in \left(\Lambda_{0,1}(0),\Lambda_{0,1}(1)\right),
\end{array} 
\right.
\\
v_1^{\infty}(k,t)
 & := \displaystyle 
\frac{8 v_0^{\infty}(k,t)^2}{4 k^2-v_0^{\infty}(k,t)^2}
\left(\Lambda_{1}(u^*(k))+
\log \left(\frac{4 k^2-v_0^{\infty}(k,t)^2}{4 (u^*(k)-1) u^*(k) v_0^{\infty}(k,t)^{3/2} \sqrt{\Lambda _{0,2}(u^*(k))}}\right)\right),\\
\\
v_2^{\infty}(k,t)
 &:= \displaystyle \frac{4}{v_0^{\infty}(k,t) \left(v_0^{\infty}(k,t)^2-4 k^2\right)^3}
\Big[8 k^4 v_1^{\infty}(k,t)v_0^{\infty}(k,t)^2 \left(v_1^{\infty}(k,t)+6\right)-16 k^6 v_1^{\infty}(k,t)^2  \Big. \\
& \Big.-2 \Upsilon(1,k) v_0^{\infty}(k,t)^3 \Big(v_0^{\infty}(k,t)^2-4 k^2\Big)^2
 - k^2 v_0^{\infty}(k,t)^4 \Big(96+v_1^{\infty}(k,t)^2+8 v_1^{\infty}(k,t)\Big) \Big. \\
& \Big.-v_0^{\infty}(k,t)^6 \left(v_1^{\infty}(k,t)+8\right)\Big].
\end{array}
\right.
\end{equation}
$\Lambda^*$, $u^*$,  $\Lambda_{i,l}$,  $\Upsilon$
are defined in~\ref{eq:LambdaStarDefinition},~\eqref{eq:u*definition},~\eqref{eq:LambdaDerivatives},~\eqref{eq:Upsilon}. 
The large-maturity forward smile asymptotic is given in the following proposition, proved in Section~\ref{sec:proofMainProp}.
When $t=0$ in~\eqref{eq:DiagSmallMaturityForwardSmile} and~\eqref{eq:LargeMaturityForwardSmile} below,
we recover---and improve---the asymptotics  in~\cite{FFJ11},~\cite{FJSmall},~\cite{FJ09},~\cite{FJL11},~\cite{FJM10}. 
It is interesting to note that the (strict) martingale property ($\Lambda_0 (1)=0$) is only required in Proposition~\ref{Prop:GeneralBSFwdVolLargeTime} below and not in Proposition~\ref{Prop:GeneralBSFwdVolShortTime} and Theorem~\ref{theorem:GeneralOptionAsymp}.

\begin{proposition}\label{Prop:GeneralBSFwdVolLargeTime}
Suppose that $(\tau^{-1}X^{(t)}_{\tau})_{\tau>0}$ satisfies Assumption~\ref{assump:Differentiability}, 
with $\varepsilon=\tau^{-1}$ and that 
$1\in\mathcal{D}_{0}^o$ and $\LO(1)=0$ 
(all defined in Assumption~\ref{assump:Differentiability}). 
The following then holds as~$\tau$ tends to infinity:
\begin{equation}\label{eq:LargeMaturityForwardSmile}
\sigma_{t,\tau}^2(k\tau)=v_0^{\infty}(k,t)+\frac{v_1^{\infty}(k,t)}{\tau}+\frac{v_2^{\infty}(k,t)}{\tau^2}+\mathcal{O}\left(\frac{1}{\tau^3}\right),
\qquad\text{for all }k\in\mathbb{R}\backslash\{\Lambda_{0,1}(0), \Lambda_{0,1}(1)\}.
\end{equation}
\end{proposition}

Since $\{0,1\}\subset\mathcal{D}_{0}^o$ and $\LO(1)=\LO(0)=0$, 
we always have $\Lambda^*(k)\geq \max(0,k)$ from Lemma~\ref{lem:Properties}(iii). 
One can also check that $0<v_0^{\infty}(k,t)<2|k|$ for $k\in\mathbb{R}\backslash\left[\Lambda_{0,1}(0), \Lambda_{0,1}(1)\right]$ and $v_0^{\infty}(k,t)>2|k|$ for $k\in\left(\Lambda_{0,1}(0), \Lambda_{0,1}(1)\right)$. 
This implies that the functions $v_i^{\infty}$ ($i=0,1,2$) 
are always well-defined. 
By Assumption~\ref{assump:Differentiability} and Lemma~\ref{lem:Properties}(iii) we have $\Lambda^*(\Lambda_{0,1}(0))=0$. 
Again from Lemma~\ref{lem:Properties}(iii) this implies that $\Lambda^*(\Lambda_{0,1}(1))=\Lambda_{0,1}(1)$. 
Hence $v_0^{\infty}(\cdot,t)$ is continuous on $\mathbb{R}$ with $v_0^{\infty}(\Lambda_{0,1}(1),t)=2\Lambda_{0,1}(1)$ and $v_0^{\infty}(\Lambda_{0,1}(0),t)=-2\Lambda_{0,1}(0)$.
The functions $v_1^{\infty}(\cdot,t)$ and $v_2^{\infty}(\cdot,t)$ are undefined on $\{\Lambda_{0,1}(0), \Lambda_{0,1}(1)\}$.
However, it can be shown that since $\Lambda_{0}$ is strictly convex (Assumption~\ref{assump:Differentiability}) 
and $\Lambda_{0}(1)=0$ all limits are well-defined and hence both functions can be extended by continuity to $\mathbb{R}$.
For example, using Taylor expansions in neighbourhoods of these points yields:
$$
 \lim_{k\to p}v_1^{\infty}(k,t)=2-2\sqrt{\frac{v_{0}^{\infty}(p,t)}{\Lambda_{0,2}(u^*(p))}}\left(1+\sgn(p)\left(\frac{\Lambda_{0,3}(u^*(p))}{6 \Lambda_{0,2}(u^*(p))}-\Lambda_{1,1}(u^*(p))\right)\right),\quad\text{for }p\in\{\Lambda_{0,1}(0), \Lambda_{0,1}(1)\},
$$
which, for $t=0$, agrees with~\cite[equation 19]{FJM10} for the specific case of the Heston model (Section~\ref{sec:HestonForwardSmile}).


\subsubsection{Type-II forward smile}
As mentioned in the introduction, another type of forward-start option has been considered in the literature.
We show here that the forward implied volatility expansions proved above carry over in this case with some minor modifications.
For the $(\mathcal{F}_u)$-martingale price $(\E^{X_u})_{u\geq0}$ (under~$\PP$) define the stopped process $\widetilde{X}_u^t:=X_{t\wedge u}$ for any $t>0$. 
Following~\cite{VL} define a new measure~$\widetilde{\mathbb{P}}$ by 
\begin{align}\label{eq:StoppedShareMeasure}
\widetilde{\mathbb{P}}(A):=\mathbb{E}\left(\E^{\widetilde{X}_{t+\tau}^t}\ind_{A}\right)
=\mathbb{E}\left(\E^{X_t}\ind_{A}\right), 
\qquad\text{for every }A\in\mathcal{F}_{t+\tau}.
\end{align}
The stopped process $(\E^{\widetilde{X}_u^t})_{u\geq0}$ is a $(\mathcal{F}_{t\wedge u})_u$-martingale
and~\eqref{eq:StoppedShareMeasure} defines the stopped-share-price measure~$\widetilde{\mathbb{P}}$. 
The following proposition shows how the Type-II forward smile $\widetilde{\sigma}_{t,\tau}$ 
can be incorporated into our framework. 

\begin{proposition}\label{prop:fwdsmile2}
If $\left(\E^{X_t}\right)_{t\geq0}$ is a $(\mathcal{F}_t)$-martingale under~$\PP$, 
then Propositions~\ref{Prop:GeneralBSFwdVolShortTime} and~\ref{Prop:GeneralBSFwdVolLargeTime} 
hold for the Type-II forward smile $\widetilde{\sigma}_{t,\tau}$ 
with the lmgf~\eqref{eq:Renorm-mgf} calculated under~$\widetilde{\mathbb{P}}$.
\end{proposition}

\begin{proof}
We can write the value of our Type-II forward-start call option as 
$$
\mathbb{E}\left[\left(\E^{X_{t+\tau}}-\E^{k+X_t}\right)^+\right]
=\mathbb{E}\left[\E^{X_t}\left(\E^{X_{t+\tau}-X_t}-\E^k\right)^+\right] 
=\mathbb{E}\left[\E^{\widetilde{X}_{t+\tau}^t}\left(\E^{X_{t+\tau}-X_t}-\E^k\right)^+\right]
=\mathbb{\widetilde{E}}\left[\left(\E^{X_{t+\tau}-X_t}-\E^k\right)^+\right].
$$
Proposition~\ref{theorem:GeneralOptionAsymp} and Corollaries~\ref{cor:ShortTimeAsymp},~\ref{cor:LargeTimeAsymp} 
hold in this case with all expectations (and the lmgf in~\eqref{eq:Renorm-mgf}) calculated under the stopped measure~$\widetilde{\mathbb{P}}$. 
An easy calculation shows that under~$\widetilde{\mathbb{P}}$, the forward BSM lmgf remains the same 
as under~$\mathbb{P}$. 
Thus all the previous results carry over and the proposition follows.
\end{proof}

\section{Applications} \label{sec:Examples}

\subsection{Heston} \label{sec:HestonForwardSmile}
In this section, we apply our general results to the Heston model, 
in which the (log) stock price process is the unique strong solution to the following SDEs:
\begin{equation}\label{eq:Heston}
\begin{array}{rll}
\D X_t & = \displaystyle -\frac{1}{2}V_t\D t+ \sqrt{V_t}\D W_t, \quad & X_0=0,\\
\D V_t & = \kappa\left(\theta-V_t\right)\D t+\xi\sqrt{V_t}\D B_t, \quad & V_0=v>0,\\
\D\left\langle W,B\right\rangle_t & = \rho \D t,
\end{array}
\end{equation}
with $\kappa>0$, $\xi>0$, $\theta>0$ and $|\rho|<1$ and $(W_t)_{t\geq0}$ and $(B_t)_{t\geq0}$ are two standard Brownian motions.
We shall also define $\bar{\rho}:=\sqrt{1-\rho^2}$.
The Feller SDE for the variance process has a unique strong solution 
by the Yamada-Watanabe conditions~\cite[Proposition 2.13, page 291]{KS97}). 
The $X$ process is a stochastic integral of the $V$ process and is therefore well-defined.  
The Feller condition, $2\kappa\theta\geq\xi^2$, ensures that the origin is unattainable. 
Otherwise the origin is regular (hence attainable) and strongly reflecting
(see~\cite[Chapter 15]{KT81}). 
We do not require the Feller condition in our analysis since we work with the forward lmgf of $X$ which is always well-defined.

\subsubsection{Diagonal Small-Maturity Heston Forward Smile}\label{sec:DiagSmallMatHeston}
The objective of this section is to apply Proposition~\ref{Prop:GeneralBSFwdVolShortTime} to the Heston forward smile, namely
\begin{proposition} \label{Proposition:HestonDiagonal}
In Heston, Corollary~\ref{cor:ShortTimeAsymp} and Proposition~\ref{Prop:GeneralBSFwdVolShortTime} hold
with $\mathcal{D}_{0}=\mathcal{K}_{t,\tau}$,  $\LO=\Xi$, $\Lambda_{1}=L$.
\end{proposition}
This proposition is proved in Section~\ref{sec:proofdiagsmallmat}, and all the functions therein are defined as follows:
\begin{equation}\label{eq:HestonDiagZeroOrder}
\Xi(u,t,\tau):=
\frac{u v}{\xi  \left(\bar{\rho} \cot \left(\frac{1}{2} \xi \bar{\rho} \tau  u\right)-\rho \right)-\frac{1}{2} \xi ^2 t u}
,\qquad\text{for all }u\in\mathcal{K}_{t,\tau}:=\left\{u\in\mathbb{R}:\Xi(u,0,\tau)<\frac{2v}{\xi^2t}\right\},
\end{equation}
and the functions $L, L_0,L_1:\mathcal{K}_{t,\tau}\times\mathbb{R}_{+}\times\mathbb{R}_{+}^*\to\mathbb{R}$ are defined as
\begin{equation}\label{eq:HestonDiag1stOrder}
\left.
\begin{array}{ll}
L(u,t,\tau)
&  := L_0(u,\tau)
+\Xi(u,t,\tau )^2 \left(\frac{v L_1(u,\tau)}{\Xi(u,0,\tau )^2}-\frac{\kappa  \xi ^2 t^2}{4 v}\right)
- \Xi(u,t,\tau)\kappa  t - \frac{2 \kappa\theta}{\xi^2} \log \left(1-\frac{\Xi(u,0,\tau)\xi ^2 t}{2 v}\right),\\
L_0(u,\tau)
 & := \frac{\kappa\theta}{\xi ^2} \left( (\I \xi  \rho -d_0)\I \tau  u-2 \log \left(\frac{1-g_0 \E^{-\I d_0 \tau  u}}{1-g_0}\right)\right),\\ 
L_1(u,\tau)&:=\frac{\exp(-\I d_0 \tau  u)}{\xi ^2 \left(1-g_0 \E^{-\I  d_0\tau  u}\right)}
 \left[(\I  \xi  \rho-d_0 )\I d_1 \tau  u+(d_1-\kappa)\left(1-\E^{\I  d_0 \tau  u}\right)
+\frac{(\I \xi\rho-d_0) \left(1-\E^{-\I d_0 \tau  u}\right) (g_1-\I  d_1 g_0 \tau  u)}{1-g_0\E^{-\I d_0 \tau  u}}\right],
\end{array}
\right.
\end{equation}
with
$$
d_0:=\xi  \bar{\rho },\quad 
d_1:=\frac{ \I\left(2 \kappa  \rho -\xi \right)}{2 \bar{\rho }},\qquad 
g_0:=\frac{\I \rho -\bar{\rho }}{\I\rho +\bar{\rho }}
\qquad\text{and}\qquad
g_1:=\frac{2 \kappa -\xi  \rho }{\xi  \bar{\rho } \left(\bar{\rho }+\I \rho \right)^2}.
$$
For any $t\geq0,\tau>0$ the functions $L_0$ and $L_1$ are well-defined real-valued functions 
for all $u\in\mathcal{K}_{t,\tau}$ (see Remark~\ref{remar:TisReal} for technical details). 
Also since $\Xi(0,t,\tau)/\Xi(0,0,\tau)=1$, $L$ is well-defined at $u=0$.
In order to gain some intuition on the role of the Heston parameters on the forward smile 
we expand~\eqref{eq:DiagSmallMaturityForwardSmile} around the at-the-money point in terms of the log strike $k$:
\begin{corollary}\label{cor:DiagTaylorExpansion}
The following expansion holds for the Heston forward smile as $\varepsilon$ and $k$ tend to zero:
$$
\sigma_{\varepsilon t,\varepsilon \tau}^2(k)
 = v+\varepsilon\nu_{0}(t,\tau)+\left(\frac{\rho\xi}{2}+\varepsilon\nu_{1}(t,\tau)\right)k
+\left(\frac{(4-7\rho ^2)\xi ^2}{48 v}+\frac{\xi^2 t}{4\tau v}+\varepsilon\nu_2(t,\tau)\right)k^2
+\mathcal{O}(k^3)+\mathcal{O}(\varepsilon k^3)+\mathcal{O}(\varepsilon^2).
$$
\end{corollary}
The corollary is proved in Section~\ref{sec:proofdiagsmallmat}, and the functions appearing in it are defined as follows:
\begin{equation}\label{eq:nu0}
\left.
\begin{array}{rl}
\nu_0(t,\tau)
 & := \displaystyle \frac{\tau}{48}\left(24\kappa\theta+\xi ^2\left(\rho ^2-4\right)+12v (\xi \rho-2\kappa )\right)
-\frac{t}{4}\left(\xi ^2+4\kappa \left(v-\theta\right)\right),\\
\nu_1(t,\tau)
 & := \displaystyle \frac{\rho\xi\tau}{24v}\left[\xi^2\left(1-\rho^2\right)-2 \kappa\left(v+\theta\right)+\xi\rho v\right]
+\frac{\rho\xi ^3 t}{8 v},\\
\nu_2(t,\tau)
 & := \displaystyle 
\Big[80\kappa\theta\left(13\rho^2-6\right)+\xi^2 \left(521 \rho^4-712 \rho^2+176\right)+40\rho^2 v\left(\xi\rho-2\kappa\right)\Big]
\frac{\xi^2\tau}{7680 v^2}\\ 
 & \displaystyle 
-\frac{\xi^2 t}{192 v^2}\Big[4\kappa\theta \left(16-7 \rho^2\right)+\left(7 \rho ^2-4\right) \left(9 \xi ^2+4 \kappa  v\right)\Big]
+\frac{\xi^2 t^2}{32 \tau  v^2}\Big(4\kappa\left(v-3\theta\right)+9 \xi ^2 \Big).
\end{array}
\right.
\end{equation}

\begin{remark}\label{remark:DiagHestonRemark}
The following remarks should convey some practical intuition about the results above:
\begin{enumerate}[(i)]
\item For $t=0$ this expansion perfectly lines up with Corollary 4.3 in~\cite{FJL11}.
\item Corollary~\ref{cor:DiagTaylorExpansion} implies
$\sigma_{\varepsilon t,\varepsilon \tau}(0) 
=\sigma_{0,\varepsilon \tau}(0)-\frac{\varepsilon t}{8\sqrt{v}}\left(\xi^2+4\kappa(v-\theta)\right)+\mathcal{O}(\varepsilon^2)$,
 as~$\varepsilon\downarrow 0$.
For small enough~$\varepsilon$, the spot at-the-money volatility is higher than the forward if and only if 
$\xi^2+4\kappa(v-\theta)>0$.
In particular, when $v\geq\theta$, the difference between the forward at-the-money volatility and the spot one is increasing in the forward-start dates and volatility of variance~$\xi$. 
In Figure~\ref{fig:fwdvsspot} we plot this effect using $\theta=v$ and $\theta>v+\xi^2/(4\kappa)$. 
The relative values of $v$ and $\theta$ impact the level of the forward smile vs spot smile.
\item 
For practical purposes, we can deduce some information on the forward skew by loosely differentiating 
Corollary~\ref{cor:DiagTaylorExpansion} with respect to $k$:
$$
\partial_k\sigma_{\varepsilon t,\varepsilon \tau}(0)=
\frac{\xi \rho }{4 \sqrt{v}}+\frac{ \left(4 \nu _1(t,\tau) v-\xi \rho  \nu _0(t,\tau)\right)}{8 v^{3/2}}\varepsilon+\mathcal{O}(\varepsilon^2).
$$
\item Likewise an expansion for the Heston forward convexity as $\varepsilon$ tends to zero is given by
$$
\partial^2_k\sigma_{\varepsilon t,\varepsilon \tau}(0)
=\frac{\xi ^2 ((2-5 \rho ^2) \tau +6 t)}{24 \tau  v^{3/2}}
-\frac{\nu _0(t,\tau)\xi^2(3t+(1-4 \rho^2) \tau)+6 \tau v (\rho\xi\nu _1(t,\tau)-4 \nu _2(t,\tau) v)}{24 \tau v^{5/2}}\varepsilon +\mathcal{O}(\varepsilon^2),
$$
and in particular
$\partial^2_k\sigma_{\varepsilon t,\varepsilon \tau}(0) 
=\partial^2_k\sigma_{0,\varepsilon \tau}(0)
+\xi ^2 t/(4 \tau  v^{3/2})+\mathcal{O}(\varepsilon)$.
For fixed maturity the forward convexity is always greater than the spot implied volatility convexity (see Figure~\ref{fig:fwdvsspot})
and this difference is increasing in the forward-start dates and volatility of variance. 
At zeroth order in $\varepsilon$ the wings of the forward smile increase to arbitrarily high levels with decreasing maturity. (see Figure~\ref{fig:hestonexplosion}(a))
This effect has been mentioned qualitatively by practitioners~\cite{B02}.
As it turns out for fixed $t>0$ the Heston forward smile blows up to infinity (except at-the-money) as the maturity tends to zero, see~\cite{JR2013} for details.
\end{enumerate}
\end{remark}

In the Heston model, $(\E^{X_t})_{t\geq0}$ is a true martingale~\cite[Proposition 2.5]{AP07}.
Applying Proposition~\ref{prop:fwdsmile2} with Lemma~\ref{lemma:sharepicemeasureheston}, 
giving the Heston forward lmgf under the stopped-share-price measure,
we derive the following asymptotic for the Type-II Heston forward smile~$\widetilde{\sigma}_{t,\tau}$:
\begin{corollary} \label{cor:DiagTaylorExpansionTypeII}
The diagonal small-maturity expansion of the Heston Type-II forward smile as $\varepsilon$ 
and $k$ tend to zero
is the same as the one in Corollary~\ref{cor:DiagTaylorExpansion} with $\nu_0$, $\nu_1$ and $\nu_2$ replaced by 
$\widetilde{\nu}_0$, $\widetilde{\nu}_1$ and $\widetilde{\nu}_2$, where
$$
\widetilde{\nu}_0(t,\tau):=\nu_0(t,\tau)+\xi\rho v t,\qquad
\widetilde{\nu}_1(t,\tau):=\nu_1(t,\tau),\qquad
\widetilde{\nu}_2(t,\tau):=\nu_2(t,\tau)+\frac{\rho\xi^3 t}{48v}\left(7\rho^2-4\right)-\frac{\rho\xi^3 t^2}{8v\tau}.
$$
\end{corollary}
Its proof is analogous to the proofs of 
Proposition~\ref{Proposition:HestonDiagonal} and Corollary~\ref{cor:DiagTaylorExpansion}, and is therefore omitted.
Note that when $\rho=0$ or $t=0$, $\nu_i=\widetilde{\nu}_i$ ($i=1,2,3$), and the Heston forward smiles Type-I and Type-II are the same.

\begin{figure}[h!tb] 
\centering
\mbox{\subfigure[Small-maturity forward smile explosion.]{\includegraphics[scale=0.7]{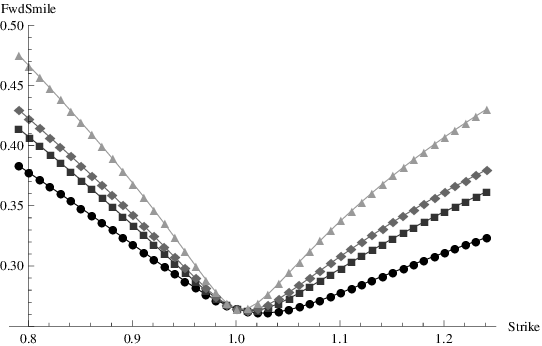}}\quad
\subfigure[Type I vs Type II forward smile.]{\includegraphics[scale=0.7]{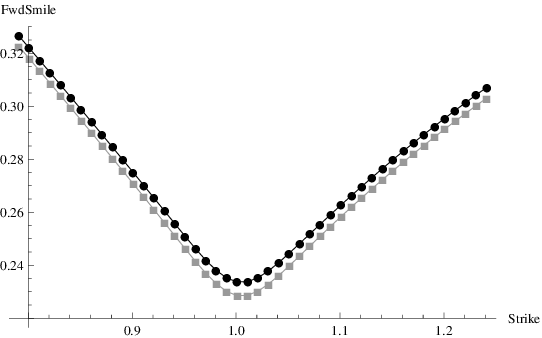}}}
\caption{(a): Forward smiles with forward-start date $t=1/2$ and maturities $\tau=1/6,1/12,1/16,1/32$ given by circles, squares, diamonds and triangles respectively using the Heston parameters 
$(v,\theta,\kappa,\rho,\xi) =(0.07,0.07,1,-0.6,0.5)$ and the asymptotic in Proposition~\ref{Proposition:HestonDiagonal}. 
(b): Type I (circles) vs Type 2 (squares) forward smile with $t=1/2$, $\tau=1/12$ and the Heston parameters $(v,\theta,\kappa,\rho,\xi) = (0.07,0.07,1,-0.2,0.34)$ using Corollaries~\ref{cor:DiagTaylorExpansion} and~\ref{cor:DiagTaylorExpansionTypeII}. }
\label{fig:hestonexplosion}
\end{figure}

\begin{figure}[h!tb] 
\centering
\includegraphics[scale=0.7]{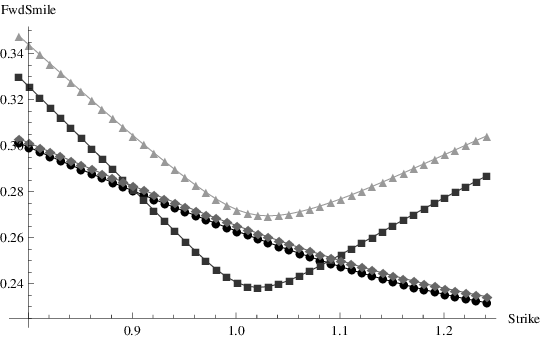}
\caption{
Forward smile vs spot smile with $v=\theta$ and $\theta>v+\xi^2/(4\kappa)$.
Circles ($t=0,\tau=1/12$) and squares ($t=1/2,\tau=1/12$) use the Heston parameters $v=\theta=0.07$,$\kappa=1$, $\rho=-0.6$, $\xi=0.3$. 
Diamonds ($t=0,\tau=1/12$) and triangles ($t=1/2$, $\tau=1/12$) use the same parameters but with 
$\theta=0.1$. Plots use the asymptotic in Proposition~\ref{Proposition:HestonDiagonal}.}
\label{fig:fwdvsspot}
\end{figure}

\subsubsection{Large-maturity Heston forward smile}\label{sec:LargeMatHeston}
Our main result here is Proposition~\ref{Prop:HestonLargeMaturity}, which is an application of  Proposition~\ref{Prop:GeneralBSFwdVolLargeTime} to the Heston forward smile. 
We shall always assume here that $\kappa>\rho\xi$. 
When this condition fails, moments of the stock price process~\eqref{eq:Heston} strictly greater than one cease to exist 
for large enough time,  
and consequently the limiting lmgf is not essentially smooth on its effective domain and Assumption~\ref{assump:Differentiability}(iv) is violated. 
This a standard assumption in the large-maturity implied volatility asymptotics literature~\cite{FJ09,FJM10,JMKR}, but
bears no consequences in markets where the implied volatility skew is downward sloping, such as equity markets, 
where the correlation is negative. 
Define the quantities
\begin{equation}\label{eq:DefUpmU*pm}
\begin{array}{ll}
u_{\pm} := \displaystyle \frac{\xi-2\kappa\rho\pm\eta}{2\xi(1-\rho^2)},
& u_{\pm}^{*} := \displaystyle \frac{\psi\pm\nu}{2 \xi (\E^{\kappa  t}-1)},\\
\eta := \displaystyle \sqrt{\xi^2(1-\rho^2)+(2\kappa-\rho\xi)^2},
& \nu := \displaystyle \sqrt{\psi^2-16\kappa^2\E^{\kappa t}},
\\
\rho_{\pm} := \displaystyle \frac{\E^{-2 \kappa  t} \left(\xi(\E^{2\kappa t}-1) \pm(\E^{\kappa  t}+1) \sqrt{16 \kappa ^2 \E^{2 \kappa  t}+\xi ^2 (1-\E^{\kappa  t})^2}\right)}{8 \kappa},
& \psi := \displaystyle \xi (\E^{\kappa  t}-1)-4\kappa\rho\E^{\kappa  t},
\end{array}
\end{equation}
as well as the interval $\Ddh\subset\mathbb{R}$ by
\begin{equation}\label{eq:DInfinityLargeMaturity}
\Ddh :=\left\{ 
  \begin{array}{l l}
   \left[u_{-},u_{+}^*\right), & \quad \text{if } -1<\rho< \rho_{-}\text{ and }t>0,\\
 \left(u_{-}^*,u_{+}\right], 
& \quad \text{if }\rho_{+}<\rho<\min(1,\kappa/\xi), t>0
\text{ and }\kappa>\rho_{+}\xi,\\
  \left[u_{-},u_{+}\right], & \quad \text{if }\rho_{-}\leq\rho\leq\min(\rho_{+},\kappa/\xi),
    \\
\end{array} \right.
\end{equation}
Details about each case are given in Lemma~\ref{lemma:rhoconditions}.
We define the functions $V$ and $H$ from $\Ddh$ to $\mathbb{R}$ by
\begin{align}
V(u) & := \frac{\kappa\theta}{\xi^2}\left(\kappa-\rho\xi u-d(u)\right)
\qquad\text{and}\qquad
H(u) :=
\frac{V(u)v \E^{-\kappa t}}{\kappa\theta -2 \beta _t V(u)}
-\frac{2\kappa\theta}{\xi ^2}\log\left(\frac{\kappa\theta-2 \beta _t V(u)}{\kappa\theta \left(1-\gamma (u)\right)}\right),
\label{eq:VandH}\\
d(u) & := \left((\kappa-\rho\xi u)^2+u\xi^2(1-u)\right)^{1/2},
\quad
\gamma(u) := \frac{\kappa-\rho\xi u-d(u)}{\kappa-\rho\xi u+d(u)},
\quad\text{and }
\beta_t := \frac{\xi^2}{4\kappa}\left(1-\E^{-\kappa t}\right).\label{eq:DGammaBeta}
\end{align}

From the proof of Proposition~\ref{Prop:HestonLimitingmgfLargeTime}, one can see that
$V$ and $H$ are always well-defined real-valued functions on~$\Ddh$.
Finally we define the functions $q^*:\mathbb{R}\to[u_{-},u_{+}]$ and $V^*:\mathbb{R}\to\mathbb{R}_{+}$ by
\begin{equation}\label{eq:V*q*}
q^*(x):=\frac{\xi-2\kappa\rho+\left(\kappa\theta\rho+x\xi\right)\eta\left(x^2\xi^2+2x\kappa\theta\rho\xi+\kappa^2\theta^2\right)^{-1/2}}{2\xi \left(1-\rho^2\right)}
\qquad\text{and}\qquad 
V^*(x):=q^*(x)x-V\left(q^*(x)\right).
\end{equation}
The following proposition gives the large-maturity forward Heston smile in Case (iii) in~\eqref{eq:DInfinityLargeMaturity},
and its proof is postponed to Section~\ref{sec:ProofsLargeMatHeston}.

\begin{proposition}\label{Prop:HestonLargeMaturity}
If $\rho_{-}\leq\rho\leq\min\left(\rho_{+},\kappa/\xi\right)$, then Corollary~\ref{cor:LargeTimeAsymp} and Proposition~\ref{Prop:GeneralBSFwdVolLargeTime} hold with $\LO=V$, $\Lambda^*=V^*$, $u^*=q^*$, 
$\Lambda_{1}=H$, $\Lambda_{2}=0$ and $\mathcal{D}_{0}=\Ddh=\left[u_{-},u_{+}\right]$.
\end{proposition}

\begin{remark}\label{remark:largemathest} \
\begin{enumerate}[(i)]
\item Note that $V^*$ is nothing else than the Fenchel-Legendre transform of $V$, 
and $q^*$ the corresponding saddlepoint (see~\cite{FJ09} for computational details).
\item In the Heston model there is no $t$-dependence for $v_0^{\infty}$ in~\eqref{eq:LargeMaturityForwardSmile} 
since $V^*$ does not depend on $t$.
Therefore under the conditions of the proposition, the limiting (zeroth order) smile is exactly of SVI form (see~\cite{GJ11}).
\item For Cases (i) and (ii) in~\eqref{eq:DInfinityLargeMaturity} the essential smoothness property in Assumption~\ref{assump:Differentiability}(iv) is not satisfied and a different strategy needs to be employed to derive a sharp large deviations result for large-maturity forward-start options. We leave this analysis for future research.
\item $t=0$ implies that $\rho_{\pm}=\pm1$ and Proposition~\ref{Prop:HestonLargeMaturity} extends the large-maturity asymptotics in~\cite{FJ09} and~\cite{FJM10}.
\item For practical purposes, note that $\rho \in [0,\min(1/2,\kappa/\xi)]$ is always satisfied under the assumptions of the proposition.
\item Even though the function $V^*$ does not depend on~$t$, 
$\rho_{\pm}$ and the function $H$ do (see the at-the-money example below). 
That said, to zeroth order and correlation close to zero, 
the large-maturity forward smile is the same as the large-maturity spot smile. 
This is a very different result compared to the Heston small-maturity forward smile 
(see Remark~\ref{remark:DiagHestonRemark}(iv)), 
where large differences emerge between the forward smile and the spot smile at zeroth order. 
\end{enumerate}
\end{remark}
We now give an example illustrating some of the differences between the Heston large-maturity forward smile 
and the large-maturity spot smile due to first-order differences in the asymptotic~\eqref{eq:LargeMaturityForwardSmile}. 
This ties in with Remark~\ref{remark:largemathest}(vi). 
Specifically we look at the forward at-the-money volatility which, when using Proposition~\ref{Prop:HestonLargeMaturity} 
with $\rho_{-}\leq\rho\leq\text{min}\left(\rho_{+},\kappa/\xi\right)$, satisfies
$\sigma_{t,\tau}^2(0)=v_0^{\infty}(0)+v_1^{\infty}(0,t)/\tau+\mathcal{O}\left(1/\tau^2\right)$,
as $\tau$ tends to infinity, with 
\begin{align*}
v_0^{\infty}(0)&=\frac{4 \theta  \kappa  (\eta -2 \kappa +\xi  \rho )}{\xi ^2 \left(1-\rho ^2\right)},\\
v_1^{\infty}(0,t) & =
\frac{16\kappa v\left(\rho\xi-2\kappa+\eta\right)}{\Delta\xi^2}
+\frac{16\kappa\theta}{\xi ^2} \log\left(\frac{\Delta  \E^{-\kappa t} \left(2 \kappa -\xi \rho+\left(1-2\rho^2\right)\eta\right)}{8\kappa\left(1-\rho^2\right)^2\eta}\right)\\
 & -8 \log\left(\frac{\xi  \left(1-\rho ^2\right)^{3/2}\sqrt{\eta\left(2 \xi\rho -4\kappa+2 \eta\right) }}
{\left(\xi  \left(1-2 \rho ^2\right)-\rho  (\eta-2 \kappa)\right) \left(\rho(\eta-2\kappa )+\xi \right)}\right);
\end{align*}
$\eta$ is defined in~\eqref{eq:DefUpmU*pm} and
$\Delta:=2 \kappa \left(1+\E^{\kappa  t}\left(1-2 \rho ^2\right) \right)-\left(1-\E^{\kappa  t}\right) (\rho\xi+\eta)$.
To get an idea of the $t$-dependence of the at-the-money forward volatility we set $\rho=0$ 
(since Proposition~\ref{Prop:HestonLargeMaturity} is valid for correlations near zero) 
and perform a Taylor expansion of $v_1^{\infty}(0,t)$ around $t=0$: 
$
v_1^{\infty}(0,t)=v_1^{\infty}(0,0)
+\left(\frac{2\theta}{1+\sqrt{1+\xi ^2/4 \kappa^2}}-v\right)t +\mathcal{O}\left(t^2\right).
$
When $v\geq\theta$ then at this order the large $\tau$-maturity forward at-the-money volatility 
is lower than the corresponding large $\tau$-maturity at-the-money implied volatility 
and this difference is increasing in $t$ and in the ratio $\xi/\kappa$. 
This is similar in spirit to Remark~\ref{remark:DiagHestonRemark}(ii) for the small-maturity Heston forward smile.


\subsection{Time-changed exponential L\'evy } \label{sec:ExponentialLevyForwardSmile}

It is well known~\cite[Proposition 11.2]{CT07} that the forward smile 
in exponential L\'evy models is time-homogeneous in the sense that $\sigma_{t,\tau}$ does not depend on $t$, 
(by stationarity of the increments). 
This is not necessarily true in time-changed exponential L\'evy models as we shall now see.
Let~$N$ be a  L\'evy process with lmgf given by
$\log\mathbb{E}\left(\E^{uN_t}\right) = t\phi(u)$
for $t\geq 0$ and $u\in\mathcal{K}_{\phi}:=\left\{u\in\mathbb{R}:|\phi(u)|<\infty\right\}$.
We consider models where $X:=(N_{V_t})_{t\geq0}$ pathwise
and the time-change is given by $V_t:=\int_{0}^{t}v_s\D s$ with $v$ being 
a strictly positive process independent of $N$. 
We shall consider the two following examples:
\begin{align}
\D v_t & = \kappa\left(\theta-v_t\right)\D t+\xi\sqrt{v_t}\D B_t, \label{eq:fellerdiff}\\ 
\label{eq:nonGaussOU}
\D v_t & = -\lambda v_t \D t+\D J_t,
\end{align}
with $v_0=v>0$ and $\kappa, \xi, \theta, \lambda>0$. 
Here $B$ is a standard Brownian motion and $J$ is a compound Poisson subordinator with exponential jump size distribution and L\'evy exponent $l(u):=\lambda \delta u/(\alpha-u)$ for all $u<\alpha$ with $\delta>0$ and $\alpha>0$. 
In~\eqref{eq:fellerdiff}, $v$ is a Feller diffusion and in~\eqref{eq:nonGaussOU}, it is a $\Gamma$-OU process. 
We now define the functions $\widehat{V}$ and $\widehat{H}$ from $\widehat{\mathcal{K}}_{\infty}$ to $\mathbb{R}$,
and the functions $\widetilde{V}$ and $\widetilde{H}$ from $\widetilde{\mathcal{K}}_{\infty}$ to $\mathbb{R}$ by
\begin{align}
& \widehat{V}(u)  := \frac{\kappa\theta}{\xi^2}\left(\kappa-\sqrt{\kappa^2-2\phi(u)\xi^2}\right),
\qquad
\widehat{H}(u)  :=
\frac{\widehat{V}(u)v \E^{-\kappa t}}{\kappa\theta -2 \beta _t \widehat{V}(u)}
-\frac{2\kappa\theta}{\xi ^2}\log\left(\frac{\kappa\theta-2 \beta _t \widehat{V}(u)}{\kappa\theta \left(1-\gamma(\phi(u))\right)}\right),\label{eq:FellerLargeTime}\\
& \widetilde{V}(u)  := \frac{\phi(u)\lambda \delta}{\alpha\lambda-\phi(u)}, 
\qquad
\widetilde{H}(u) :=
\frac{\lambda \alpha\delta}{\alpha\lambda-\phi(u)}\log\left(1-\frac{\phi(u)}{\alpha \lambda}\right)+\frac{\phi(u) v \E^{-\lambda t}}{\lambda}+d\log\left(\frac{\phi(u)-\alpha\lambda\E^{\lambda t}}{\E^{t\lambda}(\phi(u)-\alpha\lambda)}\right),\label{eq:GammaOULargeTime}
\end{align}
where we set
\begin{equation}\label{eq:domains}
\widehat{\mathcal{K}}_{\infty}:=\left\{u:\phi(u)\leq\kappa^2/(2\xi^2)\right\}, 
\qquad\text{and}\qquad
\widetilde{\mathcal{K}}_{\infty}:=\left\{u:\phi(u)<\alpha\lambda\right\};
\end{equation}
$\phi$ is the L\'evy exponent of~$N$ and $\beta_t$ and $\gamma$ are defined in~\eqref{eq:dgammabetatimechange}. 
The following proposition---proved in Section~\ref{sec:ProofsExpLevy}---is the main result of the section.

\begin{proposition}\label{prop:fwdsmiletimechange}
Suppose that $\phi$ is essentially smooth (Assumption~\ref{assump:Differentiability}(iv)), 
strictly convex and of class $\mathcal{C}^{\infty}$ on $\mathcal{K}^o_{\phi}$ with $\{0,1\}\subset\mathcal{K}^o_{\phi}$ and $\phi(1)=0$. 
Then Corollary~\ref{cor:LargeTimeAsymp} and Proposition~\ref{Prop:GeneralBSFwdVolLargeTime} hold:
\begin{enumerate}[(i)]
\item when $v$ follows~\eqref{eq:fellerdiff},
with $\LO=\widehat{V}$, $\Lambda_{1}=\widehat{H}$, $\Lambda_{2}=0$ 
and $\mathcal{D}_{0}=\widehat{\mathcal{K}}_{\infty}$;
\item when $v$ follows~\eqref{eq:nonGaussOU}, 
with $\LO=\widetilde{V}$, $\Lambda_{1}=\widetilde{H}$, $\Lambda_{2}=0$ 
and $\mathcal{D}_{0}=\widetilde{\mathcal{K}}_{\infty}$;
\item when $v_t\equiv 1$, 
with $\LO=\phi$,
$\Lambda_{1}=0$, $\Lambda_{2}=0$ and $\mathcal{D}_{0}=\mathcal{K}_{\phi}$.
\end{enumerate}
\end{proposition}

\begin{remark}\
\begin{enumerate}[(i)]
\item The uncorrelated Heston model~\eqref{eq:Heston} can be represented as 
$N_t:=-t/2+B_t$ time-changed by an integrated Feller diffusion~\eqref{eq:fellerdiff}. 
With $\phi(u)\equiv u(u-1)/2$ and $\mathcal{K}_{\phi}=\mathbb{R}$,
Proposition~\ref{prop:fwdsmiletimechange}(i) agrees with Proposition~\ref{Prop:HestonLargeMaturity}.
\item The zeroth order large-maturity forward smile is the same as its corresponding zeroth 
order large-maturity spot smile and differences only emerge at first order. 
It seems plausible that this will always hold 
if there exists a stationary distribution for $v$ and if $v$ is independent of the L\'evy process~$N$;
\item Case (iii) in the proposition corresponds to the standard exponential L\'evy case (without time-change).
\end{enumerate}
\end{remark}

We now use Proposition~\ref{prop:fwdsmiletimechange} to highlight the first-order differences 
in the large-maturity forward smile~\eqref{eq:LargeMaturityForwardSmile} and the corresponding spot smile. 
If $v$ follows~\eqref{eq:fellerdiff} then a Taylor expansion of $v^{\infty}_1$ in~\eqref{eq:vLargeTime} around $t=0$ gives
$$
v^{\infty}_1(t,k)=v^{\infty}_1(0,k)+ \frac{8 v_0^{\infty}(k)^2}{4 k^2-v_0^{\infty}(k)^2} \widehat{V}(u^*(k)) \left(\frac{\xi ^2 v \widehat{V}(u^*(k))}{2 \theta ^2 \kappa ^2}+1-\frac{v}{\theta }\right)  t +\mathcal{O}(t^2),
\quad\text{for all }k\in\RR\setminus\{\widehat{V}'(0), \widehat{V}'(1)\}.
$$
Using simple properties of $v_0^{\infty}$ and $\widehat{V}$ 
we see that the large-maturity forward smile is lower than the corresponding spot smile 
for $k\in(\widehat{V}'(0),\widehat{V}'(1))$ (which always include the at-the-money) if $v\geq\theta$. 
The forward smile is higher than the corresponding spot smile for $k\in\mathbb{R}\backslash(\widehat{V}'(0),\widehat{V}'(1))$ (OTM options) if $v\leq\theta$, and these differences are increasing in $\xi/\kappa$ and $t$. This effect is illustrated in Figure~\ref{fig:timechangedfeller} and $k\in(\widehat{V}'(0),\widehat{V}'(1))$ corresponds to strikes in the region $(0.98,1.02)$ in the figure.

If $v$ follows~\eqref{eq:nonGaussOU} then a simple Taylor expansion of $v^{\infty}_1(\cdot,k)$ 
in~\eqref{eq:vLargeTime} around $t=0$ gives
\begin{equation*}
v^{\infty}_1(t,k)
=v^{\infty}_1(0,k)+ \frac{8 v_0^{\infty}(k)^2}{4 k^2-v_0^{\infty}(k)^2}
\frac{\phi (u^*(k)) \left[\lambda  (\delta-\alpha  v)+v \phi (u^*(k))\right]}{\alpha  \lambda -\phi (u^*(k))} t +\mathcal{O}(t^2),
\quad\text{for all }k\in\RR\setminus\{\widetilde{V}'(0),\widetilde{V}'(1)\}.
\end{equation*}
Similarly we deduce that the large-maturity forward smile is lower than the corresponding spot smile for $k\in(\widetilde{V}'(0),\widetilde{V}'(1))$ if $v\geq \delta/\alpha$. 
The forward smile is higher than the corresponding spot smile for $k\in\mathbb{R}\backslash(\widetilde{V}'(0),\widetilde{V}'(1))$ (OTM options) if $v\leq \delta/\alpha$, and these differences are increasing in $t$.

If $v$ follows~\eqref{eq:fellerdiff}(\eqref{eq:nonGaussOU}) then the stationary distribution is a gamma distribution with mean $\theta$ ($\delta/\alpha$), see~\cite[page 475 and page 487]{CT07}. 
The above results seem to indicate that the differences in level between the large-maturity forward smile 
and the corresponding spot smile depend on the relative values of $v_0$ and the mean of the stationary distribution of the process $v$. This is also similar to Remark~\ref{remark:DiagHestonRemark}(ii) and the analysis below Remark~\ref{remark:largemathest} for the Heston forward smile.
These observations are also independent of the choice of $\phi$ indicating that the fundamental quantity driving the non-stationarity of the large-maturity forward smile over the corresponding spot implied volatility smile is the choice of time-change.

In the Variance-Gamma model~\cite{MadanVG}, 
$\phi(u)\equiv\mu u+C\log\left(\frac{GM}{\left(M-u\right)\left(G+u\right)}\right)$,
for $u\in(-G,M)$,
with $C>0$, $G>0$, $M>1$ and $\mu:=-C\log\left(\frac{GM}{\left(M-1\right)\left(G+1\right)}\right)$
ensures that $(\E^{X_t})_{t\geq0}$ is a true martingale ($\phi(1)=0$). 
Clearly
$\phi$ is essentially smooth, strictly convex and infinitely differentiable on $(-G,M)$ with $\{0,1\}\subset(-G,M)$; therefore Proposition~\ref{prop:fwdsmiletimechange} applies. 
For Proposition~\ref{prop:fwdsmiletimechange}(iii), 
the solution to $\phi'(u^*(k))=k$ is $u^*(\mu)=(M-G)/2$ and
$$
u_{\pm}^*(k)=\frac{-2 C-(G-M) (k-\mu )\pm\sqrt{4 C^2+(G+M)^2 (k-\mu )^2}}{2 (k-\mu )}\quad\text{for all }k\neq\mu.
$$
The sign condition $\left(M-u\right)\left(G+u\right)>0$ imposes $-2 C\pm\sqrt{4 C^2+(G+M)^2 (k-\mu )^2}>0$ for all $k\neq\mu$.
Hence $u_{+}^*$ (continuous on the whole real line) is the only valid solution and the rate function is then given in closed-form as
$\Lambda^*(k)=k u_{+}^*(k)-\phi(u_{+}^*(k))$ for all real $k$.

\begin{figure}[h!tb] 
\centering
\mbox{\subfigure[Feller time-change: forward smile vs spot smile $v>\theta$.]{\includegraphics[scale=0.7]{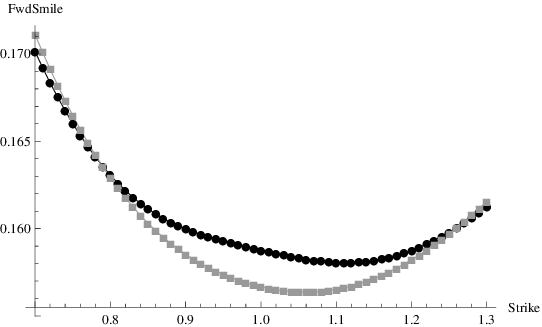}}\quad
\subfigure[Feller time-change: forward smile vs spot smile $v<\theta$.]{\includegraphics[scale=0.7]{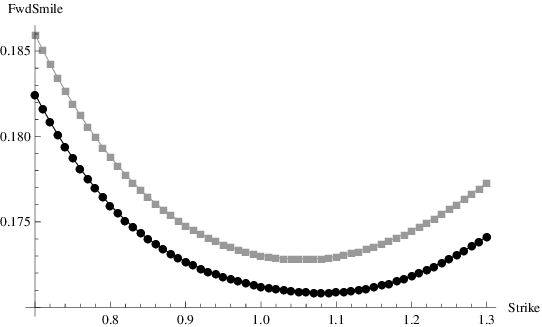}}}
\caption{Circles represent $t=0$ and $\tau=2$ and squares represent $t=1/2$ and $\tau=2$ 
using a Variance-Gamma model time-changed by a Feller diffusion and the asymptotic in Proposition~\ref{prop:fwdsmiletimechange}. 
In (a) the parameters are $C=58.12$, $G=50.5$, $M=69.37$, $\kappa=1.23$, $\theta=0.9$, $\xi=1.6$, $v=1$ and 
(b) uses the same parameters but with $\theta=1.1$.}
\label{fig:timechangedfeller}
\end{figure}

\section{Numerics}\label{sec:Numerics}

We compare here the true forward smile in various models and the asymptotics developed in Propositions~\ref{Prop:GeneralBSFwdVolShortTime} and~\ref{Prop:GeneralBSFwdVolLargeTime}. 
We calculate forward-start option prices using the inverse Fourier transform representation in~\cite[Theorem 5.1]{L04F}
and a global adaptive Gauss-Kronrod quadrature scheme. 
We then compute the forward smile $\sigma_{t,\tau}$ and
compare it to the zeroth, first and second order asymptotics given in Propositions~\ref{Prop:GeneralBSFwdVolShortTime} and~\ref{Prop:GeneralBSFwdVolLargeTime} for various models.
In Figure~\ref{fig:HestDiagComp} we compare the Heston diagonal small-maturity asymptotic in Proposition~\ref{Proposition:HestonDiagonal} with the true forward smile. 
Figure~\ref{fig:HestLargeMatComp} tests the accuracy of the Heston large-maturity asymptotic from Proposition~\ref{Prop:HestonLargeMaturity}. 
In order to use this proposition we require $\rho_{-}\leq\rho\leq\min\left(\rho_{+},\kappa/\xi\right)$ with $\rho_{\pm}$ defined in~\eqref{eq:DefUpmU*pm}. 
For the parameter choice in the figure we have $\rho_{-}=-0.65$ and the condition is satisfied. 
Finally in Figure~\ref{fig:GammaOULargeMatComp} we consider the Variance Gamma model 
time-changed by a $\Gamma$-OU process using Proposition~\ref{prop:fwdsmiletimechange}. 
Results are in line with expectations and the higher the order of the asymptotic the closer we match the true forward smile.

\begin{figure}[h!tb] 
\centering
\mbox{\subfigure[Heston diagonal small-maturity vs Fourier inversion.]{\includegraphics[scale=0.7]{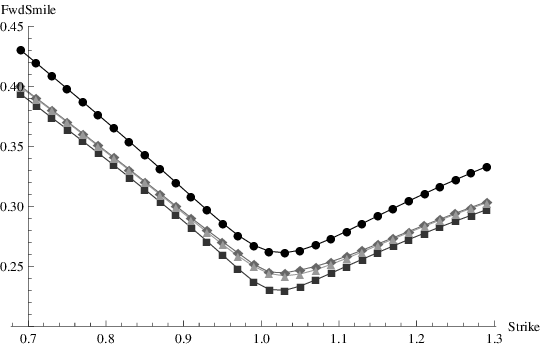}}\quad
\subfigure[Errors]{\includegraphics[scale=0.7]{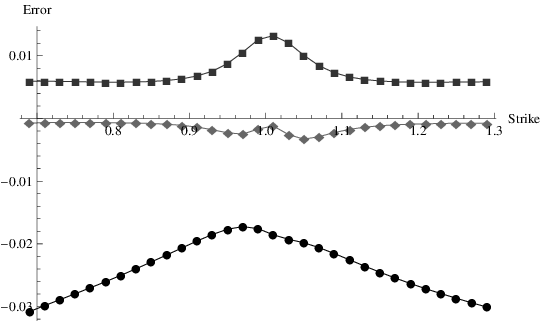}}}
\caption{In (a) circles, squares and diamonds represent the zeroth, first and second order asymptotics respectively in Proposition~\ref{Proposition:HestonDiagonal} and triangles represent the true forward smile using Fourier inversion. In (b) we plot the differences between the true forward smile and the asymptotic. 
Here, $t=1/2$, $\tau=1/12$,  $v=0.07$, $\theta=0.07$, $\kappa=1$, $\xi=0.34$, $\rho=-0.8$.}
\label{fig:HestDiagComp}
\end{figure}

\begin{figure}[h!tb] 
\centering
\mbox{\subfigure[Heston Large-Maturity vs Fourier Inversion.]{\includegraphics[scale=0.7]{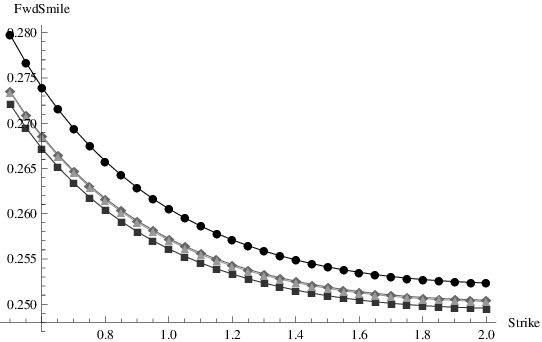}}\quad
\subfigure[Errors]{\includegraphics[scale=0.7]{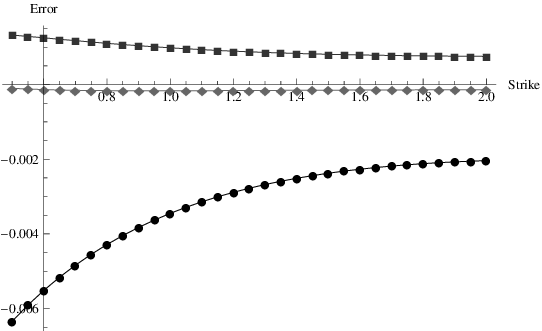}}}
\caption{In (a) circles, squares and diamonds represent the zeroth, first and second order asymptotics respectively in Proposition~\ref{Prop:HestonLargeMaturity} and triangles represent the true forward smile using Fourier inversion. In (b) we plot the differences between the true forward smile and the asymptotic. 
Here, $t=1$, $\tau=5$, $v=0.07$, $\theta=0.07$, $\kappa=1.5$, $\xi=0.34$, $\rho=-0.25$.}
\label{fig:HestLargeMatComp}
\end{figure}

\begin{figure}[h!tb] 
\centering
\mbox{\subfigure[$\Gamma$-OU time-change large-maturity / Fourier inversion.]{\includegraphics[scale=0.7]{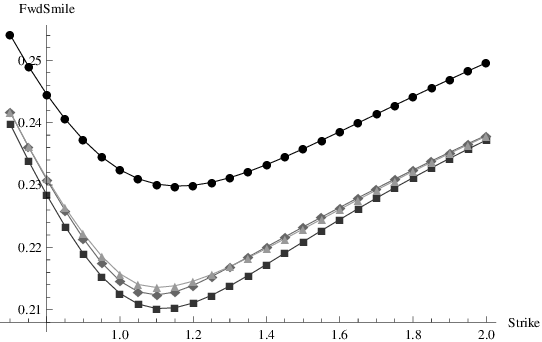}}\quad
\subfigure[Errors]{\includegraphics[scale=0.7]{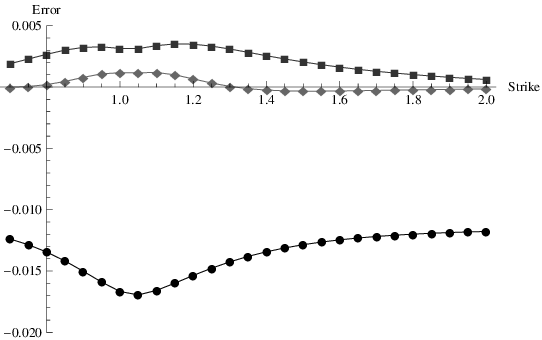}}}
\caption{In (a) circles, squares and diamonds represent the zeroth, first and second order asymptotics respectively in Proposition~\ref{prop:fwdsmiletimechange} and triangles represent the true forward smile using Fourier inversion for a variance gamma model time-changed by a $\Gamma$-OU process. In (b) we plot the differences between the true forward smile and the asymptotic. 
We use $t=1$ and $\tau=3$ 
with the parameters $C=6.5$, $G=11.1$, $M=33.4$, $v=1$, $\alpha=0.6$, $d=0.6$, $\lambda=1.8$.}
\label{fig:GammaOULargeMatComp}
\end{figure}

\newpage
\section{Proofs}\label{sec:Proofs}
\subsection{Proofs of Section~\ref{sec:GeneralResults}}\label{sec:proofMainProp}
\subsubsection{Proof of Theorem~\ref{theorem:GeneralOptionAsymp}}
Our proof relies on several steps and is based on so-called sharp large deviations tools. 
We first ---as in classical large deviations theory---define an asymptotic measure-change
allowing for weak convergence of a rescaled version of $(Y_{\varepsilon})_{\varepsilon>0}$.
In Lemma~\ref{lemma:CharactExp} we derive the asymptotics of the characteristic function of this rescaled process under this new measure.
The limit is a Gaussian characteristic function making all forthcoming computations analytically tractable.
We then write the option price as an expectation of the rescaled process under the new  measure
(see~\eqref{eq:optionPricesInverse}),
and prove an inverse Fourier transform representation (Lemma~\ref{lem:optpricerep0}) for sufficiently small $\varepsilon$.
Splitting the integration domain (Equation~\eqref{eq:IntExp}) of this inverse Fourier transform in two 
(compact interval and tails), 
(a) we integrate term by term the compact part, 
and (b) we show that Assumption~\ref{assump:Differentiability}(v) implies that the tail part is exponentially small (Lemma~\ref{lemma:tailestimatesteepcase}).
We now start the analysis and define such a change of measure by
\begin{align}\label{eq:MeasureChange}
\frac{\D\mathbb{Q}_{k,\varepsilon}}{\D\mathbb{P}}
=\exp\left(\frac{u^*(k)Y_{\varepsilon}}{\varepsilon}-\frac{\Lambda_{\varepsilon}\left(u^*(k)\right)}{\varepsilon}\right),
\end{align}
with $u^*(k)$ defined in~\eqref{eq:u*definition}. 
By Lemma~\ref{lem:Properties}(i), $u^*(k)\in\mathcal{D}_{0}^o$ for all $k\in\mathbb{R}$
and so $|\Lambda_{\varepsilon}\left(u^*(k)\right)|$ is finite for 
$\varepsilon$ small enough since $\mathcal{D}_{0}=\lim_{\varepsilon\downarrow 0} \{u\in\mathbb{R}:|\Lambda_{\varepsilon}\left(u \right)|<\infty\}$.  
Also $\D\mathbb{Q}_{k,\varepsilon}/\D\mathbb{P}$ is almost surely strictly positive and hence $\mathbb{E}\left(\D\mathbb{Q}_{k,\varepsilon}/\D\mathbb{P}\right)=1$. 
Therefore ~\eqref{eq:MeasureChange} is a valid measure change for all $k\in\mathbb{R}$. 
We define the random variable  
\begin{align}\label{eq:zvarepsilon}
Z_{k,\varepsilon}:=(Y_{\varepsilon}-k)/\sqrt{\varepsilon}
\end{align}
and set the characteristic function 
$\Phi_{Z_{k,\varepsilon}}:\mathbb{R}\to\mathbb{C}$ of $Z_{k,\varepsilon}$ in the $\mathbb{Q}_{k,\varepsilon}$-measure as follows
\begin{equation}\label{eq:PhiZ}
\Phi_{Z_{k,\varepsilon}}(u)=\mathbb{E}^{\mathbb{Q}_{k,\varepsilon}}\left(\E^{\I u Z_{k,\varepsilon}}\right).
\end{equation}
Recall from Section~\ref{sec:GeneralResults} that $\Lambda_{i}:=\Lambda_{i}(u^*(k))$ and $\Lambda_{i,l}:=\left.\partial_u^l \Lambda_{i}(u)\right|_{u=u^*(k)}$;
we first start with the following important technical lemma.

\begin{lemma}\label{lemma:CharactExp}
The expansion 
$\log\Phi_{Z_{k,\varepsilon}}(u)=-\frac{1}{2}\Lambda_{0,2} u^2 +\eta_1(u) \sqrt{\varepsilon }+\eta_2(u)\varepsilon +\eta_3(u) \varepsilon ^{3/2}+ \mathcal{O}\left(\varepsilon^2\right)$
holds as $\varepsilon\downarrow0$,
with the functions $\eta_i$, $i=1,2,3$ defined in~\eqref{eq:etadef}.
\end{lemma}
\begin{remark}\label{rem:weakconv}
By L\'evy's Convergence Theorem~\cite[Page 185, Theorem 18.1]{W03}, $Z_{k,\varepsilon}$ defined in~\eqref{eq:zvarepsilon} 
converges weakly to a normal random variable with mean 0 and variance $\Lambda_{0,2}$ in the $\mathbb{Q}_{k,\varepsilon}$-measure 
as $\varepsilon$ tends to zero.
\end{remark}
\begin{proof}
Since $\Lambda_{\varepsilon}$ is analytic~\cite[Theorem 7.1.1]{L70} on the set $\{z\in\mathbb{C}:\Re(z)\in\mathcal{D}_{0}^o\}$ for $\varepsilon$ small enough,
$\Lambda_{\varepsilon}\left(\I u\sqrt{\varepsilon} +u^*(k)\right)-\Lambda_{\varepsilon}\left(u^*(k)\right)
=\sum_{n=1}^{\infty}\frac{1}{n!}\Lambda_{\varepsilon}^{(n)}(u^*(k))\left(\I u\sqrt{\varepsilon}\right)^n$
holds as $\varepsilon\downarrow 0$.
By~\cite[Theorem 1.8.5.]{M91} the asymptotic for $\Lambda_{\varepsilon}$ in~\ref{eq:LambdaExpansion} can be differentiated with respect to $u$ due to Assumption~\ref{assump:Differentiability}(ii).
This allows us to write $\Lambda_{\varepsilon}\left(\I u\sqrt{\varepsilon} +u^*(k)\right)-\Lambda_{\varepsilon}\left(u^*(k)\right)
=\sum_{n=1}^{\infty}\frac{1}{n!}\left(\Lambda_{0,n}+\varepsilon \Lambda_{1,n} +\varepsilon^2 \Lambda_{2,n}+\mathcal{O}\left(\varepsilon^3\right)\right)\left(\I u\sqrt{\varepsilon}\right)^n$,
and hence
\begin{align*}
\log\Phi_{Z_{k,\varepsilon}}(u)
 & = \log\mathbb{E}^{\mathbb{P}}\left(\frac{\D\mathbb{Q}_{k,\varepsilon}}{\D\mathbb{P}}\E^{\I u Z_{k,\varepsilon}}\right)
=
\log\mathbb{E}^{\mathbb{P}}\left[\exp\left(\frac{u^*(k) Y_{\varepsilon}}{\varepsilon}
-\frac{\Lambda_{\varepsilon}(u^*(k))}{\varepsilon}\right)
\exp\left({\I u\sqrt{\varepsilon}\left(\frac{Y_{\varepsilon}}{\varepsilon}\right)-\frac{\I k u}{\sqrt{\varepsilon}}}\right)\right] \\ \nonumber
 & = {-\frac{1}{\varepsilon}\Lambda_{\varepsilon}\left(u^*(k)\right)-\frac{\I uk}{\sqrt{\varepsilon}}}+\log\mathbb{E}^{\mathbb{P}}\left[\exp\left({\left(\frac{Y_{\varepsilon}}{\varepsilon}\right)\left(\I u\sqrt{\varepsilon}+u^*(k)\right)}\right)\right]  \\
 & = -\frac{\I uk}{\sqrt{\varepsilon}}+\frac{1}{\varepsilon}\left(\Lambda_{\varepsilon}\left(\I u\sqrt{\varepsilon} +u^*(k)\right)-\Lambda_{\varepsilon}\left(u^*(k)\right)\right) \\
 & = -\frac{\I uk}{\sqrt{\varepsilon}}+\frac{1}{\varepsilon}\sum_{n=1}^{\infty}\frac{1}{n!}\left(\Lambda_{0,n}+\varepsilon \Lambda_{1,n} +\varepsilon^2 \Lambda_{2,n}+\mathcal{O}\left(\varepsilon^3\right)\right)\left(\I u\sqrt{\varepsilon}\right)^n.
\end{align*}



Writing out the first few terms of this expression and using~\eqref{eq:u*definition} we find that
\begin{align*}
\log\Phi_{Z_{k,\varepsilon}}(u)
 & = -\frac{\I uk}{\sqrt{\varepsilon}}+\frac{\I \Lambda_{0,1} u}{\sqrt{\varepsilon}}-\frac{1}{2}\Lambda_{0,2} u^2
 +\frac{1}{\varepsilon}\left(\sum_{n=3}^{\infty}\frac{1}{n!}\Lambda_{0,n}\left(\I u\sqrt{\varepsilon}\right)^n
+\sum_{n=1}^{\infty}\frac{1}{n!}
\left(\varepsilon \Lambda_{1,n} +\varepsilon^2 \Lambda_{2,n}+\mathcal{O}\left(\varepsilon^3\right)\right)
\left(\I u\sqrt{\varepsilon}\right)^n\right) \\
 & = -\frac{1}{2}\Lambda_{0,2} u^2 
+\sum_{n=1}^{\infty}\frac{1}{n!}\left(\left(\frac{n!}{\left(n+2\right)!}\Lambda_{0,{n+2}}\left(\I u\right)^{2}+ \Lambda_{1,n}\right) +\varepsilon \Lambda_{2,n}+\mathcal{O}\left(\varepsilon^2\right)\right)\left(\I u\sqrt{\varepsilon}\right)^n \\
 & = -\frac{1}{2}\Lambda_{0,2} u^2 +\eta_1(u) \sqrt{\varepsilon }+\eta_2(u)\varepsilon +\eta_3(u) \varepsilon ^{3/2}+ \mathcal{O}\left(\varepsilon^2\right), 
\end{align*}
where we define the functions
\begin{equation}\label{eq:etadef}
\eta_1(u) := \I u \Lambda _{1,1}-\frac{\I u^3}{6} \Lambda _{0,3},\quad
\eta_2(u) := -\frac{u^2}{2} \Lambda _{1,2}+\frac{u^4}{24} \Lambda _{0,4},\quad
\eta_3(u) := \I u \Lambda _{2,1} -\frac{\I u^3}{6} \Lambda _{1,3}+ \frac{\I u^5}{120} \Lambda _{0,5}.
\end{equation}
\end{proof}

With these preliminary results, we can now move on to the actual proof of Theorem~\ref{theorem:GeneralOptionAsymp}.
For $j=1,2,3$, let us define the functions~$g_j:\mathbb{R}_+^2\to\mathbb{R}_+$ by
$$
g_j(x,y):=\left\{ 
  \begin{array}{ll}
(x-y)^+, \quad & \text{if } j=1, \\
(y-x)^+,\quad & \text{if } j=2,\\
 \min(x,y),\quad & \text{if } j=3.
 \end{array} \right.
$$
Using the definition of the $\mathbb{Q}_{k,\varepsilon}$-measure in~\eqref{eq:MeasureChange} the option prices 
in Theorem~\ref{theorem:GeneralOptionAsymp} can be written as
\begin{align}
\mathbb{E}\left[g_j\left(\E^{Y_{\varepsilon}f(\varepsilon)},\E^{kf(\varepsilon)}\right)\right]
 & = \E^{\frac{1}{\varepsilon}\Lambda_{\varepsilon}\left(u^*(k)\right)}\mathbb{E}^{\mathbb{Q}_{k,\varepsilon}}\left[\E^{-\frac{u^*(k)}{\varepsilon}Y_{\varepsilon}}g_j\left(\E^{Y_{\varepsilon}f(\varepsilon)},\E^{kf(\varepsilon)}\right)\right]
\label{eq:optionPricesInverse} \\
 & = \E^{-\frac{1}{\varepsilon}\left[ku^*(k)-\Lambda_{\varepsilon}\left(u^*(k)\right)\right]}\mathbb{E}^{\mathbb{Q}_{k,\varepsilon}}\left[\E^{-\frac{u^*(k)}{\varepsilon}\left(Y_{\varepsilon}-k\right)}g_j\left(\E^{Y_{\varepsilon}f(\varepsilon)},\E^{kf(\varepsilon)}\right)\right]. \nonumber
\end{align}
By the expansion in Assumption~\ref{assump:Differentiability}(i) and Equality~\eqref{eq:LambdaStarRep2} we immediately have
\begin{align}\label{eq:RateFuncRep}
\exp\left({-\frac{1}{\varepsilon}\left(ku^*(k)-\Lambda_{\varepsilon}\left(u^*(k)\right)\right)}\right) 
=\exp\left(-\frac{1}{\varepsilon}
\Lambda^{*}(k)+\Lambda_{1}+\Lambda_{2}\,\varepsilon+\mathcal{O}\left(\varepsilon^2\right)\right).
\end{align}
From the definition of the random variable $Z_{k,\varepsilon}$ in~\eqref{eq:zvarepsilon} we obtain
$$ 
\mathbb{E}^{\mathbb{Q}_{k,\varepsilon}}\left[\E^{-\frac{u^*(k)}{\varepsilon}\left(Y_{\varepsilon}-k\right)}g_j\left(\E^{Y_{\varepsilon}f(\varepsilon)},\E^{kf(\varepsilon)}\right)\right]
=\E^{kf(\varepsilon)}\mathbb{E}^{\mathbb{Q}_{k,\varepsilon}}\left[\tilde{g}_j(Z_{k,\varepsilon})\right],
$$
where for $j=1,2,3$, we define the modified payoff functions~$\tilde{g}_j:\mathbb{R}\to\mathbb{R}_+$ by
$
\label{eq:gjtilde}
\tilde{g}_j(z):=\E^{-u^*(k) z/{\sqrt{\varepsilon}}}g_j(\E^{z \sqrt{\varepsilon}f(\varepsilon)},1).
$
Assuming (for now) that $\tilde{g}_j\in L^1(\RR)$, we have for any $u\in\RR$,
\begin{align*}
\left(\mathcal{F}\tilde{g}_{j}\right)(u)
:=\int_{-\infty}^{\infty}\tilde{g}_j(z)\E^{\I u z}\D z
=\int_{-\infty}^{\infty}\exp\left(-\frac{u^*(k)z}{\sqrt{\varepsilon}}\right)
g_j\left(\E^{z\sqrt{\varepsilon}f(\varepsilon)},1\right)\E^{\I u z}\D z,
\quad\text{for }j=1,2,3.
\end{align*}
For ease of notation define the function $C_{\varepsilon,k}:\mathbb{R}\to\mathbb{C}$ by
\begin{equation}\label{eq:Cdef}
C_{\varepsilon,k}(u):= \frac{\varepsilon ^{3/2} f(\varepsilon )}{\left(u^*(k)-\I u \sqrt{\varepsilon }\right) 
\left(u^*(k)-\varepsilon  f(\varepsilon )-\I u \sqrt{\varepsilon }\right)}.
\end{equation}
For $j=1$ we can write
$$
\label{eq:stripsofreg}
\int_{-\infty}^{\infty}\tilde{g}_1(z)\E^{\I u z}\D z
  = \left[\frac{\exp\Big(z\left(\sqrt{\varepsilon}f(\varepsilon)-u^*(k)/\sqrt{\varepsilon}+\I u\right)\Big)}{\sqrt{\varepsilon}f(\varepsilon)-u^*(k)/\sqrt{\varepsilon}+\I u}\right]_{0}^{\infty}
-\left[\frac{\exp\Big(z\left(-u^*(k)/\sqrt{\varepsilon}+\I u\right)\Big)}{-u^*(k)/\sqrt{\varepsilon}+\I u}\right]_{0}^{\infty}
 = C_{\varepsilon,k}(u),
$$
which is valid for $u^*(k)>\varepsilon f(\varepsilon)$. 
For $\varepsilon$ sufficiently small and by the definition of $f$ in ~\eqref{eq:LDP Rescaling} 
this holds for $u^*(k)>c$. 
For $j=2$ we can write
$$
\int_{-\infty}^{\infty}\tilde{g}_2(z)\E^{\I u z}\D z
  = \left[\frac{\exp\Big(z\left(-u^*(k)/\sqrt{\varepsilon}+\I u\right)\Big)}{-u^*(k)/\sqrt{\varepsilon}+\I u}\right]_{-\infty}^{0}
-\left[\frac{\exp\Big(z\left(\sqrt{\varepsilon}f(\varepsilon)-u^*(k)/\sqrt{\varepsilon}+\I u\right)\Big)}{\sqrt{\varepsilon}f(\varepsilon)-u^*(k)/\sqrt{\varepsilon}+\I u}\right]_{-\infty}^{0} 
 =C_{\varepsilon,k}(u),
$$
which is valid for $u^*(k)<0$ as $\varepsilon$ tends to zero. Finally, for $j=3$ we have
\begin{align*}
\int_{-\infty}^{\infty}\tilde{g}_3(z)\E^{\I u z}\D z
 & = \int_{-\infty}^{0}\E^{-\frac{u^*(k)}{\sqrt{\varepsilon}}z}g_3\left(\E^{z\sqrt{\varepsilon}f(\varepsilon)},1\right)\E^{\I u z}\D z+\int_{0}^{\infty}\E^{-\frac{u^*(k)}{\sqrt{\varepsilon}}z}g_3\left(\E^{z\sqrt{\varepsilon}f(\varepsilon)},1\right)\E^{\I u z}\D z \\
 & = \left[\frac{\exp\Big(z\left(\sqrt{\varepsilon}f(\varepsilon)-u^*(k)/\sqrt{\varepsilon}+\I u\right)\Big)}
{\sqrt{\varepsilon}f(\varepsilon)-u^*(k)/\sqrt{\varepsilon}+\I u}\right]_{-\infty}^{0}
+\left[\frac{\exp\Big(z\left(-u^*(k)/\sqrt{\varepsilon}+\I u\right)\Big)}{-u^*(k)/\sqrt{\varepsilon}+\I u}\right]_{0}^{\infty} 
  = -C_{\varepsilon,k}(u),
\end{align*}
which is valid for $0<u^*(k)<\varepsilon f(\varepsilon)$. 
For $\varepsilon$ sufficiently small and by the assumption on $f$ in ~\eqref{eq:LDP Rescaling} 
this is true for $0<u^*(k)<c$. 
In this context $u^*(k)$ comes out naturally in the analysis as a classical dampening factor. 
Note that in order for these strips of regularity to exist we require that $\{0,c\}\subset\mathcal{D}_{0}^o$, 
as assumed in the theorem. 
By the strict convexity 
and essential smoothness property in Assumption~\ref{assump:Differentiability}(iv) we have
\begin{equation}\label{eq:altdomrep}
\begin{array}{rll}
0<u^*(k)<c \quad & \text{if and only if}\quad & \Lambda_{0,1}(0)<k<\Lambda_{0,1}(c), \\
u^*(k)<0 \quad & \text{if and only if}\quad & k<\Lambda_{0,1}(0),\\
u^*(k)>c \quad & \text{if and only if}\quad & k>\Lambda_{0,1}(c).
\end{array}
\end{equation}
The following technical lemma allows us to write the transformed option price as an inverse Fourier transform.
\begin{lemma}\label{lem:optpricerep0}
There exists $\varepsilon^*_1>0$ such that for all $\varepsilon<\varepsilon^*_1$ and all $k\in\mathbb{R}\backslash\{\Lambda_{0,1}(0), \Lambda_{0,1}(c)\}$, we have
($\bar{a}$ denoting the complex conjugate of $a\in\mathbb{C}$)
\begin{align}\label{eq:Pars}
\mathbb{E}^{\mathbb{Q}_{k,\varepsilon}}\left[\tilde{g}_j(Z_{k,\varepsilon})\right]
= \begin{dcases*}
       \frac{1}{2\pi}\int_{\RR} \Phi_{Z_{k,\varepsilon}}(u)\overline{C_{\varepsilon,k}(u)} \D u,  & if $j=1,u^*(k)>c$,\\
       \frac{1}{2\pi}\int_{\RR}\Phi_{Z_{k,\varepsilon}}(u)\overline{C_{\varepsilon,k}(u)} \D u,  & if $j=2, u^*(k)<0$, \\
  -\frac{1}{2\pi}\int_{\RR}\Phi_{Z_{k,\varepsilon}}(u)\overline{C_{\varepsilon,k}(u)} \D u,  & if $j=3, 0<u^*(k)<c. $
        \end{dcases*}
\end{align}
\end{lemma}
We note in passing that 
\begin{equation}\label{eq:Cconj}
\overline{C_{\varepsilon,k}(u)} =\frac{\varepsilon ^{3/2} f(\varepsilon )}{\left(u^*(k)+\I u \sqrt{\varepsilon }\right) \left(u^*(k)-\varepsilon  f(\varepsilon )+\I u \sqrt{\varepsilon }\right)}.
\end{equation}
We now consider the integral appearing in Lemma~\ref{lem:optpricerep0}.
For $\varepsilon>0$ small enough, we can split the integral as
\begin{align}\label{eq:IntExp}
\int_{\RR} \Phi_{Z_{k,\varepsilon}}(u) \overline{C_{\varepsilon,k}(u)} \D u 
& = \int_{|u|< 1/\sqrt{\varepsilon}} \Phi_{Z_{k,\varepsilon}}(u) \overline{C_{\varepsilon,k}(u)} \D u 
 + \int_{|u|\geq 1/\sqrt{\varepsilon}} \Phi_{Z_{k,\varepsilon}}(u) \overline{C_{\varepsilon,k}(u)} \D u \\
& = \int_{|u|< 1/\sqrt{\varepsilon}} \exp\left(-\frac{\Lambda_{0,2}u^2}{2}\right) H(\varepsilon,u) 
\E^{\mathcal{O}(\varepsilon^2)}\D u
 + \mathcal{O}\left(\E^{-\beta/\varepsilon}\right),\nonumber
\end{align}
for some $\beta>0$ by Lemma~\ref{lemma:tailestimatesteepcase}, 
and using also Lemma~\ref{lemma:CharactExp} for the first integral.
The function $H:\mathbb{R}_{+}\times\mathbb{R}\to\mathbb{C}$ is defined as
$H(\varepsilon,u)
:=
\exp\left[\eta_1(u) \sqrt{\varepsilon }+\eta_2(u)\varepsilon +\eta_3(u) \varepsilon ^{3/2}\right]
\overline{C_{\varepsilon,k}(u)}$,
with $\eta_i$ ($i=1,2,3$) defined in~\eqref{eq:etadef}.
A Taylor expansion of $H$ around $\varepsilon=0$ for $c=0$ yields
$$
H(\varepsilon,u)
 = \frac{f(\varepsilon)\varepsilon^{3/2}}{u^*(k)^2}\left(1+\bar{h}_1(u,0)\sqrt{\varepsilon}+\bar{h}_2(u,0)\varepsilon+\bar{h}_3(u,0)\varepsilon^{3/2}+\frac{\varepsilon  f(\varepsilon )}{u^*(k)}+\left(\frac{\eta _1(u)}{u^*(k)}-\frac{3 \I u}{u^*(k)^2}\right)\varepsilon ^{3/2} f(\varepsilon) + \mathcal{O}\left(\varepsilon^2\right)\right),
$$
where we define the following functions:
\begin{align} \label{eq:hdefs}
h_1(u,c) & := \frac{\I u}{u^*(k)-c}\left(\frac{c}{u^*(k)}-2\right),\qquad 
h_2(u,c):=-\frac{u^2 \left(c^2-3 c u^*(k)+3 u^*(k)^2\right)}{u^*(k)^2 \left(u^*(k)-c\right)^2}, \\ \nonumber
h_3(u,c) & := \frac{\I u^3 \left(4 u^*(k)^3-c^3+4 c^2 u^*(k)-6 c u^*(k)^2\right)}{u^*(k)^3 \left(u^*(k)-c\right)^3}, \qquad
\bar{h}_1(u,c) := \eta _1(u)+h_1(u,c), \\ \nonumber
 \bar{h}_2(u,c)
 & := \frac{\eta _1^2(u)}{2}+\eta _2(u)+h_2(u,c)+\eta_1(u)h_1(u,c),
\\ \nonumber
\bar{h}_3(u,c) & := h_2(u,c)  \eta_1(u)  +h_1(u,c) \left(\frac{\eta _1^2(u)}{2}+\eta _2(u)\right)
+ \frac{\eta _1^3(u)}{6}+\eta _2(u) \eta _1(u)+\eta _3(u) +h_3(u,c),
\end{align}
with the $\eta_i$ for $i=1,2,3$, defined in~\eqref{eq:etadef}. 
A (tedious) Taylor expansion of $H$ around $\varepsilon=0$ for $c>0$ yields
\begin{align*}
H(\varepsilon,u) & = \frac{c \sqrt{\varepsilon }}{u^*(k) \left(u^*(k)-c\right)}
\left\{1+\bar{h}_1(u,c)\sqrt{\varepsilon}+\bar{h}_2(u,c)\varepsilon+\bar{h}_3(u,c)\varepsilon^{3/2}+\frac{u^*(k) (\varepsilon  f(\varepsilon )-c)}{c \left(u^*(k)-c\right)}\right. \\
&\left.+\frac{u^*(k) \sqrt{\varepsilon}(\varepsilon  f(\varepsilon )-c) }{c \left(u^*(k)-c\right)}
\left(\eta _1(u)-\frac{2 \I u}{u^*(k)-c}\right)+\mathcal{O}\left(\varepsilon^2\right)\right\}.
\end{align*}
We will shortly be integrating $H$ against a zero-mean Gaussian characteristic function over $\RR$ 
and as such all odd powers of $u$ will have a null contribution. 
In particular note that the polynomials
\begin{align*}
\eta_1,\qquad
\bar{h}_1,\qquad 
\bar{h}_3, \qquad 
\left(\frac{\eta _1(u)}{u^*(k)}-\frac{3 \I u}{\left(u^*(k)\right)^2}\right)\varepsilon ^{3/2} f(\varepsilon) 
\qquad \text{and}\qquad
\frac{u^*(k) \sqrt{\varepsilon } (\varepsilon  f(\varepsilon )-c) }{c \left(u^*(k)-c\right)}\left(\eta _1(u)-\frac{2 \I u}{u^*(k)-c}\right)
\end{align*}
are odd functions of $u$ and hence have zero contribution. 
The major quantity is $\bar{h}_2$, which we can rewrite as
$\bar{h}_2(u,c)=\bar{h}_{2,1}(c)u^2+\bar{h}_{2,2}(c)u^4-\frac{1}{72}\Lambda _{0,3}^2u^6$,
where
$$
\bar{h}_{2,1}(c)
 := -\frac{h_1(u,c) \Lambda_{1,1}}{\I}-\frac{\Lambda _{1,1}^2+\Lambda _{1,2}}{2}+h_2(1,c)
\qquad\text{and}\qquad
\bar{h}_{2,2}(c)
 := \frac{h_1(u,c) \Lambda_{0,3}}{6\I} + \frac{\Lambda _{1,1} \Lambda _{0,3}}{6}+\frac{\Lambda _{0,4}}{24}.
$$
Let
$$\phi_\varepsilon(c)\equiv \frac{c \sqrt{\varepsilon }\ind_{\{c>0\}}+\varepsilon^{3/2} f(\varepsilon )\ind_{\{c=0\}}}{ u^*(k) \left(u^*(k)-c\right)}.$$
Using simple properties of moments of a Gaussian random variable we finally compute the following
\begin{align*}
&\int_{|u|<1/{\sqrt{\varepsilon}}} \exp\left(-\frac{\Lambda_{0,2}u^2}{2}\right) H(\varepsilon,u)\E^{\mathcal{O}(\varepsilon^2)} \D u 
= 
\int_{|u|<1/{\sqrt{\varepsilon}}} \exp\left(-\frac{\Lambda_{0,2}u^2}{2}\right) H(\varepsilon,u)(1+\mathcal{O}(\varepsilon^2)) \D u
\\ \nonumber
&=\phi_\varepsilon(c) \left[
\int_{|u|<1/{\sqrt{\varepsilon}}}\E^{-\frac{1}{2}\Lambda_{0,2} u^2}
\left(1+\bar{h}_{2}(u,c)
+\frac{u^*(k) (\varepsilon  f(\varepsilon )-c)}{c \left(u^*(k)-c\right)}\ind_{\{c>0\}}+\frac{\varepsilon  f(\varepsilon )}{u^*(k)}\ind_{\{c=0\}}
\right) \D u +\mathcal{O}(\varepsilon ^2) \right]\\ \nonumber
&=\phi_\varepsilon(c) \left[
\int_{\RR}\E^{-\frac{1}{2}\Lambda_{0,2} u^2}
\left(1+\bar{h}_{2}(u,c)
+\frac{u^*(k) (\varepsilon  f(\varepsilon )-c)}{c \left(u^*(k)-c\right)}\ind_{\{c>0\}}+\frac{\varepsilon  f(\varepsilon )}{u^*(k)}\ind_{\{c=0\}}
\right) \D u +\mathcal{O}(\varepsilon ^2) \right]\\ \nonumber
&=
\phi_\varepsilon(c)\sqrt{\frac{2\pi}{\Lambda_{0,2}}}
\left(1+\frac{\bar{h}_{2,1}(c)}{\Lambda_{0,2}}+\frac{3\bar{h}_{2,2}(c)}{\Lambda_{0,2}^2}-\frac{5\Lambda _{0,3}^2}{24\Lambda_{0,2}^3}+\frac{u^*(k) (\varepsilon  f(\varepsilon )-c)}{c \left(u^*(k)-c\right)}\ind_{\{c>0\}}+\frac{\varepsilon  f(\varepsilon )}{u^*(k)}\ind_{\{c=0\}}+\mathcal{O}(\varepsilon ^2)\right).
\end{align*}
The third line follows from the Laplace method: in particular the tail estimate ($|u|>1/\sqrt{\varepsilon}$) for an integral over the Gaussian characteristic function is exponentially small,
and hence is absorbed in the $\mathcal{O}(\varepsilon^2)$ term. 
Combining this with~\eqref{eq:IntExp}, Lemma~\ref{lem:optpricerep0} and ~\eqref{eq:RateFuncRep} with the property~\eqref{eq:altdomrep}, the theorem follows.

\subsubsection{Proof of the forward implied volatility expansions (Propositions~\ref{Prop:GeneralBSFwdVolShortTime},~\ref{Prop:GeneralBSFwdVolLargeTime})}
Gao and Lee~\cite{GL11} have obtained representations for asymptotic implied volatility for small and large-maturity regimes in terms of the assumed asymptotic behaviour of certain option prices, outlining the general procedure for transforming option price asymptotics into implied volatility asymptotics. 
The same methodology can be followed to transform our  forward-start option asymptotics (Corollary~\ref{cor:ShortTimeAsymp} and Corollary~\ref{cor:LargeTimeAsymp}) into  forward smile asymptotics. 
In the proofs of Proposition~\ref{Prop:GeneralBSFwdVolShortTime} and~Proposition~\ref{Prop:GeneralBSFwdVolLargeTime} 
we hence assume for brevity the existence of an ansatz for the forward smile asymptotic and solve for the coefficients. 
We refer the reader to~\cite{GL11} for the complete methodology.

\begin{proof}[Proof of Proposition~\ref{Prop:GeneralBSFwdVolShortTime}]
 Using $\Lambda_{0,1}(0)=0$ and substituting the ansatz 
$\sigma^2_{\varepsilon t,\varepsilon \tau}(k)=v_0(k,t,\tau)+v_1(k,t,\tau) \varepsilon +v_2(k,t,\tau) \varepsilon ^2+\mathcal{O}\left(\varepsilon^3\right)$
into Corollary~\ref{Cor:BSOptionSmallTime}, we get that forward-start option prices have the asymptotics
\begin{align*}
&\mathbb{E}\left(\E^{X^{\left(\varepsilon t\right)}_{\varepsilon\tau}}-\E^k\right)^+\ind_{\{k>0\}} +
\mathbb{E}\left(\E^k-\E^{X^{\left(\varepsilon t\right)}_{\varepsilon\tau}}\right)^+\ind_{\{k<0\}} \\
&=\exp\left({-\frac{k^2}{2 \tau  v_0(k,t,\tau) \varepsilon }+\frac{k^2 v_1(k,t,\tau)}{2 \tau  v_0(k,t,\tau)^2}+\frac{k}{2}}\right)
\frac{ \left(v_0(k,t,\tau) \varepsilon \tau\right)^{3/2}}{k^2\sqrt{2\pi }}\left(1+\gamma(k,t,\tau) \varepsilon +\mathcal{O}\left(\varepsilon^2\right)\right) , 
\end{align*}
for all $k\ne 0$, where we set
$$\gamma(k,t,\tau):=-\tau v_0(k,t,\tau)\left(\frac{3}{k^2}+\frac{1}{8}\right)+\frac{k^2v_2(k,t,\tau)}{2\tau v_0(k,t,\tau)^2}-\frac{k^2 v_1(k,t,\tau)^2}{2\tau v_0(k,t,\tau)^3}+\frac{3v_1(k,t,\tau)}{2v_0(k,t,\tau)}.$$ 
The result follows after equating orders with the general formula in Corollary~\ref{cor:ShortTimeAsymp}.
\end{proof}

\begin{proof}[Proof of Proposition~\ref{Prop:GeneralBSFwdVolLargeTime}]
Substituting the ansatz 
\begin{equation} \label{eq:largematansatz}
\sigma_{t,\tau}^2(k)=v_0^{\infty}(k,t)+v_1^{\infty}(k,t)/\tau+v_2^{\infty}(k,t)/\tau^2
+\mathcal{O}\left(1/\tau^3\right),
\end{equation}
into Corollary~\ref{Cor:BSOptionLargeTime} 
we obtain the following asymptotic expansions for forward-start options:
\begin{align*}
&\mathbb{E}\left(\E^{X_{\tau}^{(t)}}-\E^{k\tau}\right)^+\ind_{\textrm{A}}
-\mathbb{E}\left(\E^{X_{\tau}^{(t)}}\wedge\E^{k\tau}\right)\ind_{\textrm{B}}
+\mathbb{E}\left(\E^{k\tau}-\E^{X_{\tau}^{(t)}}\right)^+\ind_{\textrm{C}} \\
&=\exp \left(-\tau  \left(\frac{k^2}{2 v_0(k,t)}-\frac{k}{2}+\frac{v_0(k,t)}{8}\right)
+\frac{v_1(k,t)k^2}{2v_0(k,t)^2}-\frac{v_1(k,t)}{8}\right) \\ 
&\frac{4 \tau^{-1/2} v_0(k,t)^{3/2}}{\left(4 k^2-v_0(k,t)^2\right)\sqrt{2\pi}}
\left(1+\frac{\gamma^{\infty}(k,t)}{\tau}+\mathcal{O}\left(\frac{1}{\tau^2}\right)\right),
\end{align*}
for all $k\in\mathbb{R}\setminus\{\Lambda_{0,1}(0),\Lambda_{0,1}(1)\}$, where
\begin{equation}\label{eq:largematdoms}
\textrm{A}:=\left\{k>\frac{1}{2}\sigma_{t,\tau}^2(k)\right\},\qquad 
\textrm{B}:=\left\{-\frac{1}{2}\sigma_{t,\tau}^2(k)<k<\frac{1}{2}\sigma_{t,\tau}^2(k)\right\},\qquad 
\textrm{C}:=\left\{k<-\frac{1}{2}\sigma_{t,\tau}^2(k)\right\},
\end{equation}
$$
\gamma^{\infty}(k,t) 
 := \frac{\left(12 k^2+v_0^2(k,t)\right) \left(4 k^2 v_1(k,t)-v_0^2(k,t)
\left(v_1(k,t)+8\right)\right)}
{2 v_0(k,t) \left(v_0^2(k,t)-4k^2\right)^2}
-\frac{v_1^2(k,t)k^2}{2 v_0^3(k,t)}+\frac{v_2(k,t)k^2}{2 v_0^2(k,t)}-\frac{v_2(k,t)}{8}.
$$
We obtain the expressions for $v_1^{\infty}$ and $v_2^{\infty}$ 
by equating orders with the formula in Corollary~\ref{cor:LargeTimeAsymp}. 
Choosing the correct root for the zeroth order term $v_0^{\infty}$ is now classical in this literature (this is an argument by contradiction), 
and we refer the reader to~\cite{FJM10} for details.
\end{proof}

\subsection{Proofs of Section~\ref{sec:HestonForwardSmile}}
We now let $(X_t)_{t\geq0}$ be the Heston process satisfying the SDE~\eqref{eq:Heston}. 
The tower property for expectations yields the forward lmgf:
\begin{equation}\label{eq:LambdaTau}
\log\mathbb{E}\left(\E^{u X_{\tau}^{(t)}}\right)
=
A(u,\tau)+\frac{B(u,\tau)}{1-2\beta_t B(u,\tau)}v\E^{-\kappa t}
-\frac{2\kappa\theta}{\xi^2}\log\left(1-2\beta_t B(u,\tau)\right),
\end{equation}
defined for all $u$ such that the rhs exists and where 
\begin{equation}\label{eq:ABfunctions}
\begin{array}{rl}
A(u,\tau) & := \displaystyle
\frac{\kappa\theta}{\xi^2}\left(\left(\kappa-\rho\xi u- d(u)\right)\tau-2\log\left(\frac{1-\gamma(u)\exp\left(-d(u)\tau\right)}{1-\gamma(u)}\right)\right),\\
B(u,\tau) & := \displaystyle
\frac{\kappa-\rho\xi u-d(u)}{\xi^2}\frac{1-\exp\left(-d(u)\tau\right)}{1-\gamma(u)\exp\left(-d(u)\tau\right)},
\end{array}
\end{equation}
and $d$, $\gamma$ and $\beta$ were introduced in~\eqref{eq:DGammaBeta}.
In the next two subsections we develop the tools needed to apply Propositions~\ref{Prop:GeneralBSFwdVolShortTime} and~\ref{Prop:GeneralBSFwdVolLargeTime} to the Heston model. 

\subsubsection{Proofs of Section~\ref{sec:DiagSmallMatHeston}} \label{sec:proofdiagsmallmat}

We consider here the Heston diagonal small-maturity process 
$(X^{\left(\varepsilon t\right)}_{\varepsilon\tau})_{\varepsilon>0}$ 
with $X$ defined in~\eqref{eq:Heston} and $(X_{\tau}^{(t)})_{\tau>0}$ in~\eqref{eq:XtTauDef}. 
The forward rescaled lmgf $\Lambda_{\varepsilon}$ in~\eqref{eq:Renorm-mgf} 
is easily determined from~\eqref{eq:LambdaTau}.

In this subsection, we prove Proposition~\ref{Proposition:HestonDiagonal}.
For clarity, the proof is divided into the following steps:
\begin{enumerate}[(i)]
\item In Lemma~\ref{lemma:hestonsmallmatdomain} we show that $\mathcal{D}_{0}=\mathcal{K}_{t,\tau}$ and $0\in\mathcal{D}_{0}^o$;
\item  In Lemma~\ref{lemma:smalltimehestonexp} we show that the Heston diagonal small-maturity process has an expansion of the form given in Assumption~\ref{assump:Differentiability} with $\LO=\Xi$ and $\Lambda_{1}=L$,
where $\Xi$ and $L$ are defined in~\eqref{eq:HestonDiagZeroOrder} and~\eqref{eq:HestonDiag1stOrder};
\item In Lemma~\ref{lemma:Smalltimehestonessentsmooth} we show that $\Xi$ 
is strictly convex and essentially smooth on $\mathcal{D}_{0}^o$, i.e. Assumption~\ref{assump:Differentiability}(iv);
\item The map $(\varepsilon,u)\mapsto \Lambda_{\varepsilon}(u)$
is of class $\mathcal{C}^{\infty}$ on $\mathbb{R}_{+}^*\times\mathcal{D}_{0}^{o}$, $\Lambda_{0,1}(0)=0$ and Assumption~\ref{assump:Differentiability}(v) is also satisfied.
\end{enumerate}

\begin{lemma}\label{lemma:hestonsmallmatdomain}
For the Heston diagonal small-maturity process we have $\mathcal{D}_{0}=\mathcal{K}_{t,\tau}$ 
and $0\in\mathcal{D}_{0}^o$ with $\mathcal{K}_{t,\tau}$ defined in~\eqref{eq:HestonDiagZeroOrder} 
and  $\mathcal{D}_{0}$ defined in Assumption~\ref{assump:Differentiability}.
\end{lemma}
\begin{proof}
For any $t>0$, the random variable $V_t$ in~\eqref{eq:Heston} is distributed as 
$\beta_t$ times a non-central chi-square random variable 
with $q=4\kappa\theta/\xi^2>0$ degrees of freedom 
and non-centrality parameter $\lambda=v \E^{-\kappa t}/\beta_t>0$. 
It follows that the corresponding lmgf is given by
\begin{equation}\label{eq:HestonVarianceMGF}
\Lambda_{t}^{V}(u)
:=\mathbb{E}\left(\E^{uV_t}\right)
=\exp\left({\frac{\lambda \beta_tu}{1-2\beta_tu}}\right)\left(1-2\beta_tu\right)^{-q/2},
\qquad\text{for all }u<\frac{1}{2\beta_t}.
\end{equation}
The re-normalised Heston forward lmgf $\Lambda_\varepsilon$ is then computed as
$$
\E^{\Lambda_{\varepsilon}(u)/\varepsilon}
=\mathbb{E}\left[\E^{\frac{u}{\varepsilon}\left(X_{\varepsilon t+\varepsilon\tau}-X_{\varepsilon t}\right)}\right]
= \mathbb{E}\left[\mathbb{E}\left(\E^{\frac{u}{\varepsilon}\left(X_{\varepsilon t+\varepsilon\tau}
-X_{\varepsilon t}\right)}|\mathcal{F}_{\varepsilon t}\right)\right]
=\mathbb{E}\left(\E^{A\left(\frac{u}{\varepsilon},\varepsilon\tau\right)
+B\left(\frac{u}{\varepsilon},\varepsilon\tau\right)V_{\varepsilon t}}\right)
=\E^{A\left(\frac{u}{\varepsilon},\varepsilon\tau\right)}
\Lambda_{\varepsilon t}^{V}\left(B\left(u/\varepsilon,\varepsilon\tau\right)\right),\nonumber
$$
which agrees with~\eqref{eq:LambdaTau}.
This only makes sense in some effective domain~$\mathcal{K}_{\varepsilon t, \varepsilon\tau}\subset\mathbb{R}$.
The lmgf for $V_{\varepsilon t}$ is well-defined in
$\mathcal{K}_{\varepsilon t}^{V}
:=\{u\in\mathbb{R}:B\left(u/\varepsilon,\varepsilon\tau\right)<\frac{1}{2\beta_{\varepsilon t}}\}$, 
and hence 
$\mathcal{K}_{\varepsilon t,\varepsilon \tau}
=\mathcal{K}_{\varepsilon t}^{V} \cap \mathcal{K}_{\varepsilon \tau}^{H}$, 
where
$\mathcal{K}_{\varepsilon \tau}^{H}$ is the effective domain of the (spot) Heston lmgf.
Consider first $\mathcal{K}_{\varepsilon \tau}^{H}$ for small $\varepsilon$. 
From~\cite[Proposition 3.1]{AP07} if $\xi ^2 (u/\varepsilon-1) u/\varepsilon>(\kappa -\xi  \rho  u/\varepsilon)^2$ 
then the explosion time~$\tau_H^*(u):=\sup\{t\geq 0:\mathbb{E}(\E^{u X_{t}})<\infty\}$ of the Heston lmgf is
$$\tau^*_H\left(\frac{u}{\varepsilon}\right)
=\frac{2}{\sqrt{\xi ^2 (u/\varepsilon-1) u/\varepsilon-(\kappa-\rho\xi  u/\varepsilon)^2}}
\left(\pi\ind_{\{\rho\xi u/\varepsilon-\kappa<0\}}
+\arctan\left(\frac{\sqrt{\xi ^2 (u/\varepsilon-1) u/\varepsilon-(\kappa-\rho\xi  u/\varepsilon)^2}}{\rho\xi u/\varepsilon-\kappa }\right)\right).$$
Recall the following Taylor series expansions, for $x$ close to zero:
\begin{equation*}
\left.
\begin{array}{rll}
\displaystyle\arctan\left(\frac{1}{\rho\xi u/x-\kappa }\sqrt{\xi^2 \left(\frac{u}{x}-1\right)
 \frac{u}{x}-\left(\kappa -\xi  \rho  \frac{u}{x}\right)^2}\right)
 & =\displaystyle\sgn(u)\arctan\left(\frac{\bar{\rho}}{  \rho  }\right)+\mathcal{O}\left(x\right),
\quad & \text{if } \rho\ne 0,\\
\displaystyle\arctan\left(-\frac{1}{\kappa}\sqrt{\xi ^2 \left(\frac{u}{x}-1\right) \frac{u}{x}-\kappa^2}\right)
 & =\displaystyle-\frac{\pi}{2}+\mathcal{O}(x),
\quad & \text{if } \rho = 0.
\end{array}
\right.
\end{equation*}
As $\varepsilon$ tends to zero 
$\xi ^2 (u/\varepsilon-1) u/\varepsilon>(\kappa-\rho\xi u/\varepsilon)^2$ is satisfied since
$\xi^2>\xi^2\rho^2$ and hence
$$
\tau^*_H\left(u/\varepsilon\right)=\left\{ 
  \begin{array}{l l}
\displaystyle   \frac{\varepsilon}{\xi |u|}\left(\pi\ind_{\{\rho=0\}}
+\frac{2}{\bar{\rho}}\left(\pi\ind_{\{\rho u \leq0\}}+ \sgn(u)\arctan\left(\frac{\bar{\rho}}{\rho}\right)\right)\ind_{\{\rho\neq0\}}+\mathcal{O}(\varepsilon)\right), 
& \quad \text{if } u \neq 0,\\
  \infty, & \quad \text{if }u=0.
    \\
\end{array} \right.
$$
Therefore, for $\varepsilon$ small enough, we have $\tau^*_H\left(\frac{u}{\varepsilon}\right)>\varepsilon\tau$ 
for all $u\in\left(u_{-},u_{+}\right)$, where
\begin{align*}
u_{-}
 & :=\frac{2}{\bar{\rho}\xi\tau}\arctan\left(\frac{\bar{\rho}}{\rho}\right)\ind_{\{\rho<0\}}
-\frac{\pi}{\xi\tau}\ind_{\{\rho=0\}}
+\frac{2}{\bar{\rho}\xi\tau}\left(\arctan\left(\frac{\bar{\rho}}{\rho}\right)-\pi\right)\ind_{\{\rho>0\}},\\
u_{+}
 & :=\frac{2}{\bar{\rho}\xi\tau}\left(\arctan\left(\frac{\bar{\rho}}{\rho}\right)
+\pi\right)\ind_{\{\rho<0\}}+\frac{\pi}{\xi\tau}\ind_{\{\rho=0\}}
+\frac{2}{\bar{\rho}\xi\tau}\arctan\left(\frac{\bar{\rho}}{\rho}\right)\ind_{\{\rho>0\}}.
\end{align*}
So as $\varepsilon$ tends to zero, $\mathcal{K}_{\varepsilon \tau}^{H}$ shrinks to $(u_{-},u_{+})$.
Regarding $\mathcal{K}_{\varepsilon t}^{V}$, 
we have (see~\eqref{eq:BAsymptotics} for details on the expansion computation)
$\beta_{\varepsilon t}B(u/\varepsilon,\varepsilon\tau)=\frac{\xi^2t}{4v}\Xi(u,0,\tau)+\mathcal{O}(\varepsilon)$ for any $u\in\left(u_{-},u_{+}\right)$,
with $\Xi$ defined in~\eqref{eq:HestonDiagZeroOrder}. 
Therefore $\lim_{\varepsilon\downarrow 0}\mathcal{K}_{\varepsilon t}^{V}=\{u\in\mathbb{R}:\Lambda(u,0,\tau)<\frac{2v}{\xi^2t}\}$ 
and hence 
$\lim_{\varepsilon\downarrow 0}\mathcal{K}_{\varepsilon t,\varepsilon \tau}
=\{u\in\mathbb{R}:\Xi(u,0,\tau)<\frac{2v}{\xi^2t}\}\cap\left(u_{-},u_{+}\right)$. 
It is easily checked that $\Xi(u,0,\tau)$ is strictly positive except at $u=0$ where it is zero, $\Xi'(u,0,\tau)>0$ for $u>0$, $\Xi'(u,0,\tau)<0$ for $u<0$ and that $\Xi(u,0,\tau)$ tends to infinity as $u$ approaches $u_{\pm}$. 
Since $v$ and $\xi$ are strictly positive and $t\geq0$ it follows that 
$\{u\in\mathbb{R}:\Xi(u,0,\tau)<2v/(\xi^2t)\}\subseteq (u_{-},u_{+})$ with equality only if $t=0$.
So $\mathcal{D}_{0}$ is an open interval around zero and the lemma follows
with $\mathcal{D}_{0}=\mathcal{K}_{t,\tau}$.
\end{proof}

\begin{remark}\label{remark:lambdapositive}
For $u\in\mathbb{R}^*$ the inequality 
$0<\Xi(u,0,\tau)<2v/(\xi^2 t)$ is equivalent to $\Xi(u,t,\tau)\in(0,\infty)$. 
In Lemma~\ref{lemma:smalltimehestonexp} below we show that $\Xi$ is the limiting lmgf of the rescaled Heston forward lmgf and so the condition for the limiting forward domain is equivalent to ensuring that the limiting forward lmgf does not blow up and is strictly positive except at $u=0$ where it is zero.
\end{remark}

\begin{lemma}\label{lemma:smalltimehestonexp}
For any $t\geq0$, $\tau>0$, $u\in\mathcal{K}_{t,\tau}$, the expansion 
$\Lambda_{\varepsilon}(u)=\Xi(u,t,\tau)+L(u,t,\tau)\varepsilon+\mathcal{O}\left(\varepsilon^2\right)$
holds as $\varepsilon$ tends to zero,
where $\mathcal{K}_{t,\tau}$, $\Xi$ and $L$ are defined in~\eqref{eq:HestonDiagZeroOrder},~\eqref{eq:HestonDiagZeroOrder} and~\eqref{eq:HestonDiag1stOrder} and $\Lambda_{\varepsilon}$ is the rescaled lmgf in Assumption~\ref{assump:Differentiability} 
for the Heston diagonal small-maturity process $(X^{(\varepsilon t)}_{\varepsilon\tau})_{\varepsilon>0}$.
\end{lemma}
\begin{remark}\label{remar:TisReal}
For any $u\in\mathcal{K}_{t,\tau}$, Lemma~\ref{lemma:hestonsmallmatdomain} implies that 
$\Lambda_{\varepsilon}(u)$ is a real number for any $\varepsilon>0$. 
Therefore $L$ defined in~\eqref{eq:HestonDiag1stOrder} and used in Lemma~\ref{lemma:smalltimehestonexp} 
is a real-valued function on $\mathcal{K}_{t,\tau}$. 
\end{remark}

\begin{proof}
All expansions below for $d$, $\gamma$ and $\beta_t$ defined in~\eqref{eq:DGammaBeta} hold for any $u\in\mathcal{K}_{t,\tau}$:
\begin{align}
d\left(u/\varepsilon\right)
 & =\frac{1}{\varepsilon}\left(\kappa^2\varepsilon^2+u\varepsilon\left(\xi-2\kappa\rho\right)-u^2\xi^2(1-\rho^2)\right)^{1/2}
 = \frac{\I u}{\varepsilon}d_0 + d_1 + \mathcal{O}(\varepsilon), \nonumber \\ 
\gamma\left(u/\varepsilon\right)
 & =\frac{\kappa \varepsilon-\rho\xi u-\I u d_0 - d_1\varepsilon + \mathcal{O}\left(\varepsilon^2\right)}{\kappa \varepsilon-\rho\xi u +\I u d_0 + d_1\varepsilon + \mathcal{O}\left(\varepsilon^2\right)}
= g_0-\frac{\I\varepsilon}{u}g_1+\mathcal{O}\left(\varepsilon^2\right), \nonumber\\ 
\beta_{\varepsilon t}
 & =\frac{1}{4}\xi^2 t\varepsilon-\frac{1}{8}\kappa\xi^2 t^2 \varepsilon^2+\mathcal{O}\left(\varepsilon^3\right),
\label{eq:BetaAsymp}
\end{align}
where
\begin{align}
d_0:= \bar{\rho}\xi\, \sgn(u),\qquad 
d_1:=\frac{ \I\left(2 \kappa  \rho -\xi \right)\sgn(u)}{2 \bar{\rho}},\qquad
g_0:= \frac{\I \rho -\bar{\rho}\,\sgn(u)}{\I\rho+\bar{\rho}\,\sgn(u)}\qquad
g_1:= \frac{\left(2 \kappa -\xi  \rho\right) \sgn(u) }{\xi  \bar{\rho}\left(\bar{\rho }+\I \rho\,\sgn(u) \right)^2},
\end{align}
and where $\sgn(u)=1$ if $u\geq 0$, $-1$ otherwise.
From the definition of $A$ in~\eqref{eq:ABfunctions} we obtain
\begin{equation} \label{eq:AAsymp}
A\left(\frac{u}{\varepsilon},\varepsilon\tau\right) 
  = 
\frac{\kappa\theta}{\xi^2}\left(\left(\kappa-\rho\xi u/\varepsilon- d(u/\varepsilon)\right)\varepsilon\tau-2\log\left(\frac{1-\gamma(u/\varepsilon)
\exp\left(-d(u/\varepsilon)\varepsilon\tau\right)}{1-\gamma(u/\varepsilon)}\right)\right)
=L_0(u,\tau)+\mathcal{O}(\varepsilon),
\end{equation}
where $L_0$ is defined in~\eqref{eq:HestonDiag1stOrder}. Substituting the asymptotics for $d$ and $\gamma$ above 
we further obtain
$$
\frac{1-\exp\left(-d(u/\varepsilon)\varepsilon\tau\right)}
{1-\gamma(u/\varepsilon)\exp\left(-d(u/\varepsilon)\varepsilon\tau\right)}
 = \frac{1-\exp\left(-\I u d_0 \tau - \varepsilon d_1 \tau +\mathcal{O}(\varepsilon^2)\right)}{1-\left(g_0-\I\varepsilon g_1/u+\mathcal{O}(\varepsilon^2)\right)
\exp\left(-\I u d_0 \tau - \varepsilon d_1 \tau +\mathcal{O}(\varepsilon^2)\right)},
$$
and therefore using the definition of $B$ in~\eqref{eq:ABfunctions} we obtain
\begin{equation}\label{eq:BAsymptotics}
B\left(\frac{u}{\varepsilon},\varepsilon\tau\right)
= \frac{\kappa-\rho\xi u/\varepsilon-d(u/\varepsilon)}{\xi^2}\frac{1-\exp\left(-d\left(u/\varepsilon\right)\varepsilon\tau\right)}{1-\gamma\left(u/\varepsilon\right)\exp\left(-d\left(u/\varepsilon\right)\varepsilon\tau\right)} 
=\frac{\Xi(u,0,\tau)}{v\varepsilon}+L_1(u,\tau)+\mathcal{O}(\varepsilon),
\end{equation}
with $L_1$ defined in~\eqref{eq:HestonDiag1stOrder} and $\Xi$ in~\eqref{eq:HestonDiagZeroOrder}. 
Combining~\eqref{eq:BetaAsymp} and~\eqref{eq:BAsymptotics} we deduce
\begin{equation}\label{eq:betaBasymp}
\beta_{\varepsilon t}B\left(u/\varepsilon,\varepsilon\tau\right)=\frac{\xi ^2 t\Xi(u,0,\tau)}{4v}+\left(\frac{L_1(u,\tau)\xi^2 t}{4}-\frac{\Xi(u,0,\tau)\kappa\xi^2t^2}{8v}\right)\varepsilon+\mathcal{O}(\varepsilon^2),
\end{equation}
and therefore as $\varepsilon$ tends to zero,
\begin{align}
\frac{\varepsilon B(u/\varepsilon,\varepsilon\tau) v\E^{-\kappa\varepsilon t}}
{1-2\beta_{\varepsilon t}B(u/\varepsilon,\varepsilon\tau)}
& =\frac{\left[\Xi(u,0,\tau)+vL_1(u,\tau)\varepsilon+\mathcal{O}\left(\varepsilon^2\right)\right]
\left(1-t\kappa\xi+\mathcal{O}(\varepsilon^2)\right)}
{1-\xi ^2 t\Xi(u,0,\tau)/2v+\left(\Xi(u,0,\tau)\kappa\xi^2t^2/4v-L_1(u,\tau)\xi^2 t/2\right)\varepsilon+\mathcal{O}\left(\varepsilon^2\right)}\nonumber
 \\ 
 & =\Xi(u,t,\tau)+\left(\Xi(u,t,\tau )^2 \left(\frac{v L_1(u,\tau)}{\Xi(u,0,\tau )^2}-\frac{\kappa  \xi ^2 t^2}{4 v}\right)-\kappa  t \Xi(u,t,\tau )\right)\varepsilon+\mathcal{O}(\varepsilon^2).\label{eq:B1Asymp}
\end{align}
Again using~\eqref{eq:betaBasymp} we have
\begin{align} \label{eq:B2Asymp}
-\frac{2\kappa\theta\varepsilon}{\xi^2}
\log\left(1-2\beta_{\varepsilon t} B\left(u/\varepsilon,\varepsilon\tau\right)\right)
=-\frac{2\kappa\theta}{\xi^2}\log\left(1-\frac{\Xi(u,0,\tau)\xi ^2 t}{2v}\right)\varepsilon
+\mathcal{O}(\varepsilon^2).
\end{align}
Recalling that
$$
\Lambda_{\varepsilon}(u)=\varepsilon A\left(u/\varepsilon,\varepsilon\tau\right)+\frac{\varepsilon B\left(u/\varepsilon,\varepsilon\tau\right)}{1-2\beta_{\varepsilon t} B\left(u/\varepsilon,\varepsilon\tau\right)}v\E^{-\kappa \varepsilon t}
-\frac{2\kappa\theta\varepsilon}{\xi^2}\log\left(1-2\beta_{\varepsilon t} B\left(u/\varepsilon,\varepsilon \tau\right)\right),
$$
the lemma follows by combining~\eqref{eq:AAsymp},~\eqref{eq:B1Asymp} and~\eqref{eq:B2Asymp}.
\end{proof}

\begin{lemma}\label{lemma:Smalltimehestonessentsmooth}
For all $t\geq0$,  $\tau>0$, $\Xi$ (given in~\eqref{eq:HestonDiagZeroOrder}) is convex and essentially smooth 
on $\mathcal{K}_{t,\tau}$, defined in~\eqref{eq:HestonDiagZeroOrder}.
\end{lemma}
\begin{proof}
The first derivative of $\Xi$ is given, after simplification, by
$$
\frac{\partial\Xi(u,t,\tau)}{\partial u}
  = \frac{\Xi(u,t,\tau)}{u}\left[1+\frac{\Xi(u,t,\tau)}{v}
\left(\frac{\xi^2 t}{2}+\frac{1}{2} \xi^2\bar{\rho}^2 \tau  \csc^2\left(\frac{1}{2}\bar{\rho}\xi\tau u\right)\right)\right].
$$
Any sequence tending to the boundary satisfies $\Xi(u,0,\tau)\to 2v/\xi^2 t$ 
which implies $\Xi(u,t,\tau)\uparrow\infty$ from Remark~\ref{remark:lambdapositive} and hence 
$|\partial\Xi(u,t,\tau)/\partial u|\uparrow\infty$.
Therefore $\Xi(\cdot,t,\tau)$ is essential smooth.
Now, 
$$
\frac{\partial^2\Xi(u,t,\tau)}{\partial u^2}
 = \frac{\xi^2}{2}\Xi(u,t,\tau)\frac{\left(t+\bar{\rho}^2\tau\csc^2(\psi_u)\right)^2}
{\left(\rho +\frac{1}{2} \xi t u-\bar{\rho} \cot(\psi_u)\right)^2}
+\frac{v+\bar{\rho}^2 \tau v\left(1-\psi_u \cot(\psi_u)\right) \csc^2(\psi_u)}
{\left(\rho +\frac{1}{2} \xi t u-\bar{\rho}\cot (\psi_u)\right)^2},
$$
where $\psi_u:=\bar{\rho}\xi \tau u/2$.
For $u\in\mathcal{K}_{t,\tau}\setminus\{0\}$, we have $\Xi(u,t,\tau)>0$ and $\Xi(0,t,\tau)=0$ from Remark~\ref{remark:lambdapositive}. Also we have the identity that $1-\theta/2\cot \left(\theta/2\right)\geq0$ for $\theta\in\left(-2\pi,2\pi\right)$, so that $\Xi$ is strictly convex on $\mathcal{K}_{t,\tau}$.
\end{proof}

As detailed in the beginning of this subsection, this concludes the proof of Proposition~\ref{Proposition:HestonDiagonal}.
We now prove the forward implied volatility expansions, 
namely Corollaries~\ref{cor:DiagTaylorExpansion} and~\ref{cor:DiagTaylorExpansionTypeII}.

\begin{proof}[Proof of Corollary~\ref{cor:DiagTaylorExpansion}]
We first look for a Taylor expansion of $u^*(k)$ around $k=0$ using 
$\Xi'(u^*(k),t,\tau)=k$. 
Differentiating this equation iteratively and setting $k=0$ (and using $u^*(0)=0$) 
gives an expansion for $u^*$ in terms of the derivatives of $\Xi$. 
In particular, 
$\Xi''(0,t,\tau) u^{*'}(0)=1$
and
$\Xi'''(0,t,\tau)(u^{*'}(0))^2+\Xi''(0,t,\tau) u^{*''}(0)=0$,
which implies that $u^{*'}(0)=1/\Xi''(0,t,\tau)$ and $u^{*''}(0)=-\Xi'''(0,t,\tau)/\Xi''(0,t,\tau)^3$. From the explicit expression of $\Xi$ in~\eqref{eq:HestonDiagZeroOrder}, we then obtain
\begin{align*}
u^*(k) &= \frac{k}{\tau  v}-\frac{3\xi\rho}{4\tau v^2}k^2
+\frac{\xi^2\left(\left(19\rho^2-4\right)\tau-12t\right)}{24\tau ^2 v^3}k^3 
+\frac{5\xi ^3\rho\left(48t+\left(16-37\rho^2\right)\tau\right)}{192\tau^2 v^4}k^4
\\ 
&+\frac{\xi^4\left(1080 t^2+\left(2437\rho^4-1604\rho^2+112\right)\tau^2-180 \left(27\rho^2-4\right)
 \tau t\right)}{1920\tau^3 v^5}k^5+\mathcal{O}(k^6).
\end{align*}
Using this series expansion and the fact that $\Lambda^*(k)=u^*(k)k-\Xi(u^*(k),t,\tau)$, 
the corollary follows from tedious but straightforward Taylor expansions of $v_0$ and $v_1$ defined in~\eqref{eq:v01SmallTime}.
\end{proof}
Corollary~\ref{cor:DiagTaylorExpansionTypeII} on the Type-II diagonal small-maturity Heston forward smile
follows from the following lemma:
\begin{lemma}\label{lemma:sharepicemeasureheston}
Under the stopped-share-price measure~\eqref{eq:StoppedShareMeasure} the forward Heston lmgf reads
$$
\log\mathbb{E}\left(\E^{u X_{\tau}^{(t)}}\right)
=
A(u,\tau)+\frac{B(u,\tau)}{1-2\widetilde{\beta}_t B(u,\tau)}v\E^{-\widetilde{\kappa} t}
-\frac{2\kappa\theta}{\xi^2}\log\left(1-2\widetilde{\beta}_t B(u,\tau)\right),
$$
for all $u$ such that the rhs exists, where $A$ and $B$ are defined in~\eqref{eq:ABfunctions},
$\widetilde{\beta}_t  := \frac{\xi^2}{4 \widetilde{\kappa}}(1-\E^{-\widetilde{\kappa} t})$
and
$\widetilde{\kappa}:=\kappa-\xi\rho$.
\end{lemma}

\begin{proof}
Under the stopped-share-price measure~\eqref{eq:StoppedShareMeasure} the Heston dynamics are given by
\begin{equation*}
\begin{array}{rll}
\D X_u & = \left(-\frac{1}{2}V_u+V_u\ind_{u\leq t}\right)\D u+ \sqrt{V_u}\D W_u, \quad & X_0\in\mathbb{R},\\
\D V_u & = \left(\kappa\theta-\kappa V_u+\rho\xi V_u\ind_{u\leq t}\right)\D u+\xi\sqrt{V_u}\D B_u, \quad & V_0=v>0,\\
\D\left\langle W,B\right\rangle_u & = \rho \D u.
\end{array}
\end{equation*}
Using the tower property for expectations, it is now straightforward to compute
$$
\mathbb{\widetilde{E}}\left(\E^{u\left(X_{t+\tau}-X_t\right)}\right)
=\mathbb{\widetilde{E}}\left(\mathbb{\widetilde{E}}\left(\E^{u\left(X_{t+\tau}
-X_t\right)}|\mathcal{F}_t\right)\right)
=\mathbb{\widetilde{E}}\left(\E^{A(u,\tau)+B(u,\tau)V_t}\right)\\
=\E^{A(u,\tau)}\widetilde{\Lambda}_{t}^{V}(B(u,\tau)), 
$$
where
$\widetilde{\Lambda}_{t}^{V}(u)
=\exp\left({\frac{uv \exp(-\widetilde{\kappa}t)}{1-2\widetilde{\beta}_tu}}\right)(1-2\widetilde{\beta}_tu)^{-q/2},$
for all $u<1/(2\widetilde{\beta}_t)$,
with $q:=4\kappa\theta/\xi^2$.
\end{proof}

\subsubsection{Proofs of Section~\ref{sec:LargeMatHeston}} \label{sec:ProofsLargeMatHeston}
In this section, we prove the large-maturity asymptotics for the Heston model, and we shall use the standing assumption $\kappa>\rho\xi$. Let $\varepsilon=\tau^{-1}$ and consider the Heston process  $(\tau^{-1}X^{(t)}_{\tau})_{\tau>0}$ with $(X_t)_{t>0}$ defined in~\eqref{eq:Heston} and $(X_{\tau}^{(t)})_{\tau>0}$ defined in~\eqref{eq:XtTauDef}. 
Specifically $\Lambda_{\varepsilon}$ defined in~\eqref{eq:Renorm-mgf} is then given by
$\Lambda_{\varepsilon}(u)=\tau^{-1}\mathbb{E}(\E^{u X^{(t)}_{\tau}})$,
and for ease of notation we set
\begin{align}\label{eq:LambdaTauDef}
\Lambda^{(t)}_{\tau}(u)=\Lambda_{\varepsilon}(u)\quad\text{for all }u\in\mathcal{D}_{\varepsilon}.
\end{align}

We prove here Proposition~\ref{Prop:HestonLargeMaturity} in several steps:
\begin{enumerate}[(i)]
\item In Proposition~\ref{Prop:HestonLimitingmgfLargeTime} we show that $\mathcal{D}_{0}=\Ddh$ and that $\{0,1\}\subset\Ddh^o$;
\item  Lemma~\ref{lemma:HestonmgfExpansionLargeTime} proves the expansion of Assumption~\ref{assump:Differentiability} 
with $\LO=V$, $\Lambda_{1}=H$, $\Lambda_{2}=0$;
\item By Proposition~\ref{Prop:HestonLimitingmgfLargeTime} and Lemma~\ref{lemma:Vproperties}, 
 $V$ is strictly convex and essentially smooth on $\Ddh^o$ if $\rho_{-}\leq\rho\leq\min\left(\rho_{+},\kappa/\xi\right)$;
see also Remark~\ref{remark:largemathest}(ii);
\item The map $(\varepsilon,u)\mapsto \Lambda_{\varepsilon}(u)$
is of class $\mathcal{C}^{\infty}$ on $\mathbb{R}_{+}^*\times\Ddh^{o}$ , Assumption~\ref{assump:Differentiability}(v) is also satisfied and $V(1)=0$ from Lemma~\ref{lemma:Vproperties};
\item $u^*$ can be computed in closed-form and is given by $q^*$ in~\eqref{eq:V*q*}.
\item A direct application of Proposition~\ref{Prop:GeneralBSFwdVolLargeTime} completes the proof.
\end{enumerate}

The following lemma recalls some elementary facts about the function $V$ in~\eqref{eq:VandH}, 
which will be used throughout the section.
We then proceed with a technical result needed in the proof of Proposition~\ref{Prop:HestonLimitingmgfLargeTime}.

\begin{lemma} \label{lemma:Vproperties}
The function $V$ in~\eqref{eq:VandH} is $\mathcal{C}^\infty$, strictly convex and essentially smooth 
on~$(u_{-},u_{+})$ (defined in~\eqref{eq:DefUpmU*pm}).
Also, $u_{-}<0$, $u_{+}>1$, 
$V(0)=V(1)=0$ and 
$\lim_{u\downarrow u_-}V(u)$ and $\lim_{u\uparrow u_+}V(u)$ are both finite.
\end{lemma}

\begin{lemma}\label{lemma:rhoconditions} \
Let $\rho_{\pm}$ be defined as in~\eqref{eq:DefUpmU*pm}, $\beta_t$ in~\eqref{eq:DGammaBeta}, 
and recall the standing assumption $\rho<\kappa/\xi$.
Assume further that $t>0$ and define the functions $g_+$ and $g_-$ by
$$
g_{\pm}(\rho):=
\left(2\kappa-\rho\xi\right)\pm\rho\sqrt{\xi^2\left(1-\rho^2\right)+\left(2\kappa-\rho\xi\right)^2}
-\frac{\xi^2(1-\rho^2)}{\beta_t}.
$$
\begin{enumerate}[(i)]
\item
The inequalities $\rho_{-}\in (-1,0)$ and $\rho_{+}>1/2$ always hold; 
if $\kappa/\xi>\rho_{+}$ and $t\ne 0$, then $\rho_{+}<1$;
finally $\rho_{+}=1$ (and $\rho_{-}=-1$) if and only if $t=0$;
\item
the inequality
$g_+(\rho)>0$
holds if and only if
$\rho_{+}<1$ and $\rho\in(\rho_{+},1)$;
\item
the inequality $g_-(\rho)>0$ holds if and only if $\rho\in(-1,\rho_{-})$;
\item let $u^*_{\pm}$ be as in~\eqref{eq:DefUpmU*pm} and $t>0$.
Then $u_{+}^*>1$ if $\rho\leq\rho_{-}$, and 
$u_{-}^*<0$ if $\rho\geq\rho_{+}$. 
\end{enumerate}
\end{lemma}

\begin{proof}
We first prove Lemma~\ref{lemma:rhoconditions}(i).
The double inequality $-1<\rho_{-}< 0$ is equivalent to 
$$
\frac{\xi -(8 \kappa +\xi ) \E^{2 \kappa  t}}{\E^{\kappa  t}+1}
<-\sqrt{16 \kappa ^2 \E^{2 \kappa  t}+\xi ^2 \left(1-\E^{\kappa  t}\right)^2}
<\xi \left(1-\E^{\kappa  t}\right).
$$
The upper bound clearly holds, and the lower bound  follows from the identity
$$
\sqrt{16 \kappa ^2 \E^{2 \kappa  t}+\xi ^2 \left(1-\E^{\kappa  t}\right)^2}=\sqrt{\frac{\left(\xi -(8 \kappa +\xi ) \E^{2 \kappa  t}\right)^2}{\left(\E^{\kappa  t}+1\right)^2}-\frac{16 \kappa  \E^{2 \kappa  t} \left(\E^{\kappa  t}-1\right) \left(\kappa
   +\xi +\xi  \E^{\kappa  t}+3 \kappa  \E^{\kappa  t}\right)}{\left(\E^{\kappa  t}+1\right)^2}}.
$$
We now prove that $\rho_{+}>1/2$.
From~\eqref{eq:DefUpmU*pm} this is equivalent to
$
\sqrt{16 \kappa ^2 \E^{2 \kappa  t}+\xi ^2 \left(1-\E^{\kappa  t}\right)^2}
>\frac{4 \xi +(\kappa -4 \xi ) \E^{2 \kappa  t}}{4 \left(\E^{\kappa t}+1\right)}.
$
The result follows by rearranging the left-hand side as 
$$
\sqrt{16 \kappa ^2 \E^{2 \kappa  t}+\xi ^2 \left(1-\E^{\kappa  t}\right)^2}
=
\sqrt{\frac{\left(4 \xi +(\kappa -4 \xi ) \E^{2 \kappa  t}\right)^2}{16 \left(\E^{\kappa  t}+1\right)^2}
+\frac{\kappa  \E^{2 \kappa  t} \left(8 \xi  \left(\E^{2 \kappa t}-1\right)+\kappa  \left(512 \E^{\kappa  t}+255 \E^{2 \kappa  t}+256\right)\right)}
{16 \left(\E^{\kappa  t}+1\right)^2}}.
$$

Assume now $\kappa/\xi>\rho_{+}$.
The inequality $\rho_+< 1$ is equivalent to 
$\sqrt{16 \kappa ^2 \E^{2 \kappa  t}+\xi ^2 \left(1-\E^{\kappa  t}\right)^2}
<\frac{\xi +(8 \kappa -\xi ) \E^{2 \kappa  t}}{\E^{\kappa  t}+1}$,
or
\begin{equation}\label{eq:rho+<=1}
\sqrt{\frac{\left(\xi +(8 \kappa -\xi ) \E^{2 \kappa  t}\right)^2}{\left(\E^{\kappa  t}+1\right)^2}-\frac{16 \kappa  \E^{2 \kappa  t} \left(\E^{\kappa  t}-1\right) \left(\kappa -\xi 
   \left(\E^{\kappa  t}+1\right)+3 \kappa  \E^{\kappa  t}\right)}{\left(\E^{\kappa  t}+1\right)^2}}
<\frac{\xi +(8 \kappa -\xi ) \E^{2 \kappa  t}}{\E^{\kappa  t}+1}.
\end{equation}
This statement is true if $\kappa -\xi\left(\E^{\kappa t}+1\right)+3 \kappa\E^{\kappa  t}>0$
and if the rhs is positive, which follow from the obvious inequalities $\frac{\E^{\kappa  t}+1}{3 \E^{\kappa  t}+1}<1/2<\kappa/\xi$.

We now prove Lemma~\ref{lemma:rhoconditions}(ii). 
The equation $g_+(\rho)=0$ implies (by squaring and rearranging the terms):
$$ 
4 \kappa  (\rho^2 -1)\left(4 \kappa \E^{2 \kappa  t}\rho ^2  +\xi (1 -\E^{2 \kappa  t}) \rho  -\kappa(1 +2   \E^{\kappa  t}+  \E^{2\kappa  t})\right)=0.
$$
The roots of this equation are $\pm 1$ and $\rho_{\pm}$ defined in~\eqref{eq:DefUpmU*pm}. 
The two possible positive roots are $\{\rho_+,1\}$ and the two possible negative ones are $\{\rho_-,-1\}$.
Clearly $g_+(-1)=0$.
Straightforward computations show that $g_+'(-1)<0$ and $g_+'(0)>0$.
Since $g_+$ is continuous on $(-1,0)$ with $g_+(0)<0$, it cannot have a single root in this interval,
and $\rho_-\in(-1,0)$ (by Lemma~\ref{lemma:rhoconditions}(i)) is hence not a valid root.
Consider now $\rho \in (0,1]$. 
From Lemma~\ref{lemma:rhoconditions}(i) the only possible roots are $1$ and $\rho_{+}$. 
Now  $g_+(1)=2\kappa-\xi+|2\kappa-\xi|$.
If $\kappa/\xi>1/2$ then $g_+(1)>0$ and hence $\rho_+$ is the unique root of $g_+$ in $(0,1)$.
Assume now that $\kappa/\xi \leq 1/2$, which implies $g_+(1)=0$.
Either $g'_+(1)\geq 0$ or $g'_+(1)<0$.
Since $g_+(0)<0$, the first case implies that $g_+$ has zero or more than two roots in $(0,1)$.
If it has zero root, then clearly $g_+(\rho)<0$ for $\rho\in (0,1)$.
More than two roots yields a contradiction with the fact that $\rho_+$ is the only possible root on $(0,1)$.
Now, Inequality~\eqref{eq:rho+<=1} implies that $\rho_{+}<1$ if and only if $\kappa/\xi>(\E^{\kappa t}+1)/(3\E^{\kappa t}+1)$, which is equivalent to $g'_+(1)<0$.
Therefore in the case $\kappa/\xi \leq 1/2$, the only possible scenario is $g'_+(1)<0$, 
where $g_+$ has a unique root $\rho_{+}\in (0,1)$.
In summary, on the interval~$[-1,1]$, $g_{+}(\rho)>0$ if and only if $\rho\in(\rho_{+},1)$ and $\rho_{+}<1$.
The proof of~(iii) is analogous to the proof of~(ii) and we omit it for brevity.

We now prove Lemma~\ref{lemma:rhoconditions}(iv). 
From~\eqref{eq:DefUpmU*pm} write $\nu=z(\rho)^{1/2}$, where
$z(\rho):=\xi ^2-2 \E^{\kappa  t} \left(8 \kappa ^2-4 \kappa  \xi  \rho +\xi ^2\right)+\E^{2 \kappa  t} (\xi -4 \kappa  \rho )^2$.
The two numbers $u_{-}^*$ and $u_{+}^*$ in~\eqref{eq:DefUpmU*pm} are well-defined in $\RR$ if and only if $z(\rho)\geq 0$ and $t>0$.
The two roots of this polynomial are given by
$\chi_\pm:=\frac{1}{4\kappa}\left[\E^{-\kappa  t} \left(\xi(\E^{\kappa t}-1)\pm 4\kappa  \E^{\kappa  t/2}\right)\right]$.
We now claim that
$\rho_{-}\leq \chi_-$  and $\rho_{+}\geq \chi_+$.
From the expression of $\rho_-$ given in~\eqref{eq:DefUpmU*pm}, the inequality $\rho_{-}\leq \chi_-$ 
can be rearranged as 
$$
-\sqrt{\xi ^2+16 \kappa ^2 \E^{2 \kappa  t}-2 \xi ^2 \E^{\kappa  t}+\xi ^2 \E^{2 \kappa  t}}
\leq\frac{\xi -2 \xi  \E^{\kappa  t}+\xi  \E^{2 \kappa  t}-8 \kappa  \E^{3 \kappa  t/2}}{\E^{\kappa  t}+1}.
$$
The claim then follows from the identity
$$
\sqrt{\xi ^2+16 \kappa ^2 \E^{2 \kappa  t}-2 \xi ^2 \E^{\kappa  t}+\xi ^2 \E^{2 \kappa  t}}
= 
\sqrt{\frac{4 \E^{\kappa  t} \left(\E^{\kappa  t}-1\right)^2 \left(\xi +2 \kappa  \E^{\kappa  t/2}\right)^2}{\left(\E^{\kappa  t}+1\right)^2}
+\frac{\left(\xi -2 \xi  \E^{\kappa  t}+\xi  \E^{2 \kappa  t}-8 \kappa  \E^{3 \kappa  t/2}\right)^2}{\left(\E^{\kappa  t}+1\right)^2}}.
$$
Analogous manipulations imply $\rho_{+}\geq \chi_+$,
and hence $z(\rho)$ is a well-defined real number for $\rho\in [-1,\rho_-]\cup[\rho_+,1]$.

The claim~$u_{-}^*<0$ is equivalent to 
$-\sqrt{\xi ^2-2 \E^{\kappa  t} \left(8 \kappa ^2-4 \kappa  \xi  \rho +\xi ^2\right)+\E^{2 \kappa  t} (\xi -4 \kappa  \rho )^2}<\xi  \left(1-\E^{\kappa  t}\right)+4\kappa  \rho  \E^{\kappa  t}$,
which holds as soon as $\xi  \left(1-\E^{\kappa  t}\right)+4\kappa  \rho  \E^{\kappa  t}>0$, or
$\rho >\frac{\xi}{4 \kappa}\left(1- \E^{-\kappa t}\right)$.
Therefore for any $\rho\geq \rho_+$, $u_{-}^*<0$ if and only if 
$\rho_{+}>\frac{\xi}{4 \kappa}\left(1- \E^{-\kappa t}\right)$.
This simplifies to
$\sqrt{\xi ^2+16 \kappa ^2 \E^{2 \kappa  t}-2 \xi ^2 \E^{\kappa  t}+\xi ^2 \E^{2 \kappa  t}}>\frac{\xi  \left(\E^{\kappa  t}-1\right)^2}{\E^{\kappa  t}+1}$,
which also reads
$$\sqrt{\frac{4 \E^{\kappa t} \left(4\kappa^2 \E^{\kappa t}\left(\E^{\kappa  t}+1\right)^2
+\xi ^2 \left(\E^{\kappa  t}-1\right)^2\right)}{\left(\E^{\kappa t}+1\right)^2}
+\frac{\xi^2 \left(\E^{\kappa  t}-1\right)^4}{\left(\E^{\kappa  t}+1\right)^2}}
>\frac{\xi\left(\E^{\kappa  t}-1\right)^2}{\E^{\kappa  t}+1},$$
and this is clearly true. 
Now straightforward manipulations show that the inequality $u_{+}^*>1$ is equivalent to
$$\sqrt{\left(\xi  \left(\E^{\kappa  t}-1\right)+4 \kappa  \rho  \E^{\kappa  t}\right)^2-16 \kappa  \E^{\kappa  t} \left(\kappa +\xi  \rho  \left(\E^{\kappa  t}-1\right)\right)}>\xi 
\left(\E^{\kappa  t}-1\right)+4 \kappa  \rho  \E^{\kappa  t},$$
which is true if 
$\displaystyle \rho <-\frac{\kappa }{\xi  \left(\E^{\kappa  t}-1\right)}$
or 
$\displaystyle \rho <-\frac{\xi  \left(1- \E^{-\kappa  t}\right)}{4\kappa}$.
And of course the claim $\left(u_{+}^*>1\text{ if } \rho\leq\rho_{-}\right)$ holds if 
\begin{equation}\label{eq:u*+>1i}
\rho_- <-\frac{\kappa }{\xi  \left(\E^{\kappa  t}-1\right)}
\qquad\text{or}\qquad
\rho_- <-\frac{  \xi  \left(1- \E^{-\kappa  t}\right)}{4\kappa}.
\end{equation}
The first inequality, which can be re-written as
$$
-\sqrt{\frac{16 \kappa ^2 \E^{3 \kappa  t} \left(\xi ^2 \left(\E^{\kappa  t}-1\right)^2 \left(\E^{\kappa  t}+1\right)-4 \kappa ^2 \E^{\kappa  t}\right)}{\xi ^2
   \left(\E^{2 \kappa  t}-1\right)^2}+\left(\frac{\xi^2(1-\E^{\kappa t})(1-\E^{2\kappa t})+8\kappa^2\E^{2\kappa t}}{\xi(\E^{\kappa t}+1)(1-\E^{\kappa t})}\right)^2}<\frac{\xi^2(1-\E^{\kappa t})(1-\E^{2\kappa t})+8\kappa^2\E^{2\kappa t}}{\xi(\E^{\kappa t}+1)(1-\E^{\kappa t})},
$$
holds if 
$\xi ^2 \left(\E^{\kappa  t}-1\right)^2 \left(\E^{\kappa  t}+1\right)-4 \kappa ^2 \E^{\kappa  t}>0$, 
or
$
\frac{ \left(\E^{\kappa  t}-1\right)^2 \left(1+\E^{-\kappa t}\right)}{4} >\frac{\kappa ^2}{\xi ^2}.
$
Quick manipulations turn the second inequality in~\eqref{eq:u*+>1i} into
$$
-\sqrt{\frac{4 \E^{\kappa  t} \left(4 \kappa ^2 \E^{\kappa  t} \left(\E^{\kappa  t}+1\right)^2-\xi ^2 \left(\E^{\kappa  t}-1\right)^2 \left(2 \E^{\kappa  t}+1\right)\right)}{\left(\E^{\kappa  t}+1\right)^2}+\frac{\xi ^2 \left(2 \E^{\kappa  t}-3 \E^{2 \kappa  t}+1\right)^2}{\left(\E^{\kappa  t}+1\right)^2}}<\frac{\xi  \left(2 \E^{\kappa  t}-3\E^{2 \kappa  t}+1\right)}{\E^{\kappa  t}+1}.
$$
Again this trivially holds if $4 \kappa ^2 \E^{\kappa  t} \left(\E^{\kappa  t}+1\right)^2-\xi ^2 \left(\E^{\kappa  t}-1\right)^2 \left(2 \E^{\kappa  t}+1\right)>0$, which is in turn equivalent to 
$
\frac{\kappa ^2}{\xi ^2}>\frac{\left(\E^{\kappa  t}-1\right)^2 \left(2+\E^{-\kappa t} \right)}{4 \left(\E^{\kappa  t}+1\right)^2}.
$
Since 
$
\frac{\left(\E^{\kappa  t}-1\right)^2 \left(2+\E^{-\kappa t} \right)}{4 \left(\E^{\kappa  t}+1\right)^2}<\frac{ \left(\E^{\kappa  t}-1\right)^2 \left(1+\E^{-\kappa t}\right)}{4},
$
is clearly true, it follows that for any valid choice of parameters either inequality (or both) in~\eqref{eq:u*+>1i} holds, 
and the claim follows.
\end{proof}

We now use Lemma~\ref{lemma:rhoconditions} to compute the large-maturity lmgf effective limiting domain for the forward price process $(\tau^{-1}X^{(t)}_{\tau})_{\tau>0}$.
This is of fundamental importance since in the large-maturity case (unlike the diagonal small-maturity case) we need to find conditions on the parameters of the model such that the limiting lmgf is essentially smooth (Assumption~\ref{assump:Differentiability}(iv)) on the interior of its effective domain.

\begin{proposition}\label{Prop:HestonLimitingmgfLargeTime}
 Let $\varepsilon=\tau^{-1}$ and consider the large-maturity Heston forward process  $(\tau^{-1}X^{(t)}_{\tau})_{\tau>0}$. 
Then $\mathcal{D}_{0}=\Ddh$ and $\{0,1\}\subset\mathcal{D}_{0}^o$ 
with $\Ddh$ and $\mathcal{D}_{0}$
defined in~\eqref{eq:DInfinityLargeMaturity} and in Assumption~\ref{assump:Differentiability}.
\end{proposition}

\begin{proof}
The tower property yields
$$
\mathbb{E}\left(\E^{u\left(X_{t+\tau}-X_t\right)}\right)
 = \mathbb{E}\left[\mathbb{E}\left(\E^{u\left(X_{t+\tau}-X_t\right)}|\mathcal{F}_t\right)\right]
 = \mathbb{E}\left(\E^{A(u,\tau)+B(u,\tau)V_t}\right)
 = \E^{A(u,\tau)}\mathbb{E}\left(\E^{B(u,\tau)V_t}\right),
$$
with $A$ and $B$ defined in~\eqref{eq:ABfunctions}.
For any fixed $t\geq0$ we require that
\begin{align}\label{eq:finitemom}
\mathbb{E}\left(\E^{u\left(X_{t+\tau}-X_t\right)}|\mathcal{F}_t\right)<\infty\quad\text{for all}\quad\tau>0.
\end{align}
Andersen and Piterbarg ~\cite[Proposition 3.1]{AP07} proved that if the following conditions are satisfied
\begin{equation}\label{eq:FiniteMomCond}
\kappa > \rho\xi u
\qquad\text{and}\qquad
(\kappa-\rho\xi u)^2+u(1-u)\xi^2 \geq 0,
\end{equation}
then the explosion time is infinite and~\eqref{eq:finitemom} is satisfied. 
In~\cite{FJ09} the authors proved that these conditions are equivalent to $\kappa>\rho\xi$ 
and $u\in\left[u_{-},u_{+}\right]$, with $u_{-}<0$ and $u_{+}>1$ 
($u_{\pm}$ defined in~\eqref{eq:DefUpmU*pm}).
Further we require that
\begin{eqnarray}\label{eq:VCondition}
\mathbb{E}\left(\E^{B(u,\tau)V_t}\right)<\infty,\quad\text{for all }\tau>0.
\end{eqnarray}
Now denote 
$\mathcal{K}_V:=\{u\in\mathbb{R}:\mathbb{E}(\E^{B(u,\tau)V_t})<\infty,\text{ for all } \tau>0\}$.
Then if $\kappa>\rho\xi$, the domain of the limiting forward lmgf is given by $\Ddh=\left[u_{-},u_{+}\right]\cap\mathcal{K}_V$.
Condition~\eqref{eq:VCondition} is equivalent to
$B(u,\tau)<1/(2\beta_t)$ for all $\tau>0$.
A simple calculation gives $B(0,\tau)=B(1,\tau)=0$ for all $\tau>0$ .
Furthermore for $u\in (0,1)$, and given Conditions~\eqref{eq:FiniteMomCond},
we have
$d(u)>\kappa-\rho\xi u$ and $\gamma(u)<0$.
This implies that $B(u,\tau)<0$ for $u\in (0,1)$ and $\tau>0$. 
In particular $[0,1]\subset\Ddh$ (martingale condition). 
For fixed $u\in\mathbb{R}$,
$$
\frac{\partial B(u,\tau)}{\partial \tau}
=\frac{2u (u-1) d(u)^2 \E^{d(u)\tau}}{\left(\kappa-\kappa \E^{d(u)\tau}+\xi\rho u\left(\E^{d(u)\tau}-1\right)-d(u) \left(\E^{d(u)\tau}+1\right)\right)^2},$$
so that for any $u\not\in\left[0,1\right]$, $B\left(u,\cdot\right)$ is strictly increasing. 
Therefore 
\begin{equation}\label{eq:Dv}
\mathcal{K}_{V}
=\left\{u\in\mathbb{R}:\lim_{\tau\uparrow\infty}B(u,\tau)<\frac{1}{2\beta_t}\right\}.
\end{equation}
We have
$\lim_{\tau\uparrow\infty}B(u,\tau) 
= \xi^{-2}(\kappa-\rho\xi u-d(u))
$.
So the condition is equivalent to
$\kappa-\rho\xi u-d(u)<2\kappa/(1-\E^{-\kappa t}).$
If $\rho\leq 0$ ($\rho\geq 0$) and $u\leq0$ ($u\geq0$) then
$\kappa-\rho\xi u-d(u)
\leq\kappa-\rho\xi u 
\leq\kappa
<\frac{2\kappa}{1-\E^{-\kappa t}}
$,
and the condition in~\eqref{eq:Dv} is always satisfied. 
So if $\rho=0$, $\Ddh=[u_{-},u_{+}].$ 
If $\rho<0$ ($\rho>0$), 
then $\RR_-\subset \mathcal{K}_V$ ($\RR_+\subset \mathcal{K}_V$), and hence
$\Ddh$ contains $[u_{-},0]$ ($[0,u_{+}]$). 
Now suppose that $\rho<0$ and $u>0$. 
The condition in~\eqref{eq:Dv} ($V$ given in~\eqref{eq:VandH}) is equivalent to
$V(u)<\kappa\theta/(2\beta_t)$.
From Lemma~\ref{lemma:Vproperties}, on~$(0,u_+]$, the function~$V$ 
attains its maximum at~$u_{+}$.
Using the properties in Lemma~\ref{lemma:Vproperties}, there exists $u_+^*\in(1,u_+)$ solving the equation
\begin{eqnarray}\label{eq:Uroots}
\frac{V( u_+^*)}{\kappa\theta}=\frac{1}{2\beta_t},
\end{eqnarray}
if and only if 
$g_-(\rho)>0$ ($g_-$ defined in Lemma~\ref{lemma:rhoconditions}),
which is equivalent (see Lemma~\ref{lemma:rhoconditions}) to $-1<\rho<\rho_{-}$ and $t>0$. 
The solution to~\eqref{eq:Uroots} has two roots $u_-^*$ and $u_+^*$ defined in~\eqref{eq:DefUpmU*pm},
and the correct solution here is $u_+^*$ by Lemma~\ref{lemma:rhoconditions}(iv).
So if $\rho_{-}\leq\rho<0$ then $\Ddh = [u_{-},u_{+}]$. 
If $-1<\rho<\rho_{-}$ and $t>0$ then $\Ddh = [u_{-},u_{+}^*)$. 
Analogous arguments show that for $0<\rho\leq\min\left(\kappa/\xi,\rho_{+}\right)$, we have
$\Ddh = [u_{-},u_{+}]$.
If $\rho_{+}<\rho<\min\left(\kappa/\xi,1\right)$, $t>0$ and $\kappa>\rho_{+}\xi$ then 
$\Ddh = (u_{-}^*,u_{+}]$, with $u_{-}<u_{-}^*<0$.
\end{proof}

\begin{lemma}\label{lemma:HestonmgfExpansionLargeTime}
The following expansion holds for the forward lmgf $\Lambda^{(t)}_{\tau}$ defined in~\eqref{eq:LambdaTauDef}:
$$
\Lambda^{(t)}_{\tau}(u)
=V(u)+\frac{H(u)}{\tau}\left(1+\mathcal{O}\left(\E^{-d(u)\tau}\right)\right),
\quad\text{for all }u\in\Ddh^o, \text{ as }\tau\text{ tends to infinity},
$$
where the functions $V$, $H$, $d$ and the interval $\Ddh$ are defined 
in~\eqref{eq:VandH},~\eqref{eq:DGammaBeta} and~\eqref{eq:DInfinityLargeMaturity}.
\end{lemma}

\begin{proof}[Proof of Lemma~\ref{lemma:HestonmgfExpansionLargeTime}]
From the definition of $\Lambda^{(t)}_{\tau}$ in~\eqref{eq:LambdaTauDef} and the Heston forward lmgf given in~\eqref{eq:LambdaTau} 
we immediately obtain the following asymptotics as $\tau$ tends to infinity:
\begin{align*}
A(u,\tau)=\tau  V(u)-\frac{2\kappa\theta}{\xi ^2}\log \left(\frac{1}{1-\gamma (u)}\right)+\mathcal{O}\left(\E^{-d(u)\tau}\right), \qquad
B(u,\tau)=\frac{V(u)}{\kappa\theta}+\mathcal{O}\left(\E^{-d(u)\tau}\right), 
\end{align*}
where $A$ and $B$ are defined in~\eqref{eq:ABfunctions}, $V$ in~\eqref{eq:VandH} and $d$ and $\gamma$  in~\eqref{eq:DGammaBeta}.
In particular this implies that 
$\frac{B(u,\tau)}{1-2\beta_t B(u,\tau)}
 = \frac{V(u)}{\theta  \kappa -2 \beta _t V(u)}+\mathcal{O}\left(\E^{-d(u)\tau}\right)$ and 
$\log\left(1-2\beta_t B(u,\tau)\right)
  = \log \left(1-\frac{2 \beta _t V(u)}{\theta  \kappa }\right)+\mathcal{O}\left(\E^{-d(u)\tau}\right)$,
which are well-defined for all $u\in\Ddh^o$. 
We therefore obtain
$$
H(u)=\frac{V(u)}{\kappa\theta -2 \beta _t V(u)}v \E^{-\kappa t}-\frac{2\kappa\theta}{\xi ^2}\log \left(1-\frac{2 \beta _t V(u)}
{\kappa\theta }\right)-\frac{2\kappa\theta}{\xi ^2}\log\left(\frac{1}{1-\gamma (u)}\right),
$$
and the lemma follows from straightforward simplifications.
Note in passing that $d(u)>0$ for any $u\in\Ddh^o$.
\end{proof}

\subsection{Proofs of Section~\ref{sec:ExponentialLevyForwardSmile}}
\label{sec:ProofsExpLevy}
Let $\phi$ be the L\'evy exponent of the L\'evy process $N$. 
If $v$ follows~\eqref{eq:fellerdiff}, a straightforward application of the tower property for expectations yields
the forward lmgf:
\begin{equation}\label{eq:MGFFellerTC}
\log\mathbb{E}\left(\E^{u X_{\tau}^{(t)}}\right)
=
A(\phi(u),\tau)+\frac{B(\phi(u),\tau)}{1-2\beta_t B(\phi(u),\tau)}v\E^{-\kappa t}
-\frac{2\kappa\theta}{\xi^2}\log\left(1-2\beta_t B(\phi(u),\tau)\right),
\end{equation}
defined for all $u$ such that the rhs exists and where 
\begin{align}\label{eq:ABFellerTC}
A(u,\tau) & := 
\frac{\kappa\theta}{\xi^2}\left(\left(\kappa- d(u)\right)\tau-2\log\left(\frac{1-\gamma(u)\E^{-d(u)\tau}}{1-\gamma(u)}\right)\right),\qquad
B(u,\tau) := \frac{\kappa-d(u)}{\xi^2}\frac{1-\E^{-d(u)\tau}}{1-\gamma(u)\E^{-d(u)\tau}},\\ \label{eq:dgammabetatimechange}
d(u) & := \left(\kappa^2-2u\xi^2\right)^{1/2},
\qquad
\gamma(u) := \frac{\kappa-d(u)}{\kappa+d(u)}
\qquad\text{and}\qquad
\beta_t := \frac{\xi^2}{4\kappa}\left(1-\E^{-\kappa t}\right).
\end{align}
Similarly if $(v_t)_{t\geq0}$ follows~\eqref{eq:nonGaussOU} the forward lmgf is given by
\begin{equation}\label{eq:MGFNonGaussOU}
\log\mathbb{E}\left(\E^{u X_{\tau}^{(t)}}\right)
=
A(\phi(u),\tau)+B(\phi(u),\tau)v\E^{-\lambda t}
+\delta\log\left(\frac{B(\phi(u),\tau)-\E^{t\lambda}\alpha}{\E^{t\lambda}(B(\phi(u),\tau)-\alpha)}\right),
\end{equation}
defined for all $u$ such that the rhs exists and where 
\begin{equation}\label{eq:ABNonGaussOU} 
A(u,\tau) := 
\frac{\lambda \delta}{\alpha\lambda-u}\left[u \tau +\alpha \log\left(1-\frac{u}{\alpha \lambda}\left(1-\E^{-\lambda \tau}\right)\right)\right]
\qquad\text{and}\qquad
B(u,\tau) :=\frac{u}{\lambda}\left(1-\E^{-\lambda\tau}\right).
\end{equation}

\begin{proof}[Proof of Proposition~\ref{prop:fwdsmiletimechange}]

We show that Proposition~\ref{Prop:GeneralBSFwdVolLargeTime} is applicable given the assumptions of Proposition~\ref{prop:fwdsmiletimechange}. Consider case (i). The expansion for $\Lambda^{(t)}_{\tau}$ defined in~\eqref{eq:LambdaTauDef} is straightforward and analogous to Lemma~\ref{lemma:HestonmgfExpansionLargeTime}. In particular we establish that
$$
\Lambda^{(t)}_{\tau}(u)
=\widehat{V}(u)+\frac{\widehat{H}(u)}{\tau}\left(1+\mathcal{O}\left(\E^{-d(\phi(u))\tau}\right)\right),
\quad\text{for all }u\in\widehat{\mathcal{K}}^o_{\infty}, \text{ as }\tau\text{ tends to infinity},
$$
where the functions $\widehat{V}$, $\widehat{H}$, $d$ and the domain $\widehat{\mathcal{K}}_{\infty}$ 
are defined in~\eqref{eq:FellerLargeTime},~\eqref{eq:dgammabetatimechange} and~\eqref{eq:domains}. 
Since $\phi$ is essentially smooth and strictly convex on $\mathcal{K}_{\phi}$ 
and $\widehat{\mathcal{K}}_{\infty}\subseteq\mathcal{K}_{\phi}$, then the limiting lmgf
$\LO=\widehat{V}$ is essentially smooth and strictly convex on $\widehat{\mathcal{K}}_{\infty}$. 
The map $(\varepsilon,u)\mapsto \Lambda_{\varepsilon}(u)$ (defined in~\eqref{eq:LambdaTauDef})
is of class $\mathcal{C}^{\infty}$ on $\mathbb{R}_{+}^*\times\widehat{\mathcal{K}}_{\infty}^{o}$ since $\phi$ 
is of class $\mathcal{C}^{\infty}$ on $\widehat{\mathcal{K}}_{\infty}^{o}$ and Assumption~\ref{assump:Differentiability}(v) is also satisfied.
Since $\phi(1)=0$ we have that $\widehat{V}(1)=0$ and $\{0,1\}\subset\widehat{\mathcal{K}}_{\infty}^o$. 
It remains to be checked that the limiting domain is in fact given by~$\widehat{\mathcal{K}}_{\infty}$. 
We first note that that by conditioning on $(V_u)_{t\leq u \leq t+\tau}$ and using the independence of the time-change and the L\'evy process we have 
$\mathbb{E}\left(\E^{u\left(X_{t+\tau}-X_t\right)}\right)
=\mathbb{E}\left(\E^{\phi(u)\int_t^{t+\tau}v_s\D s }\right)$
and so any~$u$ in the limiting domain must satisfy $\phi(u)<\infty$. 
Using~\cite[page 476]{CT07} and the tower property we compute
\begin{equation}\label{eq:towerlaw}
\mathbb{E}\left(\E^{u\left(X_{t+\tau}-X_t\right)}\right)
 = \mathbb{E}\left[\mathbb{E}\left(\E^{\phi(u)\int_t^{t+\tau}v_s\D s}|\mathcal{F}_t\right)\right]
 = \mathbb{E}\left(\E^{A\left(\phi(u),\tau\right)+B\left(\phi(u),\tau\right)v_t}\right)
 = \E^{A\left(\phi(u),\tau\right)}\mathbb{E}\left(\E^{B\left(\phi(u),\tau\right)v_t}\right),
\end{equation}
with $A$ and $B$ given in~\eqref{eq:ABFellerTC}. 
Further from~\eqref{eq:HestonVarianceMGF} we have
$\log\mathbb{E}\left(\E^{u v_t}\right)=\frac{uv\E^{-\kappa t}}{1-2\beta_t u}-\frac{2\kappa\theta}{\xi^2}\log\left(1-2\beta_t u\right)$,
for all $u<1/(2\beta_t)$.
Following a similar argument to the proof of Proposition~\ref{Prop:HestonLimitingmgfLargeTime} we can show that for any $t\geq0$,
$B(\phi(u),\tau)<1/(2\beta_t)$ is always satisfied for each $\tau>0$. 
This follows from the independence of the L\'evy process~$N$ and the time-change. 
We also require that for any $t\geq 0$,
$
\mathbb{E}\left(\E^{\phi(u)\int_t^{t+\tau}v_s\D s }|\mathcal{F}_t\right)<\infty,
$
for every $\tau>0$. 
Here we use~\cite[Corollary 3.3]{AP07} with zero correlation to find that we require $\phi(u)\leq \kappa^2/(2\xi^2)$. 
It follows that $\widehat{\mathcal{K}}_{\infty}=\left\{u:\phi(u)\leq \kappa^2/(2\xi^2)\right\}$.

Regarding case (ii), arguments analogous to case (i) hold and we focus on showing that the limiting domain is $\widetilde{\mathcal{K}}_{\infty}$. Using~\cite[page 488]{CT07} Equality~\eqref{eq:towerlaw} also holds with $A$ and $B$ defined in~\eqref{eq:ABNonGaussOU}. 
Since we require that for any $t\geq 0$,
$\mathbb{E}\left(\E^{\int_t^{t+\tau}v_s\D s \phi(u)}|\mathcal{F}_t\right)<\infty$,
for every $\tau>0$ we have $\phi(u)<\alpha\lambda$. 
Using~\cite[page 482]{CT07} we also have 
\begin{align*}
\log\mathbb{E}\left(\E^{u v_t}\right)=u v\E^{-\lambda t}
+\delta\log\left(\frac{u-\alpha \E^{\lambda t}}{(u-\alpha)\E^{\lambda t}}\right),
\qquad
\text{for all }u<\alpha.
\end{align*}
But it is straightforward to show that $\phi(u)<\alpha\lambda$ implies $B(\phi(u),\tau)<\alpha$ for any $\tau>0$ and it follows that $\widetilde{\mathcal{K}}_{\infty}=\left\{u:\phi(u)<\alpha\lambda \right\}$. Case (iii) is straightforward and omitted.
\end{proof}

\appendix

\section{Tail Estimates}\label{sec: TailEstimates}

The purpose of this appendix is to prove that under Assumption~\ref{assump:Differentiability}(v) the tail integral
$
\left| \int_{|u|>1/\sqrt{\varepsilon}} \Phi_{Z_{k,\varepsilon}}(u) \overline{C_{\varepsilon,k}(u)} \D u \right|
$
is exponentially small, where $\Phi_{Z_{k,\varepsilon}}$ is defined in~\eqref{eq:PhiZ} and $C_{\varepsilon,k}$ is given in~\eqref{eq:Cdef}.
This is required in the proof of Theorem~\ref{theorem:GeneralOptionAsymp}. 
Now we recall~\eqref{eq:Cconj} that 
$$
\overline{C_{\varepsilon,k}(u)}=\frac{\varepsilon ^{3/2} f(\varepsilon)}
{\left(u^*(k)+\I u \sqrt{\varepsilon }\right) \left(u^*(k)-\varepsilon  f(\varepsilon )+\I u \sqrt{\varepsilon}\right)},
$$
and one can easily check that, for $|z|>1$,
\begin{equation}\label{eq:payoffbound}
\left| \overline{C_{\varepsilon,k}(z/\sqrt{\varepsilon})} \right|
\leq \varepsilon ^{3/2} f(\varepsilon)/z^2.
\end{equation}
Further by definition,
$
\left| \Phi_{Z_{k,\varepsilon}}(z/\sqrt{\varepsilon}) \right|
 = \exp\left[\frac{1}{\varepsilon}\left(\Re\left(\Lambda_{\varepsilon}\left(\I z +u^*(k)\right)\right)-\Lambda_{\varepsilon}\left(u^*(k)\right)\right)\right],
$
and for fixed $u^*(k)$, 
the ridge property for characteristic functions~\cite[Theorem 7.1.2]{L70} implies
$\left| \Phi_{Z_{k,\varepsilon}}(z/\sqrt{\varepsilon}) \right|\leq 1$
for all $z\in\mathbb{R}$ and $\varepsilon>0$.
Therefore the tail estimate
\begin{equation}\label{eq:FiniteBound}
\left| \int_{|u|>1/\sqrt{\varepsilon}} \Phi_{Z_{k,\varepsilon}}(u) \overline{C_{\varepsilon,k}(u)} \D u \right| \leq
\frac{1}{\sqrt{\varepsilon}}\int_{|z|>1} \left| \Phi_{Z_{k,\varepsilon}}(z/\sqrt{\varepsilon})\right| \left| \overline{C_{\varepsilon,k}(z/\sqrt{\varepsilon})} \right| \D z
\leq \varepsilon f(\varepsilon) \int_{|z|>1} \frac{\D z}{z^2} < \infty,
\end{equation}
is finite for sufficiently small $\varepsilon$ since $f(\varepsilon)\varepsilon=c+\mathcal{O}(\varepsilon)$.
We proceed now to show that Assumption~\ref{assump:Differentiability}(v) allows us to further conclude that this term is in fact exponentially small:
\begin{lemma}\label{lemma:tailestimatesteepcase}
There exists $\beta>0$ such that the following tail estimate holds for all $k \neq \Lambda_{0,1}(0)$
as $\varepsilon\downarrow 0$:
$$
\left| \int_{|u|>1/\sqrt{\varepsilon}} \Phi_{Z_{k,\varepsilon}}(u) \overline{C_{\varepsilon,k}(u)} \D u \right| = \mathcal{O}(\E^{-\beta/\varepsilon}).
$$
\end{lemma}
\begin{proof}
By Definition~\ref{eq:PhiZ},
$
\Phi_{Z_{k,\varepsilon}}(u)
 = \exp\left[-\frac{\I uk}{\sqrt{\varepsilon}}+\frac{1}{\varepsilon}\left(\Lambda_{\varepsilon}\left(\I u\sqrt{\varepsilon} +u^*(k)\right)-\Lambda_{\varepsilon}\left(u^*(k)\right)\right)\right].
$
Let
$\mathcal{R}(\varepsilon,z)\equiv\mathcal{R}_0(\varepsilon,z)+\mathcal{R}_1(\varepsilon)$, 
with
$
\mathcal{R}_0(\varepsilon,z):=\frac{1}{\varepsilon}\left[\Re\left(\Lambda_{\varepsilon}\left(\I z +u^*(k)\right)\right) - \Re\left(\LO\left(\I z +u^*(k)\right)\right) \right]$
and 
$\mathcal{R}_1(\varepsilon):=\frac{1}{\varepsilon}\left[\LO\left(u^*(k)\right) - \Lambda_{\varepsilon}\left(u^*(k)\right) \right]$.
Then
$$
|\Phi_{Z_{k,\varepsilon}}(z/\sqrt{\varepsilon})|
 = \exp\left[\frac{1}{\varepsilon}\left(\Re\left(\LO\left(\I z +u^*(k)\right)\right)-\LO\left(u^*(k)\right)\right)+\mathcal{R}(\varepsilon,z)\right],
$$
as $\varepsilon\downarrow 0$.
Set $F(z):=\Re\left(\LO\left(\I z +u^*(k)\right)\right)-\LO\left(u^*(k)\right)$.
Using~\eqref{eq:payoffbound} the tail estimate is then given by
$$
\left| \int_{|u|>1/\sqrt{\varepsilon}} \Phi_{Z_{k,\varepsilon}}(u) \overline{C_{\varepsilon,k}(u)} \D u \right|
\leq
\frac{1}{\sqrt{\varepsilon}}\int_{|z|>1}\left| \Phi_{Z_{k,\varepsilon}}(z/\sqrt{\varepsilon}) \right| \left|  \overline{C_{\varepsilon,k}(z/\sqrt{\varepsilon})}  \right| \D z
\leq
\varepsilon f(\varepsilon) \int_{|z|>1}\E^{F(z)/\varepsilon+\mathcal{R}(\varepsilon,z)}\frac{\D z}{z^2},
$$
as $\varepsilon$ tends to zero.
We deal with the case $z>1$. 
We note that
\begin{align*}
 \int_{z>1}\E^{F(z)/\varepsilon+\mathcal{R}(\varepsilon,z)}\frac{\D z}{z^2}
&=\ind_{\{p_i^*>1\}}\int_{1}^{p_i^*}\E^{F(z)/\varepsilon+\mathcal{R}(\varepsilon,z)}\frac{\D z}{z^2}
+ \int_{z>\max(p_i^*,1)}\E^{F(z)/\varepsilon+\mathcal{R}(\varepsilon,z)}\frac{\D z}{z^2}  \\
&\leq \frac{ (p_i^*-1)^+\E^{F(\widetilde{p}_i)/\varepsilon+\mathcal{R}(\varepsilon,\widetilde{p}_i)}}{ \widetilde{p}_i^2}
+ \int_{z>p_i^*}\E^{ F(z)/\varepsilon+\mathcal{R}(\varepsilon,z)} \frac{\D z}{z^2},
\end{align*}
where the first integral on the rhs follows from the extreme value theorem which implies that the integrand attains its maximum on~$[1,p_i^*]$ at some point $\widetilde{p}_i$ and the inequality for the second integral on the rhs follows since the integrand is positive.
Using Assumption~\ref{assump:Differentiability}(v)(c), for $z>p_i^*$ there exists $\varepsilon_1>0$ 
and $M$ (independent of $z$) such that $\mathcal{R}_0(\varepsilon,z)<M$ for $\varepsilon<\varepsilon_1$.
In particular for $\varepsilon<\varepsilon_1$ we have
$$
 \int_{z>1}\E^{F(z)/\varepsilon+\mathcal{R}(\varepsilon,z)}\frac{\D z}{z^2}
\leq \frac{ (p_i^*-1)^+\E^{F(\widetilde{p}_i)/\varepsilon+\mathcal{R}(\varepsilon,\widetilde{p}_i)}}{ \widetilde{p}_i^2}
+ \E^{M+\mathcal{R}_1(\varepsilon)} \int_{z>p_i^*}\E^{ F(z)/\varepsilon} \frac{\D z}{z^2}.
$$
Note Assumption~\ref{assump:Differentiability}(i) implies that $\mathcal{R}_1(\varepsilon)$ 
and $\mathcal{R}(\varepsilon,\widetilde{p}_i)$ are both $\mathcal{O}(1)$ quantities.
By a similar argument to~\eqref{eq:FiniteBound} the integral on the rhs is finite and we now use the Laplace method.
Since $F$ is continuous, has  a unique maximum at $z=0$ and is bounded away from zero as $|z|$ tends to infinity (Assumption~\ref{assump:Differentiability}(v)(b)) there exists $z_{+}^*>0$ such that $F(z_{+}^*)>F(z)$ for $z>z_{+}^*$. 
We now write
$$   \int_{z>p_i^*}\E^{ F(z)/\varepsilon}\frac{\D z}{z^2}
\leq   \int_{z>\min(p_i^*,z_{+}^*)}\E^{ F(z)/\varepsilon}\frac{\D z}{z^2}
\leq \frac{(z^*_{+}-p_i^*)^+\E^{F(z_{+})/\varepsilon}}{ z_{+}^2}
+ \int_{z>z^*_{+}}\E^{ F(z)/\varepsilon}\frac{\D z}{z^2},
 $$
where again the final step follows from the extreme value theorem: if $z_{+}^*>p_i^*$
the integrand attains its maximum on $[p_i^*,z_+^*]$ at $z_{+}$.
Since the contribution of the last integral is centered around $z=z_{+}^*$ as $\varepsilon\downarrow 0$, 
the Laplace method yields
$$
\int_{z>z^*_{+}}\E^{ F(z)/\varepsilon}\frac{\D z}{z^2}
\sim
-\frac{\varepsilon \E^{2 F(z^*_{+})/\varepsilon}}{2 F'(z^*_{+})(z^*_{+})^2}.
 $$
A similar argument holds for $z<1$ and the result follows.
\end{proof}

\section{Proof of Lemma~\ref{lem:optpricerep0}}
The proof of Lemma~\ref{lem:optpricerep0} proceeds in two steps:
we first prove that the integrand in the right-hand side of Equality~\eqref{eq:Pars} belongs to $L^1(\RR)$ 
(and hence the integral is well-defined), 
and we then prove that this very equality holds. 
The first step is contained in the following lemma.

\begin{lemma}\label{lem:L1lem}
There exists $\varepsilon^*_0>0$ such that
$
 \int_{\RR}|\Phi_{Z_{k,\varepsilon}}(u) \overline{C_{\varepsilon,k}(u)}|  \D u 
<\infty
$
 for all $\varepsilon<\varepsilon^*_0$ and  $k\in\mathbb{R}\backslash\{\Lambda_{0,1}(0), \Lambda_{0,1}(c)\}$.
\end{lemma}
\begin{proof}
We denote by $\Phi$ the characteristic function of the Gaussian with mean zero and variance $\Lambda_{0,2}$.
Since~$\Phi_{Z_{k,\varepsilon}}$ converges pointwise to $\Phi$ by Lemma~\ref{lemma:CharactExp}, then,
for any $B>0$,
it converges uniformly to $\Phi$ in every compact set $[-B,B]$ as $\varepsilon$ tends to zero 
(see~\cite[Corollary 1, page 50]{L70}).
Furthermore since $ \Phi_{Z_{k,\varepsilon}}$ is continuous and bounded on $\RR$,
then 
$|\Phi_{Z_{k,\varepsilon}}(u)|$ also converges uniformly to $|\Phi(u)|$ on $[-B,B]$,
and hence 
there exist $\varepsilon_0>0$ and $\delta_0>0$ such that for all $\varepsilon<\varepsilon_0$:
\begin{equation}\label{eq:unifcharact}
\int_{|u|<B}\left|\Phi_{Z_{k,\varepsilon}}(u)\right| \D u \leq 
\int_{|u|<B}\left|\Phi(u)\right| \D u + \delta_0.
\end{equation}
Following an analogous argument to the proof of Lemma~\ref{lemma:tailestimatesteepcase} we know that $\int_{|u|>1/\sqrt{\varepsilon}}| \Phi_{Z_{k,\varepsilon}}(u) \overline{C_{\varepsilon,k}(u)}|  \D u$ is exponentially small as $\varepsilon$ tends to zero.
Thus there exist $\varepsilon_1>0$ and $\delta_1>0$ and such that for all $\varepsilon<\varepsilon_1:$
\begin{align*}
\int_{\RR}\left| \Phi_{Z_{k,\varepsilon}}(u) \overline{C_{\varepsilon,k}(u)}  \right|  \D u 
&=
\int_{|u|\leq 1/\sqrt{\varepsilon}}\left| \Phi_{Z_{k,\varepsilon}}(u) \overline{C_{\varepsilon,k}(u)}  \right|  \D u 
+\int_{|u|>1/\sqrt{\varepsilon}}\left| \Phi_{Z_{k,\varepsilon}}(u) \overline{C_{\varepsilon,k}(u)}  \right|  \D u  \\
&\leq \frac{\varepsilon^{3/2}f(\varepsilon)}{|u^*(k)(u^*(k)-\varepsilon f(\varepsilon))|}\int_{|u|\leq 1/\sqrt{\varepsilon}}\left| \Phi_{Z_{k,\varepsilon}}(u) \right| \D u 
+ \delta_1,
\end{align*}
where the inequality for the first integral follows from the simple bound 
$$\left|  \overline{C_{\varepsilon,k}(u)}  \right|\leq \frac{\varepsilon^{3/2}f(\varepsilon)}{|u^*(k)(u^*(k)-\varepsilon f(\varepsilon))|},\quad\text{for all }|u| \leq1/\sqrt{\varepsilon}.$$
The quantity on the rhs is finite for $\varepsilon$ small enough since $\varepsilon f(\varepsilon)=c+\mathcal{O}(\varepsilon)$, $u^*(k)\neq c$, $u^*(k)\neq 0$.
Using~\eqref{eq:unifcharact}, fix some $B>0$ and $\varepsilon_0<1/B^2$;  
then for any $\varepsilon<\varepsilon_0<1/B^2$, 
$\int_{|u|<1/\sqrt{\varepsilon}}|\Phi_{Z_{k,\varepsilon}}(u)| \D u \leq \int_{|u|<1/\sqrt{\varepsilon}}|\Phi(u)| \D u + \delta_0$. 
This quantity is finite since the Gaussian characteristic function is in $L^1(\RR)$, and the lemma follows.
\end{proof}

We now move on to the proof of Lemma~\ref{lem:optpricerep0}.
We only look at the case $j=1$, the other cases being completely analogous.
We denote the convolution of two functions $f,h\in L^1(\RR)$ by $(f\ast g)(x):=\int_{\RR}f(x-y)g(y) \D y$,
and recall that $(f\ast g)\in L^1(\RR)$.
For such functions, we denote the Fourier transform by
$(\mathcal{F}f)(u):=\int_{-\infty}^{\infty}\E^{\I u x}f(x) \D x$ 
and the inverse Fourier transform by $
(\mathcal{F}^{-1}h)(x):=\frac{1}{2\pi}\int_{-\infty}^{\infty}\E^{-\I u x}h(u) \D u.
$ 

With $\tilde{g}_j$ defined on page~\pageref{eq:gjtilde}, 
the $\mathbb{Q}_{k,\varepsilon}$-measure in~\eqref{eq:MeasureChange} and the random variable $Z_{k,\varepsilon}$ in~\eqref{eq:zvarepsilon}, we have
$$
\mathbb{E}^{\mathbb{Q}_{k,\varepsilon}}\left[\tilde{g}_j(Z_{k,\varepsilon})\right]
=\int_{\RR}q_j(k/\sqrt{\varepsilon}-y)p(y) \D y = (q_j\ast p)(k/\sqrt{\varepsilon}),
$$
with $q_j(z)\equiv\tilde{g}_j(-z)$ and $p$ denoting the density of $Y_{\varepsilon}/\sqrt{\varepsilon}$.
On the strips of regularity derived on page~\pageref{eq:stripsofreg} we know there exists $\varepsilon_0>0$ such that
$q_j\in L^1(\RR)$ for $\varepsilon<\varepsilon_0$.
Since $p$ is a density, $p\in L^1(\RR)$, and therefore 
\begin{equation}\label{eq:conv}
\mathcal{F}(q_j\ast p)(u)=\mathcal{F}q_j(u) \mathcal{F}p(u).
\end{equation}
We note that $\mathcal{F}q_j(u)\equiv\mathcal{F}\tilde{g}_j(-u)\equiv\overline{\mathcal{F}\tilde{g}_j(u)}$ and hence
\begin{equation}\label{eq:simpconv}
\mathcal{F}q_j(u) \mathcal{F}p(u)
\equiv\E^{\I u k/\sqrt{\varepsilon}}\Phi_{Z_{k,\varepsilon}}(u) \overline{C_{\varepsilon,k}}(u).
\end{equation}
Thus by Lemma~\ref{lem:L1lem} there exists an $\varepsilon_1>0$ 
such that $\mathcal{F}q_j \mathcal{F}p\in L^1(\RR)$ for $\varepsilon<\varepsilon_1$. 
By the inversion theorem~\cite[Theorem 9.11]{R87} this then implies from~\eqref{eq:conv} and~\eqref{eq:simpconv} that for $\varepsilon<\min(\varepsilon_0,\varepsilon_1)$:
\begin{align*}
\mathbb{E}^{\mathbb{Q}_{k,\varepsilon}}\left[\tilde{g}_j(Z_{k,\varepsilon})\right]
&= (q_j\ast p)(k/\sqrt{\varepsilon}) 
=\mathcal{F}^{-1}\left(\mathcal{F}q_j(u) \mathcal{F}p(u)\right)(k/\sqrt{\varepsilon}) \\ 
&=\frac{1}{2\pi}\int_{\RR} \E^{-\I u k/\sqrt{\varepsilon}}\mathcal{F}q_j(u) \mathcal{F}p(u) \D u
= \frac{1}{2\pi}\int_{\RR} \Phi_{Z_{k,\varepsilon}}(u) \overline{C_{\varepsilon,k}(u)}  \D u .
\end{align*}

\begin{remark}
There exists $\varepsilon_0>0$ such that for the strips of regularity derived on page~\pageref{eq:stripsofreg}, the modified payoffs $\tilde{g}_j$ are in $L^2(\RR)$ for $\varepsilon<\varepsilon_0$. 
If there further exists $\varepsilon_1>0$ such that $\Phi_{Z_{k,\varepsilon}}\in L^2(\RR)$ for $\varepsilon<\varepsilon_1$
then we can directly apply Parseval's Theorem~\cite[Theorem 13E]{G70} for $\varepsilon<\min(\varepsilon_0,\varepsilon_1)$ and we obtain the same result as in Lemma~\ref{lem:optpricerep0}.
This requires though a stronger tail assumption compared to ~\ref{assump:Differentiability}(v)(c).
\end{remark}

\section{Verification of Assumption~\ref{assump:Differentiability}(v)}\label{append:tailverif}

The tail assumption~\ref{assump:Differentiability}(v) needs to be verified in order to apply Theorem~\ref{theorem:GeneralOptionAsymp}.
It is readily satisfied by most models used in practice.
Its verification is tedious but straightforward, and 
we give here an outline for the time-changed exponential L\'evy case where the time-change is given by an integrated Feller process~\eqref{eq:fellerdiff}, i.e. Proposition~\ref{prop:fwdsmiletimechange}(i).
Analogous arguments hold for all other models in the paper.

We recall that the forward lmg is given in~\eqref{eq:MGFFellerTC}
and the limiting lmgf (\eqref{eq:FellerLargeTime},\eqref{eq:domains}) is given by $\widehat{V}:{\widehat{\mathcal{K}}_{\infty}}\ni u \mapsto \frac{\kappa\theta}{\xi^2}\left(\kappa-\sqrt{\kappa^2-2\phi(u)\xi^2}\right)$
with $\widehat{\mathcal{K}}_{\infty}:=\left\{u:\phi(u)\leq\kappa^2/(2\xi^2)\right\}$ and $\phi$ is the L\'evy exponent.
Straightforward computations yield Assumption~\ref{assump:Differentiability}(v)(a).
For fixed $a\in\widehat{\mathcal{K}}_{\infty}^{0}$ denote $L_r:\mathbb{R}\to\mathbb{R}$ by $L_r(z):=\Re(\widehat{V}(\I z +a))$ and
$L_i:\mathbb{R}\to\mathbb{R}$ by $L_i(z):=\Im(\widehat{V}(\I z +a))$.
Then $\widehat{V}(\I z +a)=L_r(z)+\I L_i(z)$. 
Similarly we define $\phi_{r}$ and $\phi_{i}$ such that $\phi(\I z +a)=\phi_{r}(z)+\I \phi_{i}(z)$. 
From~\cite[Lemma A.1, page 10]{FFJ11} we know that 
$\phi_{r}$ has a unique maximum at zero and is bounded way from zero as $|z|$ tends to infinity. 
Now 
$
L_r(z):
= \frac{\kappa^2\theta}{\xi^2}
- \frac{\kappa\theta}{\xi^2}\Re\left(\sqrt{\kappa^2-2\phi(\I z + a)\xi^2}\right)
$
and
$
\Re\left(\sqrt{\kappa^2-2\phi(\I z + a)\xi^2}\right)
=\frac{1}{2}\sqrt{2(\kappa^2-2\phi_r(z)\xi^2) +2\sqrt{(\kappa^2-2\phi_r(z)\xi^2)^2+4\xi^4\phi_{i}(z)^2}}.
$
Since $\phi_r(z)<\phi_r(0)\leq \kappa^2/(2\xi^2)$  we certainly have
$$
\sqrt{2(\kappa^2-2\phi_r(z)\xi^2) +2\sqrt{(\kappa^2-2\phi_r(z)\xi^2)^2}}\leq
\sqrt{2(\kappa^2-2\phi_r(z)\xi^2) +2\sqrt{(\kappa^2-2\phi_r(z)\xi^2)^2+4\xi^4\phi_{i}(z)^2}},
$$
with equality only if $\phi_{i}(z)=0$.
Since $\phi_r$ has a unique maximum at zero we have $\phi_r(z)<\phi_r(0)\leq \kappa^2/(2\xi^2)$ and further
$
4\sqrt{(\kappa^2-2\phi_r(0)\xi^2)} \leq
\sqrt{2(\kappa^2-2\phi_r(u)\xi^2) +2\sqrt{(\kappa^2-2\phi_r(z)\xi^2)^2}},
$
with equality only if $z=0$.
Since $\phi_{i}(0)=0$ it follows that $u=0$ is the unique minimum of $\Re\left(\sqrt{\kappa^2-2\phi(\I z + a)\xi^2}\right)$.
Since $\phi_r$ is bounded away from $\phi_r(0)$ as $|z|$ tends to infinity there exists a $q^*>0$ and $M>0$ such that for $|z|>q^*$ we have that $\phi_r(z) \leq M < \phi_r(0)$. 
But then for $|z|>q^*$ we certainly have
$$
\Re\left(\sqrt{\kappa^2-2\phi(a)\xi^2}\right)=
4\sqrt{(\kappa^2-2\phi_r(0)\xi^2)} <
4\sqrt{(\kappa^2-2 M \xi^2)}  \leq
\Re\left(\sqrt{\kappa^2-2\phi(\I z + a)\xi^2}\right).
$$
This proves Assumption~\ref{assump:Differentiability}(v)(b).
The proof of Assumption~\ref{assump:Differentiability}(v)(c) involves tedious but straightforward computations and we only highlight the main steps.
Let $a\in\widehat{\mathcal{K}}_{\infty}^{0}$ and define $\overline{A}(u,\tau):=A(u,\tau)-\tau\widehat{V}(u)$ with $A$ given in~\eqref{eq:ABFellerTC}.
From the analysis above we know that the map $z\mapsto \Re{d(\phi(\I z + a))}$ has a unique minimum at $z=0$.
Also we recall that $0<d(\phi(a))$ and straightforward calculations show that $|\gamma(\phi(\I z + a))|<1$ with $d$ and $\gamma$ given in~\eqref{eq:dgammabetatimechange}.
Using the triangle and reverse triangle inequality we have for all $z\in\RR$ and $\tau>0$ that
\begin{equation}\label{eq:eq0unif}
\Re{\overline{A}(\phi(\I z +a),\tau)}=\frac{2\kappa\theta}{\xi^2 }\log\left|\frac{1-\gamma(\phi(\I z + a))}{1-\gamma(\phi(\I z + a))\E^{-d(\phi(\I z + a))\tau}}\right|
\leq \frac{2\kappa\theta}{\xi^2 }
\log \left(\frac{2}{1-\E^{-d(\phi(a))\tau}}\right).
\end{equation}
Tedious computations (see Figure~\ref{fig:RealPart} for a visual help) also reveal that ($B$ given in~\eqref{eq:ABFellerTC}):
$
\Re{B(\phi(\I z +a),\tau)}\leq B(\phi(a),\tau),
$
for all $z\in\RR$ and $\tau>0$.
Consider the second and third terms in~\eqref{eq:MGFFellerTC}. 
For all $z\in\RR$ and $\tau>0$:
\begin{equation}\label{eq:eq1unif}
\Re\log \left( \frac{1}{1-2\beta_t B(\phi(\I z +a),\tau) }\right)
=
\log \left|\frac{1}{1-2\beta_t B(\phi(\I z +a),\tau) }\right|
\leq
\log\left(\frac{1}{1-2\beta_t B(\phi(a),\tau) }\right),
\end{equation} 
where we note in the last inequality that $1-2\beta_t B(\phi(a),\tau)>0$.
We also compute 
$$
\Re \left( \frac{B(\phi(\I z +a),\tau)}{1-2\beta_t B(\phi(\I z +a),\tau) }\right)
=\frac{\Re B(\phi(\I z +a),\tau)-2\beta_t|B(\phi(\I z +a),\tau)|^2}{1-4\beta_t \Re B(\phi(\I z +a),\tau)+4 \beta_t^2 |B(\phi(\I z +a),\tau)|^2 },
$$
and hence using $\Re B(\phi(\I z+a),\tau)\leq|B(\phi(\I z +a),\tau)|$ we see that for all $z\in\RR$ and $\tau>0$:
\begin{equation}\label{eq:eq2unif}
\Re \left( \frac{B(\phi(\I z +a),\tau)}{1-2\beta_t B(\phi(\I z +a),\tau) }\right)
\leq
\frac{\Re B(\phi(\I z +a),\tau)}{1-2\beta_t \Re B(\phi(\I z +a),\tau) }
\leq
\frac{ B(\phi(a),\tau)}{1-2\beta_t  B(\phi(a),\tau) },
\end{equation}
where the last inequality follows since the term in the second inequality is strictly increasing in $\Re B(\phi(\I z +a),\tau)$.
Combining~\eqref{eq:eq0unif},~\eqref{eq:eq1unif} and~\eqref{eq:eq2unif} we see that as $\tau$ tends to infinity:
$$
\Re\left[ \tau^{-1}\log\mathbb{E}\left(\E^{(\I z + a) X_{\tau}^{(t)}}\right)-\widehat{V}(\I u + a)\right]\leq 
\left[\frac{\widehat{V}(a)v\E^{-\kappa t}}{1-2\beta_t \widehat{V}(a)}
+\frac{2\kappa\theta}{\xi^2}\log\left(\frac{2}{1-2\beta_t \widehat{V}(a)}\right)\right] \frac{1}{\tau}
+\mathcal{O}\left(\frac{1}{\tau^2}\right), \quad
\text{for all }
z\in\mathbb{R}.
$$
where the remainder does not depend on $z$.
This proves Assumption~\ref{assump:Differentiability}(v)(c).

\begin{figure}[h!tb] 
\centering
\mbox{\subfigure[$z\mapsto \Re B(\phi(\I z + a),\tau) $ .]{\includegraphics[scale=0.7]{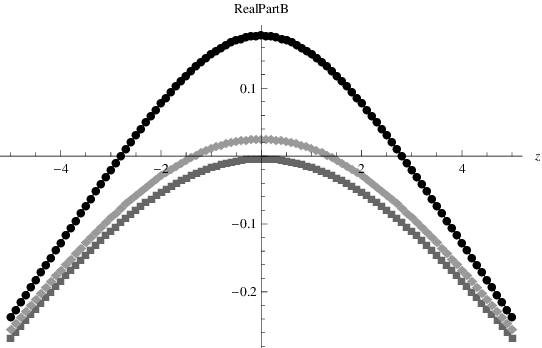}}\quad
\subfigure[Upper Bound]{\includegraphics[scale=0.7]{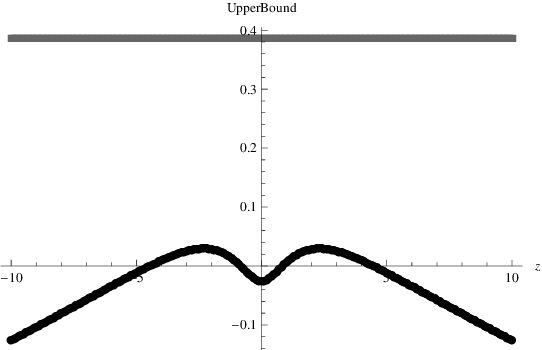}}}
\caption{The parameters are the same as Figure~\ref{fig:timechangedfeller}(a) with $t=1/2$ and $\tau=2$;
(a): plot of $z\mapsto \Re B(\phi(\I z + a),\tau) $ with $a=-3$ (circles), $a=0.5$ (squares) and $a=4$ (diamonds);
(b): plot of $z\mapsto \Re\left[\tau^{-1}\log\mathbb{E}\left(\E^{(\I z + 4) X_{\tau}^{(t)}}\right)-\widehat{V}(\I z +4 )\right] $ vs the upper bound derived in the appendix. }
\label{fig:RealPart}
\end{figure}


\end{document}